%% file: main.tex
\documentclass[a4paper, amsfonts, amssymb, amsmath, reprint, showkeys, nofootinbib,superscriptaddress,twoside, english,twocolumn, pra]{revtex4-2}
\usepackage[english]{babel}
\usepackage[utf8]{inputenc}
\usepackage[colorinlistoftodos, color=green!40, prependcaption]{todonotes}
\usepackage{lipsum}

\input{preamble}
\usepackage[pdftex, pdftitle={Article}, pdfauthor={Author}]{hyperref} 

\usepackage[normalem]{ulem}

\newtheorem{theorem}{Theorem}
\newtheorem{corollary}[theorem]{Corollary}
\newtheorem{lemma}[theorem]{Lemma}

\newcommand{\edit}[1]{{\color{red}#1 }}

\begin{document}
\title{Thermal state structure in the Tavis--Cummings model and rapid simulations in mesoscopic quantum ensembles}
    \author{Lane G. Gunderman}
    \email{lanegunderman@gmail.com}
   \affiliation{The Institute for Quantum Computing, University of Waterloo, Waterloo, Ontario, N2L 3G1, Canada}
    \affiliation{Department of Physics and Astronomy, University of Waterloo, Waterloo, Ontario, N2L 3G1, Canada}
    \affiliation{Current affiliation: University of Illinois Chicago, Department of Electrical and Computer Engineering, Chicago, Illinois, 60607}
    \author{Troy Borneman}
    \affiliation{The Institute for Quantum Computing, University of Waterloo, Waterloo, Ontario, N2L 3G1, Canada}
    \affiliation{High Q Technologies Inc, Waterloo, Ontario, N2L 3G1, Canada}
    \author{David G. Cory}
    \affiliation{The Institute for Quantum Computing, University of Waterloo, Waterloo, Ontario, N2L 3G1, Canada}
    \affiliation{Department of Chemistry, University of Waterloo, Waterloo, Ontario, N2L 3G1, Canada}    

\date{\today} 

\begin{abstract}
Hybrid quantum systems consisting of a collection of $N$ spin-$1/2$ particles uniformly interacting with an electromagnetic field, such as one confined in a cavity, are important for the development of quantum information processors and will be useful for metrology, as well as tests of collective behavior. Such systems are often modeled by the Tavis--Cummings model and having an accurate understanding of the thermal behaviors of this system is needed to understand the behavior of them in realistic environments. We quantitatively show in this work that the Dicke subspace approximation is at times invoked too readily, in specific we show that there is a temperature above which the degeneracies in the system become dominant and the Dicke subspace is minimally populated. This transition occurs at a lower temperature than priorly considered. When in such a temperature regime, the key constants of the motion are the total excitation count between the spin system and cavity and the collective angular momentum of the spin system. These enable perturbative expansions for thermal properties in terms of the energy shifts of dressed states, called Lamb shifts herein. They enable efficient numeric methods that scale in terms of the size of the spin system. Notably the runtime to obtain certain parameters of the system scales as $O(\sqrt{N})$, and is thus highly efficient. These provide methods for approximating, and bounding, properties of these systems as well as characterizing the dominant population regions, including under perturbative noise. In the regime of stronger spin-spin coupling the perturbations outweigh the expansion series terms and inefficient methods likely are needed to be employed, removing the computational efficiency of simulating such systems. The results in this work can also be used for related systems such as coupled-cavity arrays, cavity mediated coupling of collective spin ensembles, and collective spin systems.
\end{abstract}


\maketitle

\section{Introduction}\label{sec:intro}

In this paper, we focus on the collective interaction of a mesoscopic spin ensemble, with $N$ spins, with a single mode quantum electromagnetic field, henceforth referred to as a cavity. The case of a single spin, $N=1$, is called the Jaynes--Cummings model and has been extensively studied and yielded many insights into fundamental aspects of light-matter interactions \cite{shore1993jaynes,larson2021jaynes}. This system has been used to verify their quantum nature \cite{fink2008climbing} as well as being a useful tool for designing quantum devices, and so a more full understanding of the thermal properties of these systems in the case of larger $N$ could aid further progress in the development of quantum information processors \cite{blais_circuit_2020,fink2009dressed,yang_probing_2020,zou_implementation_2014}, hybrid quantum devices \cite{kurizki_quantum_2015,morton_hybrid_2011,xiang_hybrid_2013}, and radiative cooling of spin ensembles \cite{wood_cavity_2014,wood_cavity_2016,bienfait_controlling_2016,albanese_radiative_2020,ranjan_pulsed_2019}, among other possible uses. 

In our system we take each particle to be a non-interacting spin-1/2 particle, coupled uniformly to the electromagnetic field. After a rotating wave approximation, this is modeled by the Tavis-Cummings (TC) Hamiltonian,
    
    \begin{equation}
        \hat{H} = \omega_c \hat{a}^\dagger \hat{a} + \omega_s \hat{J}_z + g_0(\hat{a}\hat{J}_+ + \hat{a}^\dagger \hat{J}_-),\label{eq.tc}
    \end{equation}
    where $g_0$ is the single particle-mode coupling strength--we work in natural units throughout this work ($\hbar=k_b=1$) \cite{tavis1968exact}.
    
    This work builds off the techniques and results of a prior paper of a superset of the authors of this work \cite{gunderman2022lamb}. As shown in our prior work, we have the ability to efficiently compute, as well as provide meaningful statistics on, the energy level structure of the Tavis-Cummings Hamiltonian under a certain set of assumptions. While those results provided some information they failed to provide practical methodology for taking these results and testing them, this work proceeds to carry those results to experimental expectations as well as providing tools of use for other systems with similar structures. Most of the experimental assumptions required here can be reasonably realized in a laboratory setting, as we just require that the ensemble be on resonance with the cavity, $\omega_c = \omega_s \equiv \omega_0$, and that $g_0\sqrt{N} \ll \omega_0$ (albeit this expression is not quite correct as we argued) and lastly no spin-spin interactions. In our prior results we also did not include thermal effects, thus limiting the practical of the results. Under these assumptions, we have that angular momentum, $j$, and the number of total excitations, $k$, that is excitations stored in both the spin ensemble and the cavity, are conserved. This induces a two parameter family of non-interacting subspaces.
    
    
    This work is organized as follows. It begins with the definitions needed for our analysis. Following this we remark on the temperature regime of particular interest and show that so long as the system is sufficiently above absolute zero there will be nearly no population in the Dicke ($N/2$) subspace, nor in the lowest-excitation levels within the angular momentum subspaces, finding a smooth transition where the degeneracies of this system increasingly dominate. This tends to fall in the $10$ to $300$ mK range, although varies with the parameters of the system. \if{false}We then proceed to provide the uncoupled, $g_0=0$, partition function and average energy, which we then use to highlight the changes that occur from the introduction of the coupling between the spins and the cavity. Through this analysis we see that below a sufficiently small temperature the Dicke subspace is well populated, but beyond it the Dicke subspace is minimally populated and a perturbative approach must instead be used. We complement these mathematical results with numerical results which can show features that are overshadowed otherwise. These numerical results are achieved by implementing the improved computational methods from our prior work \cite{gunderman2022lamb}.\fi
    If the temperature of the system is low enough that the degeneracy does not dominate, more brute-force methods may be employed. If, however, the temperature is in a more degeneracy dominated regime, the introduced perturbative expansion may be used for computing the partition function and the Boltzmann averaging, which allows us to generate expressions for the shifts in the systems Helmholtz and average energies due to the coupling between the cavity and spin ensemble. Finally, through this perturbative expansion, we are able to generate optimal algorithms for computing many expectation values, as well as histograms for these properties--bringing the runtime down from $O(N^3)$ to $\Theta(\sqrt{N})$, both with a $O(T)$ dependence on temperature, for temperature $T$ \footnote{In this work $\Theta$ time means that the expression is upper-bounded and lower-bounded by the argument's time complexity.}. This allows for readily computing expectations for systems of $1000$ spins, although larger would be possible with careful programming.
    
    While we focus on the partition function, average energy of the system, and the Lamb-shift induced fluctuations of the photonic number operator as well as rotation about a collective axis, the techniques employed and the results shown can be used for other thermodynamic properties. Having access to both quick numerical methods for computing properties of these systems and perturbative expansions for these properties will allow additional insights into this model, as well as systems with similar structures such as coupled-cavity arrays \cite{hartmann2008quantum,nissen2012nonequilibrium,tomadin2010many}, collective spin ensembles with cavity mediated interactions \cite{astner2017coherent,norcia2018cavity}, multi-connected Jaynes-Cummings models and the extension of this work to the Tavis-Cummings version \cite{seo2015quantum,xue2017quantum}, as well as works that consider a collective spin system on its own \cite{wang2012quantum,munoz2023phase,omanakuttan2023spin}. The utility of this work for a myriad of physically motivated models with applications to a variety of quantum technologies suggests that this work will help with understanding and designing experiments for collective mesoscopic systems.
    

\if{false}
\subsection{Literature Review/Connections}

\begin{itemize}
    \item coupled-cavity arrays :
    \item pair of spin ensembles :
    \item TC of JC Lin Tian :
    \item General TC stuff \cite{schuster_high-cooperativity_2010}:
\end{itemize}
\fi
\section{Definitions} \label{sec:def}
    We include our notations and definitions in this section. The Pauli operators are written in the Zeeman basis. \if{false}, so that
    \begin{equation}
        \hat{\sigma}_z = \ketbra{\uparrow}{\uparrow} - \ketbra{\downarrow}{\downarrow},\quad 
        \hat{\sigma}_+ \ket{\downarrow} = \ket{\uparrow},\quad \hat{\sigma}_+ |\uparrow\rangle=0 .
    \end{equation}\fi
    These spin operators can be combined as a sum of tensor products to produce the collective versions of these spin operators. Let $N$ be the number of spin--$1/2$ particles. Then the collective spin operators are given by
    \begin{equation}
        \hat{J}_z = \frac{1}{2} \sum_{i=1}^N \hat{\sigma}_z^{(i)},\quad
        \hat{J}_\pm = \sum_{i=1}^N \hat{\sigma}_\pm^{(i)}.
    \end{equation}
    where the superscript on the Pauli operator indicates action only on the $i$-th particle.
    
    \if{false}
    The collective operator algebra is a sub--algebra of self--adjoint operators acting on the $N$--spin system, and conveniently satisfies the same commutation relations as those for a single particle spin operator. By a change of basis, we can identify the transverse spin operators,
    \begin{align}
        \hat{J}_x &= \frac{1}{2}\big(\hat{J}_+ + \hat{J}_-\big),\\
        \hat{J}_y &= \frac{1}{2i}\big(\hat{J}_+ - \hat{J}_-\big).
    \end{align}
    The transverse spin operators, along with $\hat{J}_z$, span a collective $\mathfrak{su}(2)$ algebra, which differ from a single spin--1/2 Pauli operators in that the collective operators are not involutory. The representations of the $\mathfrak{sl}(2;\C)$ operators can be defined by their action on a state of total angular momentum $j$ with $z$ component $m$:\fi
    The action of these operators on a state of total angular momentum $j$ with $z$ component $m$ are given by:
    \begin{align}
        \hat{J}_z \ket{j,m} &= m\ket{j,m}\\
        \hat{J}_\pm \ket{j,m} &= \sqrt{j(j+1) - m(m\pm 1)}\ket{j,m \pm 1}.
    \end{align}
    
    We note that our Hamiltonian has two good quantum numbers representing conserved quantities, so long as the system is cool enough that the rotating-wave approximation continues to hold. The first of these is the total angular momentum, $j$, which determines the eigenvalues of the total angular momentum operator, $\hat{\bm{J}}^2 = \hat{J}_x^2 + \hat{J}_y^2 + \hat{J}_z^2$, with eigenvalues $j(j+1)$. The second of these conserved quantities is the number of total excitations, $k$, given as the eigenvalues of the excitation operator,
    \begin{equation}
        \hat{K} = \hat{a}^\dagger\hat{a} + \hat{J}_z + \frac{N}{2}\openone.
    \end{equation}
    The scaled identity term in the excitation operator ensures excitations are non--negative, as the action of $\hat{J}_z$ on the ground state has eigenvalue $-N/2$.
    \if{false}
    All of the Hamiltonians considered in this work share the same internal structure, defined by 
    \begin{equation}
        \H_0 =  \omega_0 (\hat{a}^\dagger\hat{a} + \hat{J}_z).
    \end{equation}
    The interaction Hamiltonian is given by the Dicke model \cite{dicke1954coherence}:
    \begin{equation}\label{dicke}
        \H_{D,int} = 2g_0 \big(\hat{a} + \hat{a}^\dagger\big) \hat{J}_x.
    \end{equation}
    Applying the rotating wave approximation, the counter--rotating term is discarded, leaving us with
    \begin{equation}
        \H_{int} = g_0 \big(\hat{a} \hat{J}_+ + \hat{a}^\dagger \hat{J}_- \big),
    \end{equation}
    which is the interaction term in the Tavis--Cummings (TC) Hamiltonian. \if{false}Note that tensor products are omitted for brevity, so that, for example,
    \begin{equation}
        \hat{a} \otimes \hat{J}_+ \equiv \hat{a}\hat{J}_+.
    \end{equation}\fi

    Lastly, we note that all states live within the Hilbert space,
    \begin{equation}
        \mathscr{H} = \operatorname{L}^2\left(\mathbb{R}\right) \otimes \big(\C^2\big)^{\otimes N}.
    \end{equation}
    \fi

We recall next the pertinent aspects and definitions from our prior work. Our Hamiltonian may be written in direct sum form as
    \begin{equation}
        \hat{H} \cong \bigoplus_{j,k}\big( \omega_0 k \openone_{j,k} +  g_0 L(j,k)\big)^{\otimes d_j},
    \end{equation}
where $L(j,k)$ are hollow tridiagonal square matrices with known entries and dimension $|\mathcal{B}_{j,k}|=\min\{2j+1,k-k_0(j)+1\}$, with $k_0(j)=N/2-j$. The set of eigenvalues for $g_0 L(j,k)$ is written as $\Lambda(j,k)$, which are referred to as the \textit{Lamb shifts} for those values of $j$ and $k$. These correspond to the energy shifts induced from dressing the eigenstates of a given $j$ and $k$ value. \if{false} interaction Hamiltonian, which we define in equation \eqref{eq:coupling_matrix} and refer to as the coupling matrices. 
    
We define a basis for a general $(j,k)$ subspace with total angular momentum $j$ and $k$ excitations as 
    \begin{equation}\label{eq:lamb_basis}
        \mathcal{B}_{j,k} = \lbrace \ket{\alpha_{j,k}} \,|\, \alpha = 1,\cdots,n_{j,k}, n_{j,k} + 1\rbrace,
    \end{equation}
    using a shorthand ket representation of the tensor product of a spin-cavity state
    \begin{equation}
        \ket{\alpha_{j,k}} = \ket{k-\alpha-k_0(j)}\ket{j,-j+\alpha}.
    \end{equation}
    The single parameter, $\alpha$, provides a convenient representation of states within a $(j,k)$ subspace. The value, $n_{j,k} = \abs{ \mathcal{B}_{j,k} } - 1$, one less than the dimension, is chosen for convenience. We define $k_0(j) = N/2-j$ as the number of excitations present in the ground state of an angular momentum $j$ subspace within an $N$ spin ensemble. Explicitly, $n_{j,k}$ is given as
    \begin{equation}
        n_{j,k} = \operatorname{min}\lbrace 2j, k - k_0(j)\rbrace.
    \end{equation}
    If $k < k_0(j)$, then the basis set is empty and there are no states present at this excitation level within the $j$ angular momentum subspace.

    \if{false}The coupling matrices' entries can be found by applying the interaction term from $\mathcal{\hat{H}}_{TC}$ on the bases defined in equation \eqref{eq:lamb_basis}. \fi The Lamb shift coupling matrix for the $(j,k)$ subspace is then given by
    \begin{align}\label{eq:coupling_matrix}
        \nonumber L(j,k) &= \sum_{\alpha=1}^{n_{j,k}} l_{\alpha}(j,k)\bigg(\ketbra{\alpha_{j,k}}{(\alpha+1)_{j,k}}\\
        &\hspace{1.5cm}+ \ketbra{(\alpha+1)_{j,k}}{\alpha_{j,k}}\bigg),
    \end{align}
    where the matrix elements $l_{\alpha}(j,k)$ are given by
    \begin{align}
        \nonumber &\frac{1}{ g_0}\matrixel{\alpha_{j,k}}{\H_{int}}{(\alpha+1)_{j,k}}\\
        &= \sqrt{\big(2\alpha j - \alpha(\alpha-1)
        \big)\big(k-k_0(j)-\alpha+1\big)}.
    \end{align}
    In the above expression subscripts are only included within kets such as $\ket{\alpha_{j,k}}$, while $\alpha$ itself is a scalar.\fi Some properties of note for these $L(j,k)$ includes that all odd moments of the eigenvalues are zero, while the second moment has an analytic expression, and has expressions for bounds on the largest eigenvalue \cite{gunderman2022lamb}.
    
    The index $j$ runs from $N/2$ to $0$ ($1/2$) when $N$ is even (odd). Each angular momentum space is of dimension $2j+1$. The degeneracy of the subspace with total angular momentum $j$ on $N$ spins is given as
    \begin{equation}
        d_j = \frac{N! (2j+1)}{(N/2 - j)!(N/2+j+1)!}.
    \end{equation}
    That is, there are $d_j$ disjoint angular momentum subspaces with total angular momentum $j$ present in a direct sum decomposition of $\big(\C^2\big)^{\otimes N}$ \cite{wesenberg2002mixed}.
    
    The eigenvalues for one of these coupling matrices may be computed in $O(|\mathcal{B}_{j,k}|\log |\mathcal{B}_{j,k}|)$. \if{false}The eigenvectors take $O(|\mathcal{B}_{j,k}|^2)$.\fi Lastly, the degeneracies have a region of strong support of $O(\sqrt{N})$ about $j^*$, the most degenerate angular momentum subspace, where $j^*\approx \sqrt{N}/2$. This means that the vast majority of the degeneracies lay within this region. These observations aide in our analysis. To avoid particular opacity in the main text, we relegate rigorous arguments and computation details to the appendices.
    

\section{Temperature domain and $g_0=0$ case}

In this section we begin by providing the solutions to the uncoupled versions of the partition function and average energy, so that they may be used as references with which we compare the coupled TC model. Notably, the uncoupled case has an analytic solution and so there are no regimes to worry about, whereas the coupled case requires differing approaches depending on the regime. The movement between these regimes is shown, where one is degeneracy dominated and the other is Boltzmann limited, which provides the foundation for our further results as we focus on the rich degeneracy dominated regime.

\subsection{$g_0=0$, or uncoupled, case}

Throughout this work, we will compare our results with those which would be expected from an uncoupled cavity-spin ensemble system. The uncoupling can be due to $g_0=0$, or having a sufficiently large difference between the coefficients $\omega_s$ and $\omega_c$. We take the $g_0=0$ case so the expressions more closely match those we will obtain later in this work. In this case the reference Hamiltonian is given by $H_0=\omega_0(a^\dag a+J_z+\frac{N}{2}\openone)$ which has a partition function, as shown in appendix \ref{g0}, given by:
    \begin{equation}\label{z0}
        Z_0=(1-e^{-\beta_c\omega_0})^{-1}(1+e^{-\beta_s\omega_0})^N,
    \end{equation}
where $\beta_c=T_c^{-1}$ is the inverse temperature of the cavity, while $\beta_s=T_s^{-1}$ is the inverse temperature of the spin system. This has an average energy of
\begin{equation}\label{z0energy}
        \langle E\rangle_{0}=\omega_0 (e^{\beta_c \omega_0}-1)^{-1}+N\omega_0 \frac{e^{-\beta_s\omega_0}}{1+e^{-\beta_s\omega_0}}.
    \end{equation}

This analytical solution holds regardless of the temperature of each of the uncoupled subsystems and could even permit field inhomogeneities for the spins. These equations themselves are of minimal interest on their own, but can help illuminate the differences in the system upon introducing the coupling Hamiltonian and considering the total Tavis--Cummings model. As no analytical solution like equation (\ref{z0}) is known for the Tavis--Cummings model, in the next subsection we consider the various temperature regimes that can be of interest in the Tavis--Cummings model and specify which regime we focus on in the remainder of this work.

\subsection{Low $T$, minimal Dicke}

Of all the angular momentum subspaces, quite possibly the easiest one to say much about is the Dicke space, which is the completely symmetric subspace. This subspace has particularly nice eigenstates and has been studied extensively \cite{garraway2011dicke,holstein1940field,roses2020dicke}. However, this subspace is only singularly degenerate, $d_{N/2}=1$, while for larger $N$ all other angular momentum subspaces have significantly greater degeneracies. Even though the Dicke space contains the unique ground state with $k=0$ excitations there is a temperature at which the exponential suppression due to the Boltzmann factor is no longer significant compared to the vast multiplicity of other angular momentum subspaces and the Dicke space is left with minimal population. This regime shift temperature is, as we shall argue, a small temperature, meaning that if one is not below this value, it is likely not correct to describe the system as being within the Dicke space. For typical parameters this occurs for temperatures between $10$ mK and $300$ mK for small systems of $10^{3}$ spins, above which the Dicke population is minimal.

As a short, simplified argument for the expression for this regime change temperature we consider the following reduced picture. Let $k\leq 1$ and let the system have $N$ spins. In this case only the first two levels in the Dicke space are populated and the $j=\frac{N}{2}-1$ angular momentum space's lowest energy state is populated. Figure \ref{subcut} shows this subspace of the full spectrum. The partition function due to the Dicke subspace is given by $1+e^{-\beta(\omega_0-g_0\sqrt{N})}+e^{-\beta(\omega_0+g_0\sqrt{N})}$, while that due to $j=\frac{N}{2}-1$ is $(N-1)e^{-\beta\omega_0}$. Then
\begin{widetext}
\begin{equation}
    p(Dicke)=\frac{1+e^{-\beta(\omega_0-g_0\sqrt{N})}+e^{-\beta(\omega_0+g_0\sqrt{N})}}{1+e^{-\beta(\omega_0-g_0\sqrt{N})}+e^{-\beta(\omega_0+g_0\sqrt{N})}+(N-1)e^{-\beta\omega_0}}.
\end{equation}
\end{widetext}

\begin{figure}[t]
    \centering
    \includegraphics[scale=0.3]{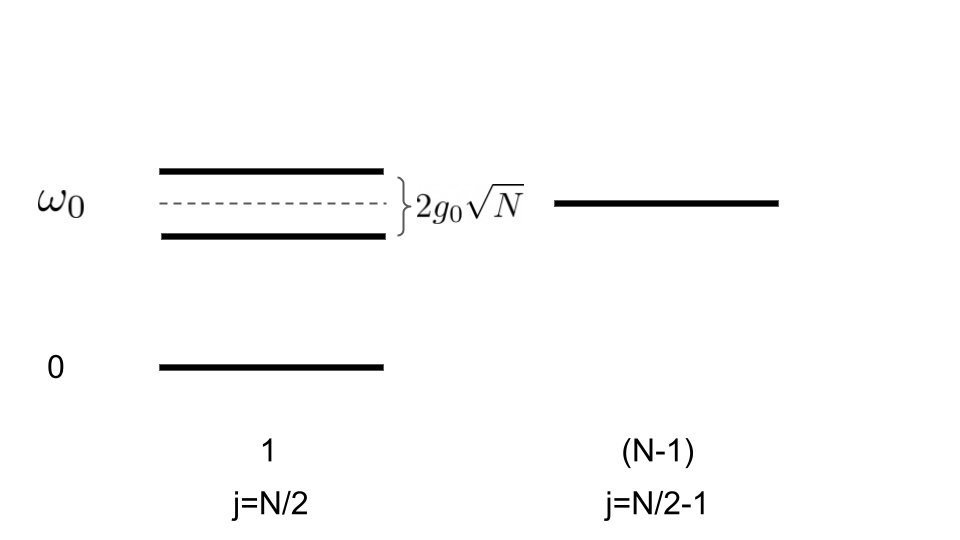}\label{subsection}
    \caption{This shows the portion of the energy spectrum considered for the regime change argument.}
    \label{subcut}
\end{figure}

\if{false}
\begin{equation}
    p(Dicke)< \frac{Z_{Dicke}}{Z_{Dicke}+(N-1)e^{-\beta\omega_0}}
\end{equation}
We will truncate $Z_{Dicke}$ as $1+2e^{-\beta\omega_0}\cosh(\beta g_0\sqrt{N})$, where we have assumed that $e^{-2\beta\omega_0}\ll 1$ and that due to the rotating-wave approximation the next set of Lamb shifts are far smaller than $\omega_0$. Then the probability of being within the Dicke (totally symmetric) subspace is bounded by:
\begin{equation}
    \frac{1+2e^{-\beta\omega_0}\cosh(\beta g_0\sqrt{N})}{1+e^{-\beta\omega_0}(N-1+2\cosh(\beta g_0\sqrt{N}))}
\end{equation}
\fi
This is only close to $1$ when $1+e^{-\beta\omega_0}2\cosh(\beta g_0\sqrt{N})> (N-1)e^{-\beta\omega_0}$. Since $g_0\sqrt{N}\ll \omega_0$ due to the rotating-wave approximation, we may approximate this requirement with $e^{\beta\omega_0}> (N-1)$, which provides $T< \omega_0/\log(N-1)$, with $\log$ being the natural log function throughout this work. Since $N$ is typically large we will take this as 
\begin{equation}
T< \omega_0/\log N.
\end{equation}
If further subspaces are included in the denominator this value can only decrease, however, as we show in appendix \ref{regimes}, this asymptotic expression is still correct.

\begin{center}
\begin{table}
\begin{tabular}{||c| c||} 
 \hline
  & $N_c$  \\ [0.5ex] 
 \hline\hline
 $T=10$ mK & $6.96\cdot 10^{20}$  \\ 
 \hline
 $T=100$ mK & $121$  \\ 
 \hline
 $T=300$ mK & $4$  \\ 
 \hline
 $T=1$ K & $1$ \\
 \hline
\end{tabular}
\caption{$N_c$ is the value of $N$ such that $N<e^{\hbar\omega_0/(k_b T)}$, so that mostly the Dicke subspace is populated. The above table uses $\omega_0=20\pi $ GHz. It means that for $T$ at and above $1$ K the population will not be predominantly in the Dicke space, regardless of the number of spins in the system. When $T=100$ mK, once the ensemble involves around $120$ spins the system is minimally in the Dicke subspace, meaning that for mesoscopic systems population in the Dicke subspace will be minimal.}
\end{table}
\end{center}

The rest of this paper focuses on the regime where $T>\omega_0/\log N$, but the system is still within the rotating-wave approximation. See Figure \ref{dosfig} for a pictorial representation of this regime. In this higher temperature regime rich structure and dynamics are exhibited, which will be explored in detail throughout this work. While the consideration of spin ensembles in the "high temperature approximation" is reminiscent of this result, those results sometimes use simplifications not appropriate for our system. For instance in Slichter, it is assumed that the interactions between the spins is effectively \textit{local}, while here we inherently require the system to be \textit{collective} in its behavior \cite{slichter2013principles}.

\begin{figure}[t]
    \centering
    \includegraphics[scale=0.25]{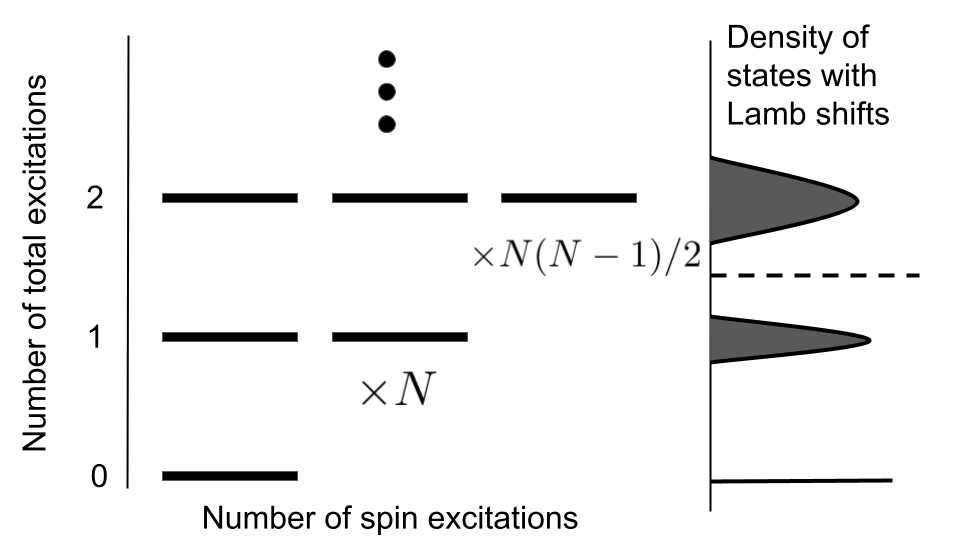}
    \caption{This figure shows schematically what the density of states looks like for an ensemble of spins coupled with a cavity. The rotating-wave approximation holds so long as the envelope of states for two different total number of excitations do not overlap.}
    \label{dosfig}
\end{figure}

\subsection{Higher $T$}

Next, we consider $T>\omega_0/\log N$ and what happens as $T$ is slowly increased. The analytic continuation of the function for the degeneracy of the angular momentum subspaces is a unimodal function with a peak at $j^*=\sqrt{N}/2-1/2$ and with the vast majority of the probability mass between $0$ and $O(\sqrt{N})$. To gain insights into the behavior of the system as the temperature increases, let us consider the population ratio of adjacent angular momentum subspaces. In particular, for some temperature $T$ we consider the summation over $k$, resulting in a population ratio of:
\begin{eqnarray}
    \frac{p(j)}{p(j+1)}&=&\frac{d_j(e^{-\beta \omega_0}+2e^{-2\beta\omega_0}+3e^{-3\beta\omega_0}+\cdots)}{d_{j+1}(1+2e^{-\beta\omega_0}+3e^{-2\beta\omega_0}+\cdots)}\\
    &\approx &\frac{d_je^{-\beta\omega_0}}{d_{j+1}},
\end{eqnarray}
where for this consideration we use that the Lamb shift splittings are small, $g_0 \max \Lambda (j,k)\ll \omega_0$ so that the rotating-wave approximation holds, and so neglect them here. While one of the above geometric series terminates a term earlier than the other, the size of that term is near zero as $xe^{-ax}$ tends to zero rapidly. More generally the ratio of the populations in angular momentum space $j$ and the Dicke space is given, up to the negligible geometric series terms, by
\begin{equation}
    \frac{p(j)}{p(N/2)}=d_je^{-\beta\omega_0(N/2-j)},
\end{equation}
which means that this subspace has more population when  $d_je^{-\beta\omega_0(N/2-j)}>1$, which is when $T>\omega_0 (N/2-j)/\log d_j$. This expression is a monotonic function in $j$ which very slowly increases to a maximal value of $\omega_0/(2\log 2)$, as shown in appendix \ref{ratio}. So this regime change temperature in the prior subsection is only a value whereby the population in the Dicke subspace is small, but it is not sufficient to say that further angular momentum subspaces are predominantly populated. To continue this result further, consider the ratio of the populations in adjacent angular momentum values once again, which provides:

\begin{equation}
    \frac{d_je^{-\beta\omega_0}}{d_{j+1}}>1\Rightarrow T> \omega_0/\log (d_{j+1}/d_j).
\end{equation}
The denominator changes sign for $j<j^*$, meaning that there is no temperature where these angular momentum values have a greater probability. However, for $j>j^*$ there will be more population in the smaller angular momentum space once the temperature is sufficiently large, then the width of $j^*$ to $O(\sqrt{N})$ ought to be considered to capture the dominant population. The smallness in the ratio is not sharp though, meaning that a sliding collection of angular momentum values must be considered. This is illustrated in Figure \ref{ROSSfigure}. Formally the rectangle of what must be considered slides to the left (towards lower $j$ values) as the temperature increases, however, since the region left of $j^*$ is small, the entire thing may be considered without an increase in computational time. Since the ratio is not sufficiently small for most temperatures, we do not know that the subset of angular momentum values which must be considered is asymptotically smaller than $N$. While on the one hand taking these sums over $k$ would reduce the temperature dependence in the runtime to simulate these systems, the dependency on $N$ is not particularly good, especially once the Lamb shifts are considered. In the next section we consider the impact of including the Lamb shifts on the partition function as well as other expectations.

\begin{figure}[t]
    \centering
    \includegraphics[scale=0.30]{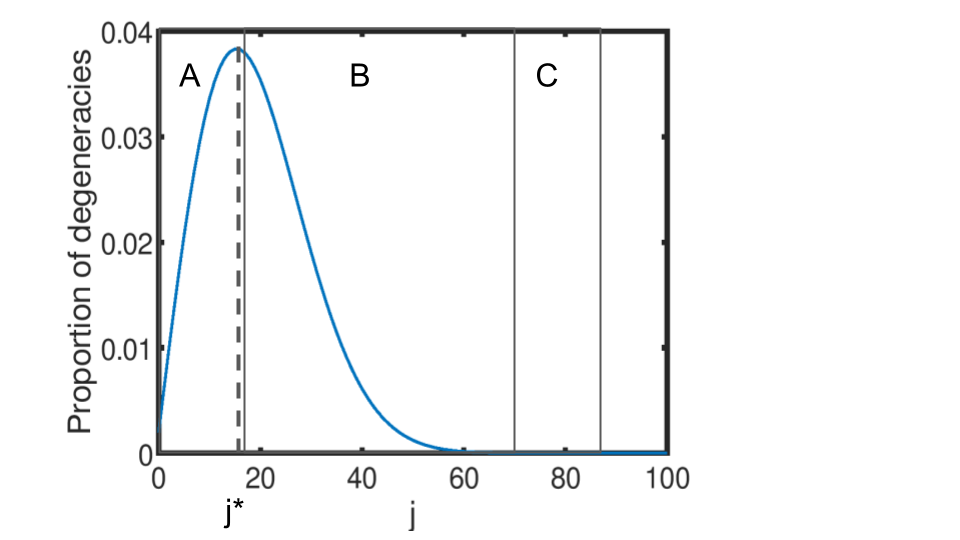}
    \caption{For the temperature regime considered, the boxes labelled $B$ and $C$ (right of the $j^*$ line) are the only subspaces with majority population so long as $T<\omega_0/(2\log 2)$. Above this temperature the boxes labelled $A$ and $B$'s angular momentum subspaces must be included. The width is $2\sqrt{N}$ here for both $A+B$ and $B+C$. Plotted for $N=1000$, zoomed in to the $j\leq 100$ subset.}
    \label{ROSSfigure}
\end{figure}



\section{Perturbative Expansions and shifts in expectations}

In the prior section we neglected the effects of the Lamb shifts as they were considered small in comparison to the excitation energy splitting of $\omega_0$. We now include the effects of the Lamb shifts generating a perturbative expansion that can be carried to the desired precision. This provides a more tractable way to describe the system and allows for easier qualitative understanding of it, as well as leading to a faster simulation runtime.

As the system involves many spins, we take the system as an ensemble, using a density matrix and obeying the canonical Boltzmann distribution for populations. In this statistical mechanics setting, knowing the partition function, and its rate of change with respect to parameters, is of central importance as it provides many experimentally accessible values. For the Tavis--Cummings model the partition function is given by:

\begin{equation}
        Z_{total}(\beta) = \sum_{k=0}^\infty \sum_{j=j_0(k)}^{N/2}\bigg( d_j e^{-\beta k \omega_0} \sum_{\lambda\in\Lambda(j,k)} e^{-\beta \lambda g_0}\bigg),
    \end{equation}
    where we have defined
    \begin{equation}
        j_0(k) = \max \lbrace N/2-k, \ N/2 - \floor{N/2}\rbrace.
    \end{equation}
    Upon Taylor expanding and using our results on the parity of eigenvalues we may write  $Z_{total}(\beta)=Z_0(\beta)+Z_{pert}(\beta)$. In this expression $Z_0(\beta)$ is the partition function of the system with $g_0=0$, which is exactly solved in equation (\ref{z0}). $Z_{pert}(\beta)$ represents the perturbative expansion portion of the total partition function. Throughout the work we will truncate $Z_{pert}(\beta)$ to only involve the leading term, however, further terms may be included and their structure is discussed in appendix \ref{pertexp}. In this expansion we have:
\begin{equation}
        Z_0(\beta)=\sum_k e^{-\beta k\omega}\sum_j d_j|\mathcal{B}_{j,k}|,
    \end{equation}
    which is the same as equation (\ref{z0}), and
    \begin{equation}
        Z_{pert}(\beta)= \frac{(\beta g_0)^2}{2}\sum_k e^{-\beta k\omega}\sum_j d_j|\mathcal{B}_{j,k}| Var(\Lambda(j,k)),
    \end{equation}
    where we have dropped the higher order terms\footnote{By definition, $|\mathcal{B}_{j,k}| Var(\Lambda(j,k))=tr(L(j,k)^2)$}. While $Z_0$ corresponds to the case of $g_0=0$ it can also be interpreted as the spin ensemble undergoing excitation swaps with the cavity zero times, while $Z_{pert}$ is induced from a double-quantum transition (transferring to the cavity and back, or vice versa) with the next term being a quadruple-quantum transition, and so forth. In effect these perturbative expansion terms provide the correlations induced in the system from the coupling. The strength of these higher order transitions can be neglected compared to the leading perturbative term so long as $\frac{1}{2}(\beta g_0)^2Var(\Lambda(j,k))<1$ for all $j,k$ with appreciable population, where $Var(\Lambda(j,k))=\frac{1}{|\mathcal{B}_{j,k}|}\sum_{\lambda\in \Lambda(j,k)}\lambda^2$ and has an analytic expression \cite{gunderman2022lamb}. This can be used as a guide for the size of the spin ensemble needed in order to see the effects of the coupling with the cavity, or to know that this perturbative expansion may be truncated as early as it is presented here. 
    
    This expansion for the partition function can be used to determine the shifts induced in certain expectations. In principle any expression of the partition function may be computed, but to illustrate the utility, let us consider the shift in the Helmholtz free energy of the system. \if{false}This is one of the few macroscopic thermodynamic properties that may be computed directly from the quantum model, without specification of other variables.\fi For a closed system at thermal equilibrium we have:
    \begin{eqnarray}
        \nonumber -\beta \Delta\langle A\rangle&:=& -\beta(\langle A\rangle_{Z_{total}}-\langle A\rangle_{Z_0})\\
        \nonumber &=& \log(Z_0+Z_{pert})-\log Z_0\\
        \nonumber &=& (\log Z_0+\log(1+Z_{pert}/Z_0))-\log Z_0\\
        &= & Z_{pert}/Z_0+O((Z_{pert}/Z_0)^2).\label{helmz}
    \end{eqnarray}
    
As this is a negative value this means that the introduction of the coupling has slightly reduced the amount of free energy in the system, which agrees with intuition as for each excitation a larger portion of the population will reside within the dressed states with the most negative Lamb shift values. A plot of $-\beta\Delta\langle A\rangle$, to first order, is shown in Figure \ref{zratio}, numerically computed as a function of $N$. Over one order of magnitude the ratio is a linear function, so assuming a similar linear trend continues for larger $N$ a system with around $10^{16}$ spins results in a ratio of $1/100$, leading to an experimentally noticeable shift. For differing temperatures a similar trend occurs, with varying slopes.
\if{false}
    \begin{figure}[t]
    \centering
    \includegraphics[scale=0.45]{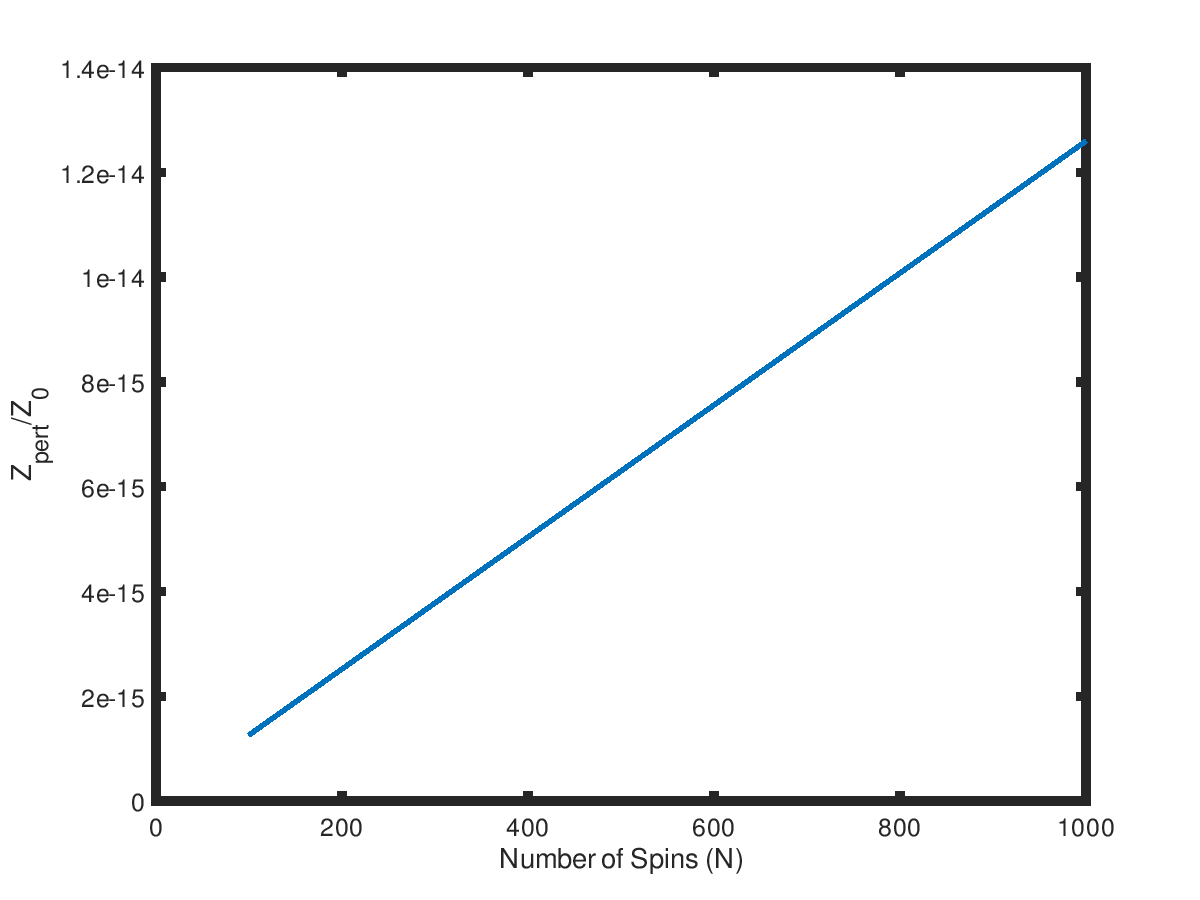}
    \caption{Figure of $Z_{pert}/Z_0$ as a function of $N$ with $N$ in $100$ to $1000$, $T=0.3$ K, $g_0=100$ Hz, $\omega_0=10$ GHz.}
    \label{zratio}
\end{figure}
\fi
    \begin{figure}[t]
    \centering
    \includegraphics[scale=0.40]{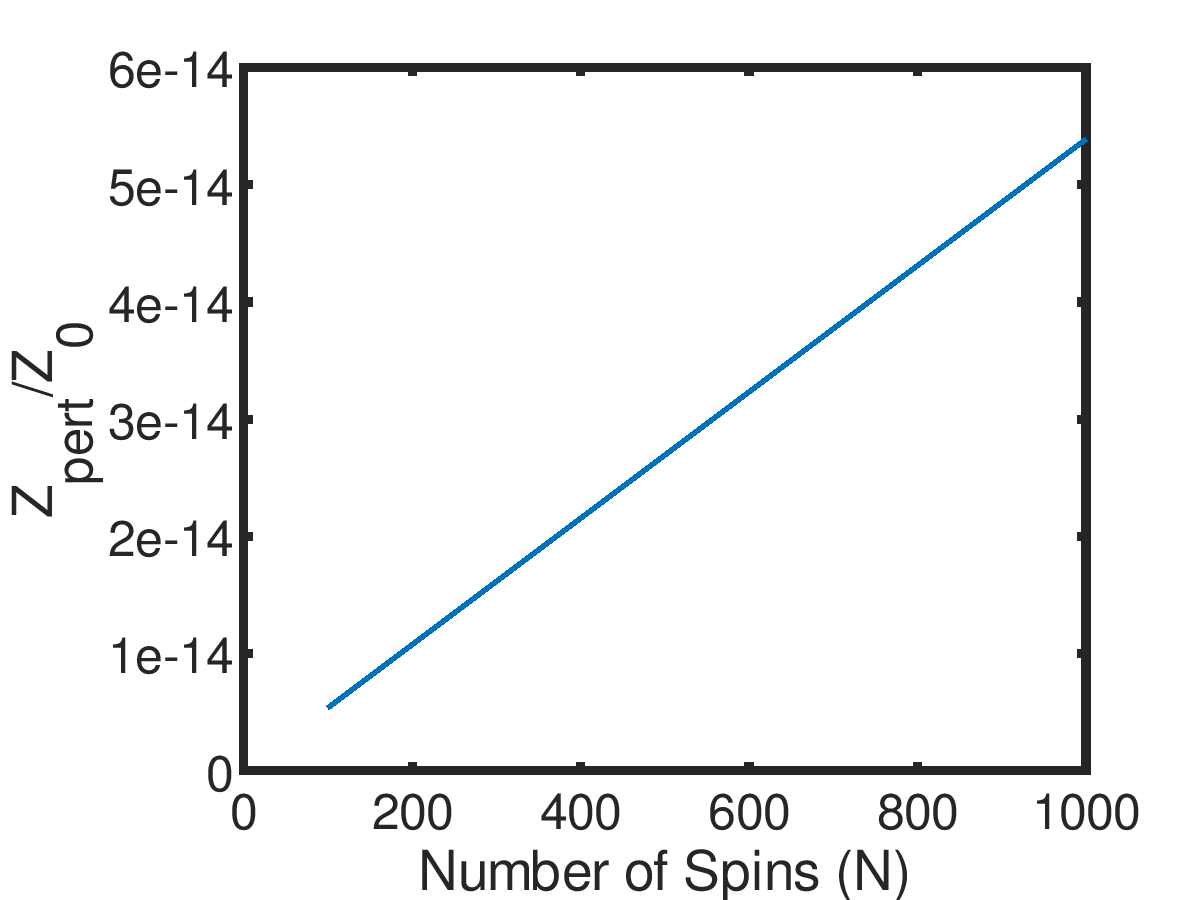}
    \caption{Figure of $Z_{pert}/Z_0$ as a function of $N$ with $N$ in $100$ to $1000$, $T=0.3$ K, $g_0=200\pi $ Hz, $\omega_0=20\pi$ GHz. Equation (\ref{helmz}) states this is proportional to the first order term in the shift of the Helmholtz free energy due to the coupling: $-\beta\Delta\langle A\rangle$.}
    \label{zratio}
\end{figure}

Performing a similar analysis, see appendix \ref{epert}, one finds that the average energy shift in the system is given by:
\begin{eqnarray}
    \Delta \langle E\rangle &:=& \langle E\rangle_{Z_{total}}-\langle E\rangle_{Z_0}\\
    &\approx &((\langle k\rangle_{Z_{pert}}-\langle k\rangle_{Z_0}) \omega_0-\frac{2}{\beta})\frac{Z_{pert}}{Z_0}.
\end{eqnarray}
In the above $\langle k\rangle_{Z_{pert}}$ is the expectation of $k$ over the normalized distribution for $Z_{pert}(k)$, and $\langle k\rangle_{Z_0}$ is the same but for $Z_0(k)$, the same as equation (\ref{z0energy}). This expression for the shift in average energy is also always negative, again matching intuition that having lower-energy dressed states to populate means that the average energy of the system will be decreased.

While these analytical results are of some interest and use, they are somewhat limited in practical utility. The primary problem is that they still require evaluation of summations to obtain values. In the next section we show that this perturbative expansion, paired with the reduced subset of $j$ values needing to be considered, may be used to provide a drastic speedup in simulation runtimes of this system.

\section{Rapid Simulation of Thermal Observables}

In the study of mesoscopic systems it is rare to obtain analytical results for physically meaningful systems due to the complexity of such systems. Simulating mesoscopic systems often requires a trade-off between precision and speed, at times using heuristics to bound the accuracy of such methods. Using the tools shown here, we may simulate the Tavis--Cummings model with arbitrarily small error and in optimal runtime dependence on the size of the ensemble. While this work only provides the speedup for the Tavis--Cummings model, a similar methodology might be useful for other systems exhibiting collective behaviors such as a pair of spin ensembles interacting via mediating fields.

Before diving into the new results, let us consider the naive approach for computing thermal functions. Assuming one uses the decomposition of the Hamiltonian as a direct sum of spaces labelled by good quantum numbers $j$ and $k$, this task would require finding the eigenvalues of the terms in the direct sum. These are, up to, $(N+1)\times (N+1)$ matrices, which takes time $O(N^2)$ to solve and there are $O(N)$ of them for each value of $k$, lastly these eigenvalues must be found for each $k$ of importance, for which there are $O(T)$ of them. Putting these together the total complexity is $O(TN^3)$. A slight improvement can be obtained by using the faster eigenvalue solutions for tridiagonal matrices, but this only reduces the problem to $O(TN^2\log N$) generally \cite{coakley2013fast}.

To illustrate our reduction we will consider the computation of the partition function's two components $Z_0$ and $Z_{pert}$ for some chosen temperature. $Z_0$ takes $O(1)$ time since it has an analytic expression, so the focus will be on $Z_{pert}$. To compute $Z_{pert}$ we must perform the sums in:
\begin{widetext}
\begin{equation}
\sum_k e^{-\beta k\omega_0 } \sum_j d_j \left[
    \frac{(\beta g_0)^2}{2}\sum_{\lambda\in \Lambda(j,k)} \lambda^2+\frac{(\beta g_0)^4}{4!} \sum_{\lambda\in \Lambda(j,k)} \lambda^4+\ldots\right].
\end{equation}
\end{widetext}

While one could solve for the eigenvalues $\lambda$ in the above by finding the roots of the characteristic polynomial for each $L(j,k)$, this can be circumvented by computing traces. In general the cost of computing the trace is $O(N)$ as the coupling matrices may be as large as $(N/2)\times (N/2)$. If, however, we truncate after the first term, the time to compute $\sum \lambda^2$ is only $O(1)$ since we have an analytic expression for this in terms of $j$, $k$, and $N$. This is a valid approximation so long as $\frac{(\beta g_0)^2}{2}\sum \lambda^2< 1$. This removes two full powers of $N$ dependence.

Following this, due to the temperature regime that we are in only a small subset of $j$ values contribute much to the summation. From the region of strong support for the degeneracies, we see that $\Theta(\sqrt{N})$ angular momentum values must be considered assuming a selected fractional error $\delta$ is permitted, as shown in appendix \ref{ndepen}--this collection of $j$ values which must be considered are illustrated in Figure \ref{jvals}. Lastly, the outermost summation is bounded by a geometric series with rate $\beta\omega_0$, so using $O(T)$ terms suffices to capture the majority of the probability density. While $O(T)$ suffices, this bound is of somewhat limited use since $T$ will eventually become large enough that the rotating-wave approximation breaks down for a notable number of populated states. In total, however, this gives a runtime of $O(T\sqrt{N})$. Note that this region of strong support is known to be tightly $\Theta(\sqrt{N})$ meaning that this runtime is \textit{optimal} scaling in the parameter of $N$. An important caveat to this discussion is that this result only applies for the degeneracy dominated regime, such that $O(\sqrt{N})$ excitations exist in the system, below that it is best to apply brute-force methods due to the comparative smallness of the problem.

If one wished to include any further perturbative terms, the trace must be computed, which at worst will take $\Theta(N)$ time since there are values of $j$ and $k$ with that dimension. The subset of $j$ values that must be considered still remains $\Theta(\sqrt{N})$. This generates a dependency of $\Theta(N^{3/2})$ for including additional terms. If one wished to compute the Lamb shifts themselves, or achieve machine precision, this raises the cost to $O(N^{3/2}\log N)$, which is a modest increase for the additional information gained.

    \begin{figure}[th]
    \centering
    \includegraphics[scale=0.25]{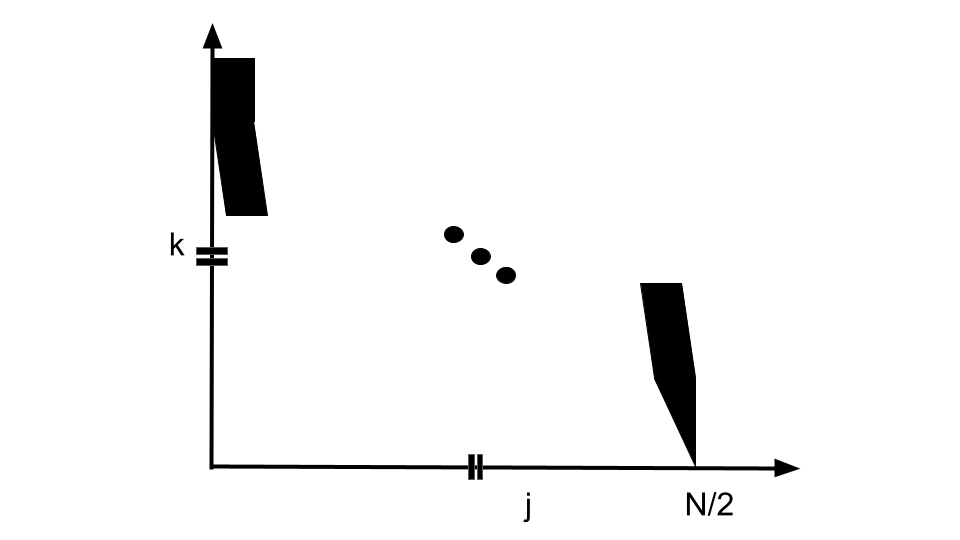}
    \caption{Schematic figure showing $j$ values that need summing over for any given $k$ choice. The width of each shape is $\Theta (\sqrt{N})$, aside from the tapered tip of the triangle, on the lower right portion of the diagram, for which there are insufficiently many angular momentum values available, so all of them are used there.}
    \label{jvals}
\end{figure}

While the above showed the rapidity of computing a scalar, the partition function, our reductions also permit the rapid simulation of some time evolving operators as well. For an observable $\mathcal{O}$ the structure of the summation, with the good quantum numbers $j$ and $k$, is preserved so long as the operator only depends on $j$ and $k$, as well as some other parameter that can be summed out. Any function of the cavity and collective spin operators, which includes most observables of interest, satisfy this requirement. While these may be simulated using the same summation, the same rapidity cannot be promised without one more requirement, which is that the operator $\mathcal{O}$ is a subexponential function in $j$ and $k$. \if{false}--some examples of subexponentially growing functions include $J_x$, $e^{itJ_z}$, and $e^{(a^\dag a^\dag- aa) it}$.\fi When $\mathcal{O}$ is subexponential the region of strong support in the degeneracies and the geometric series bound on excitations can be applied, otherwise these break down and the full summation must be considered.

Naturally, if one wanted to generate a histogram of the populations as a function of $k$ or $j$ the respective summation can be suspended and instead used to generate the points. This is used to show the populations of the terms in the partition function for various temperatures in Figure \ref{zgang}. In this figure we see that the population's center drastically shifts as we move into the degeneracy dominated regime ($T=0.3K$), additionally the width of the distribution increases, as expected.

A summary of the algorithmic improvement results are presented in the table below:

\begin{widetext}

     \begin{center}
     \begin{table}[thb]\label{speed}
\begin{tabular}{||c| c| c| c| c||} 
 \hline
 Error & $O(g_0^4)$ & $O(g_0^{2c}),\ c> 2$ & Machine Precision & Naive \\ [0.5ex] 
 \hline\hline
 $Z_{pert}$, $Z_{tot}$, $\Delta\langle A\rangle$, $\Delta \langle E\rangle$ & $\Theta(\sqrt{N})$ & $\Theta(N^{3/2})$ & $O(N^{3/2}\log N)$ & $O(N^3)$ \\ 
 \hline
 $\mathcal{O}(j,k)$ (no region of strong support promise) & $O(N)$ & $O(N^2)$ & $O(N^2\log N)$ & $O(N^3)$ \\
 \hline
 $\mathcal{O}(j,k)$ (region of strong support promised) & $\Theta(\sqrt{N})$ & $\Theta(N^{3/2})$ & $O(N^{3/2}\log N)$ & $O(N^3)$ \\ [1ex]
 \hline
\end{tabular}
\caption{The above table lists the different run times needed for a chosen level of precision, as a function of $g_0$. The quickest, nontrivial level of precision is markedly faster. For comparison, the time needed in the naive approach is listed. In order to utilize the region of strong support, the observable $\mathcal{O}$ must be a subexponential function in $j$ and $k$.}
\end{table}
\end{center}

\if{false}
    \begin{figure}[h]
    \centering
    \includegraphics[scale=1]{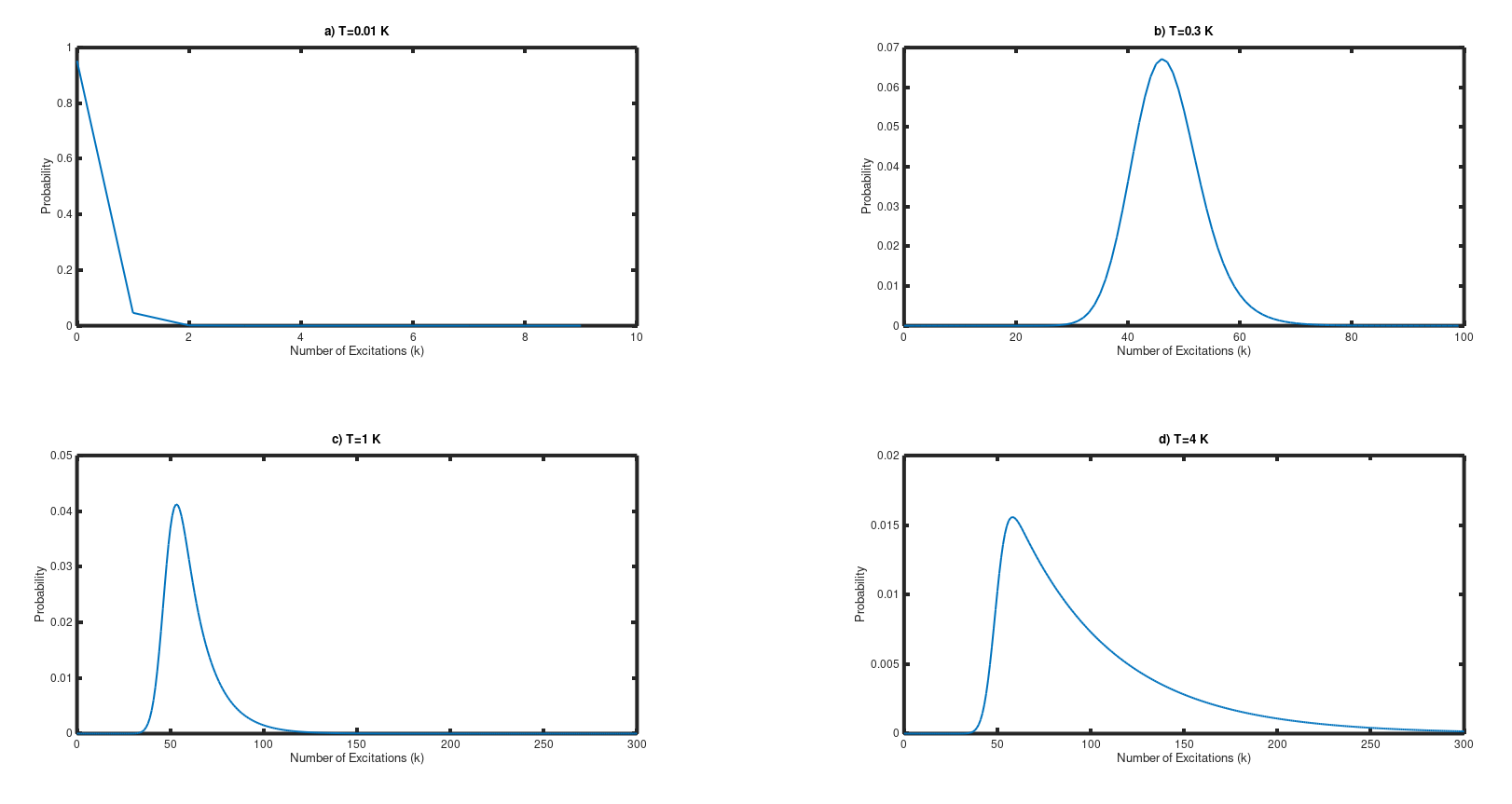}
    \caption{Figure of the distribution of the population for $Z_0$ as a function of the number of excitations in the system ($k$). These are plotted for $N=100$, $\omega_0=10$ GHz.}
    \label{zgang}
\end{figure}
\fi
    \begin{figure}[h]
    \centering
    \includegraphics[scale=0.8]{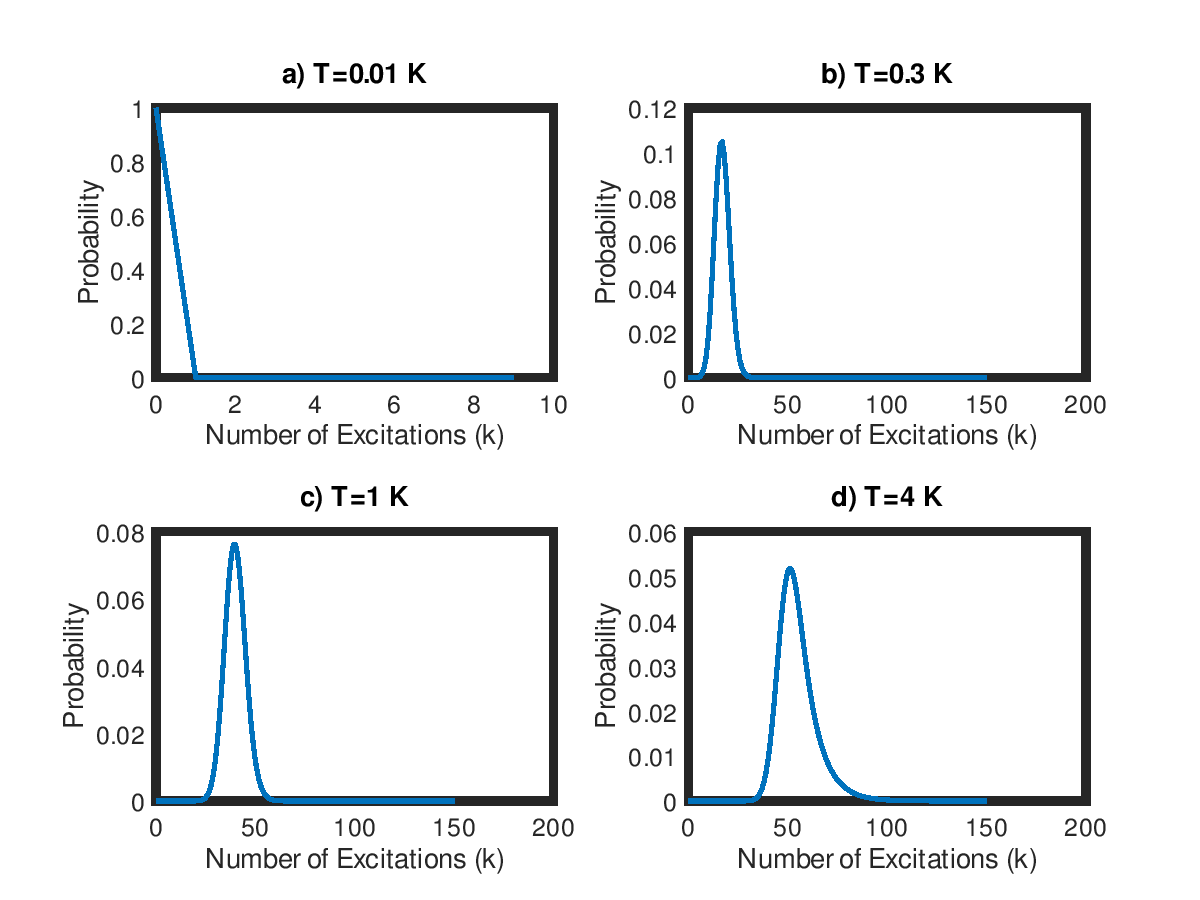}
    \caption{Figure of the distribution of the population for $Z_0$ as a function of the number of excitations in the system ($k$). These are plotted for $N=100$, $\omega_0=20\pi$ GHz.}
    \label{zgang}
\end{figure}

\if{false}
    \begin{figure}[h]
    \centering
    \includegraphics[scale=0.9]{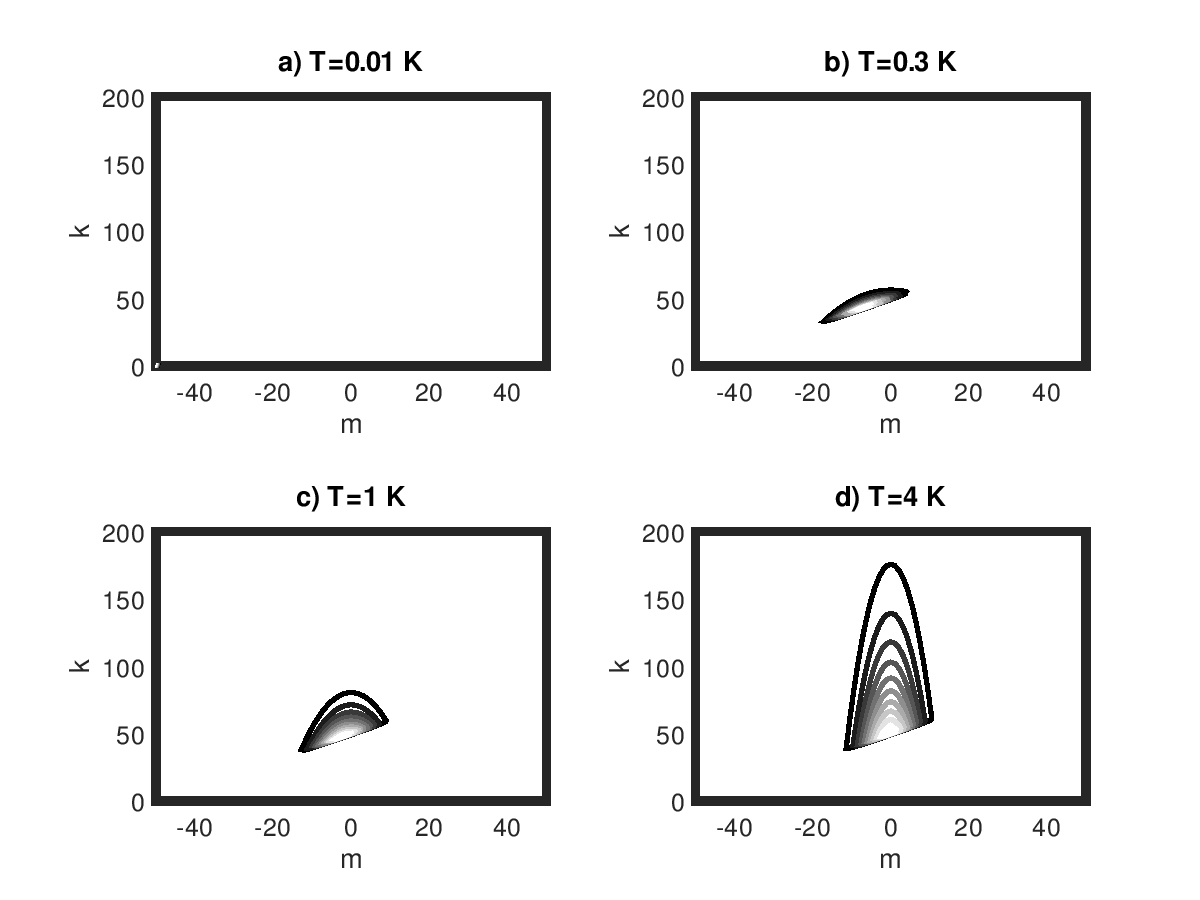}
    \caption{Figure of the distribution of the population for $Z_{total}$ as a function of the number of excitations in the system ($k$) against the secondary spin quantum number ($m$). These are plotted for $N=100$, $g_0=100$ Hz, and $\omega_0=10$ GHz. Note that for $T=0.01$ Kelvin the population is essentially in the Dicke subspace and global groundstate.}
    \label{distprobs}
\end{figure}
\fi

    \begin{figure}[h]
    \centering
    \includegraphics[scale=0.9]{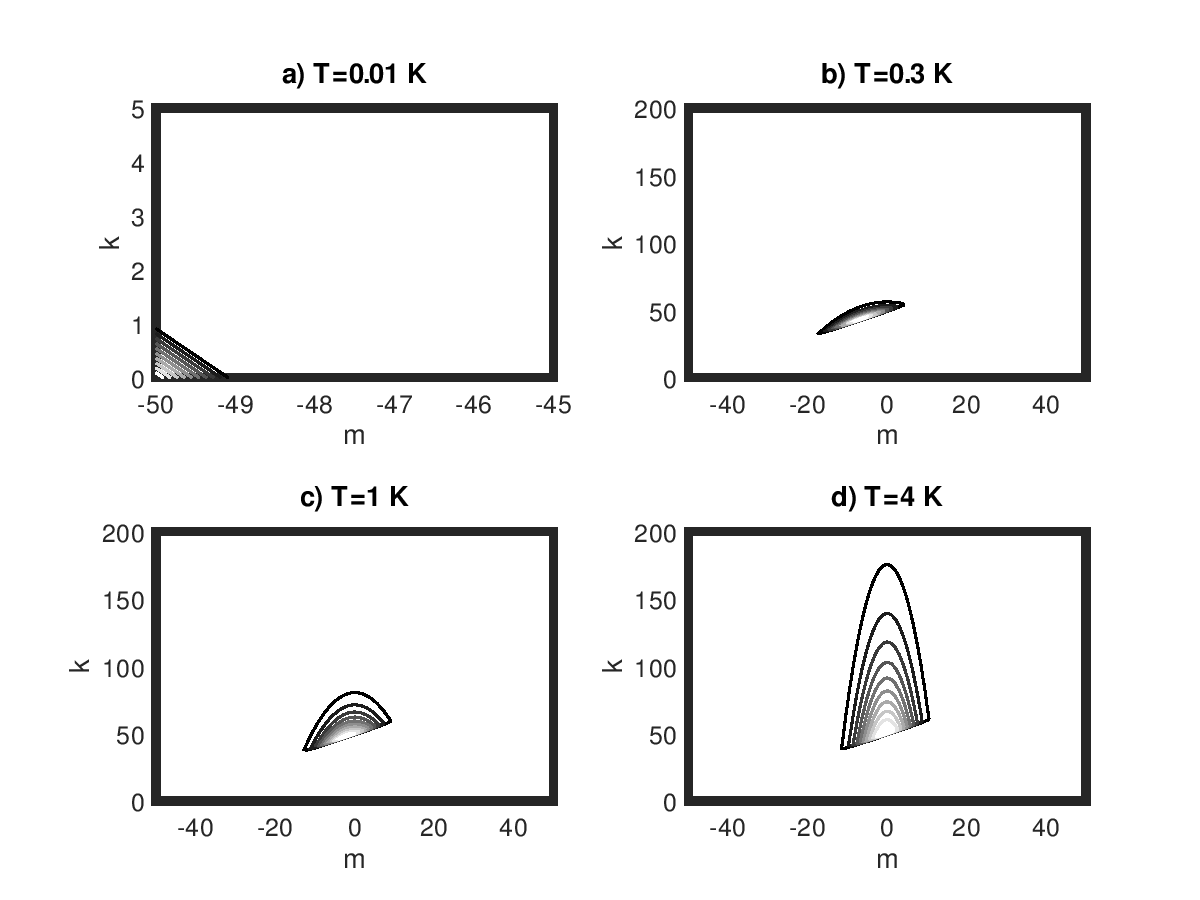}
    \caption{Figure of the distribution of the population for $Z_{total}$ as a function of the number of excitations in the system ($k$) against the secondary spin quantum number ($m$). These are plotted for $N=100$, $g_0=200\pi$ Hz, and $\omega_0=20\pi$ GHz. Note that for $T=0.01$ Kelvin the population is essentially in the Dicke subspace and global groundstate. }
    \label{distprobs}
\end{figure}

\if{false}
    \begin{figure}[h]
    \centering
    \includegraphics[scale=0.9]{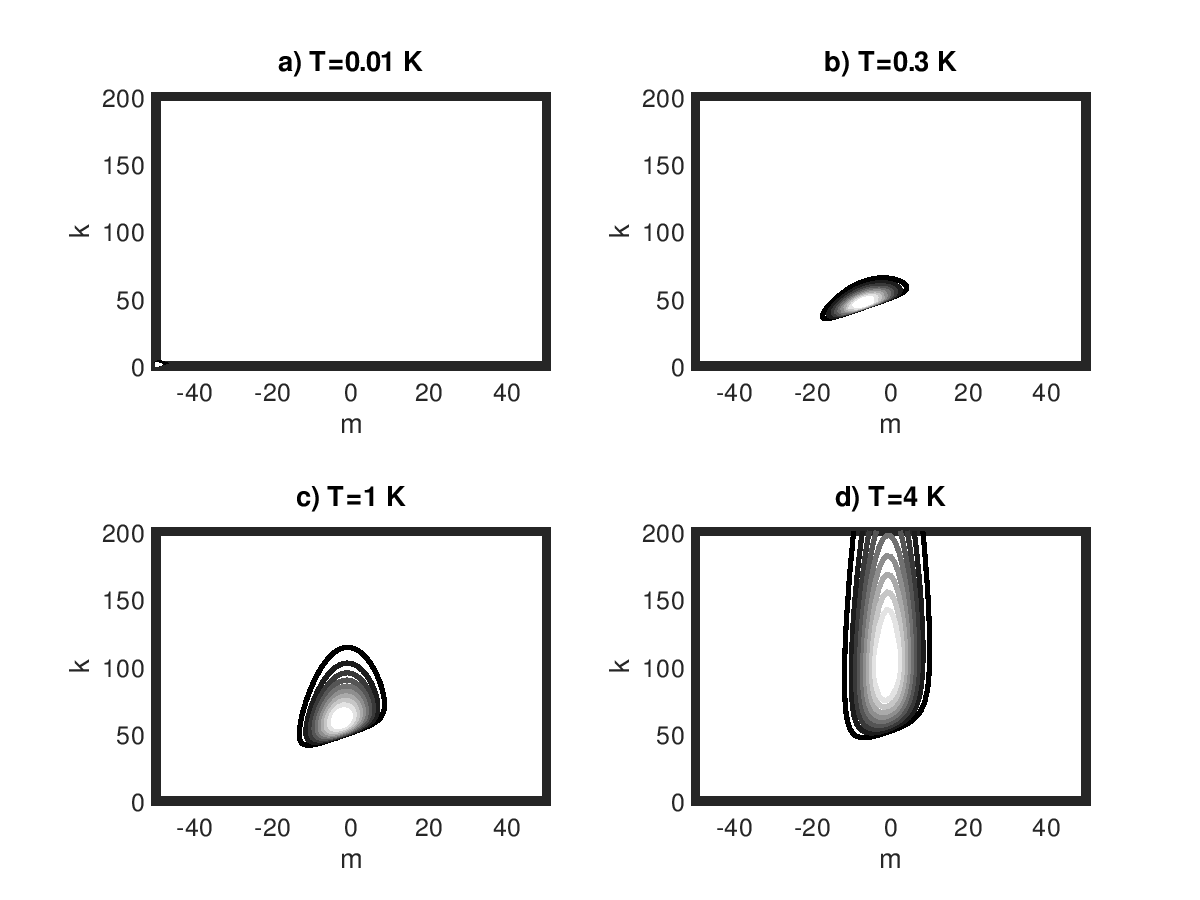}
    \caption{Figure of the distribution of the population for $Z_{pert}$, the Lamb shifts alone, as a function of the number of excitations in the system ($k$) against the secondary spin quantum number ($m$). These are plotted for $N=100$, $g_0=100$ Hz, and $\omega_0=10$ GHz. Note that for $T=0.01$ Kelvin the population is essentially in the Dicke subspace and global groundstate. The volumes are given by $2.81\cdot 10^{-16}$, $1.25\cdot 10^{-15}$, $3.82\cdot 10^{-16}$, and $7.91\cdot 10^{-17}$, respectively. The temperature dependence of these volumes does not seem to have a clear trend.}
    \label{shiftsalone}
\end{figure}
\fi

    \begin{figure}[h]
    \centering
    \includegraphics[scale=0.9]{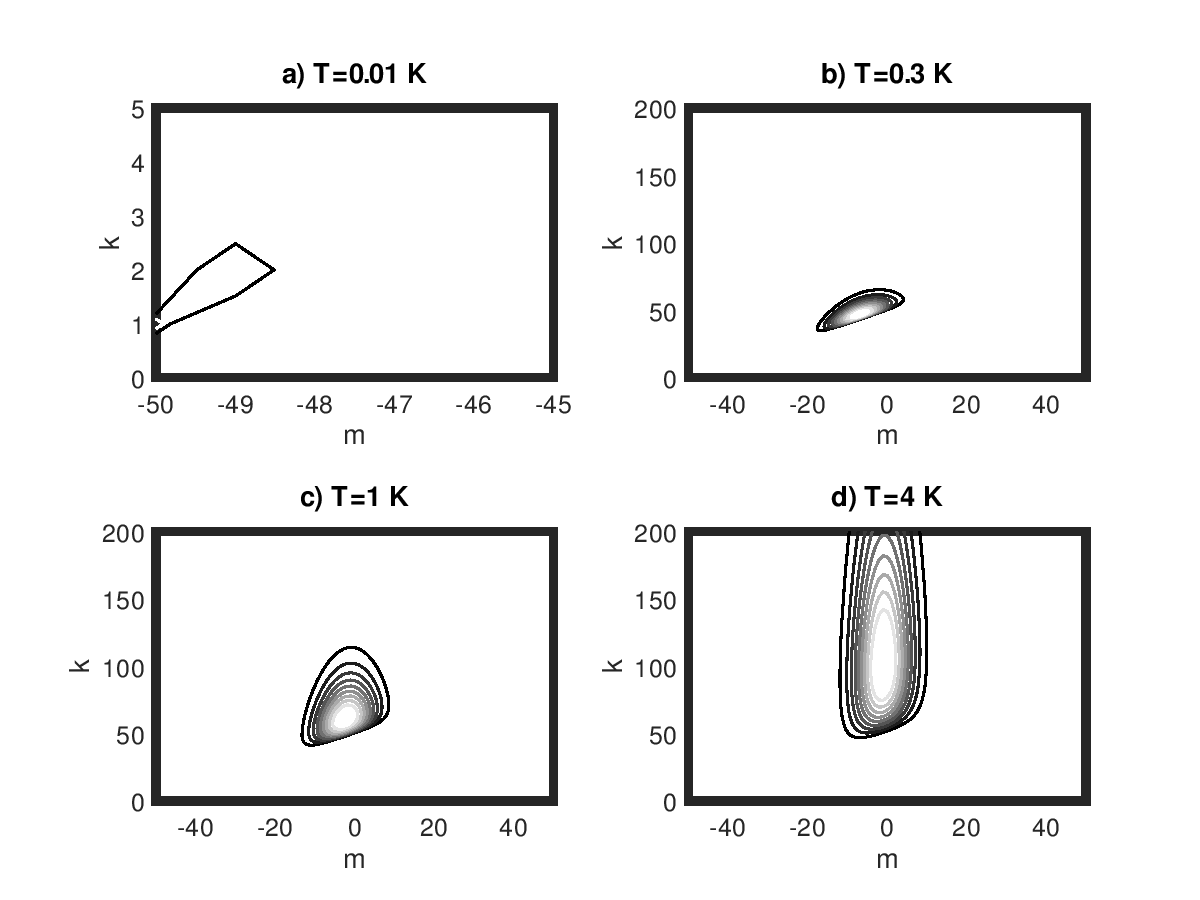}
    \caption{Figure of the distribution of the population for $Z_{pert}$, the Lamb shifts alone, as a function of the number of excitations in the system ($k$) against the secondary spin quantum number ($m$). These are plotted for $N=100$, $g_0=200\pi$ Hz, and $\omega_0=20\pi$ GHz. Note that for $T=0.01$ Kelvin the population is essentially in the Dicke subspace and global groundstate. The volumes, normalized by $Z_{total}$, are given by $2.81\cdot 10^{-16}$, $1.25\cdot 10^{-15}$, $3.82\cdot 10^{-16}$, and $7.91\cdot 10^{-17}$, respectively.}
    \label{shiftsalone}
\end{figure}
\if{false}
    \begin{figure}[h]
    \centering
    \includegraphics[scale=0.8]{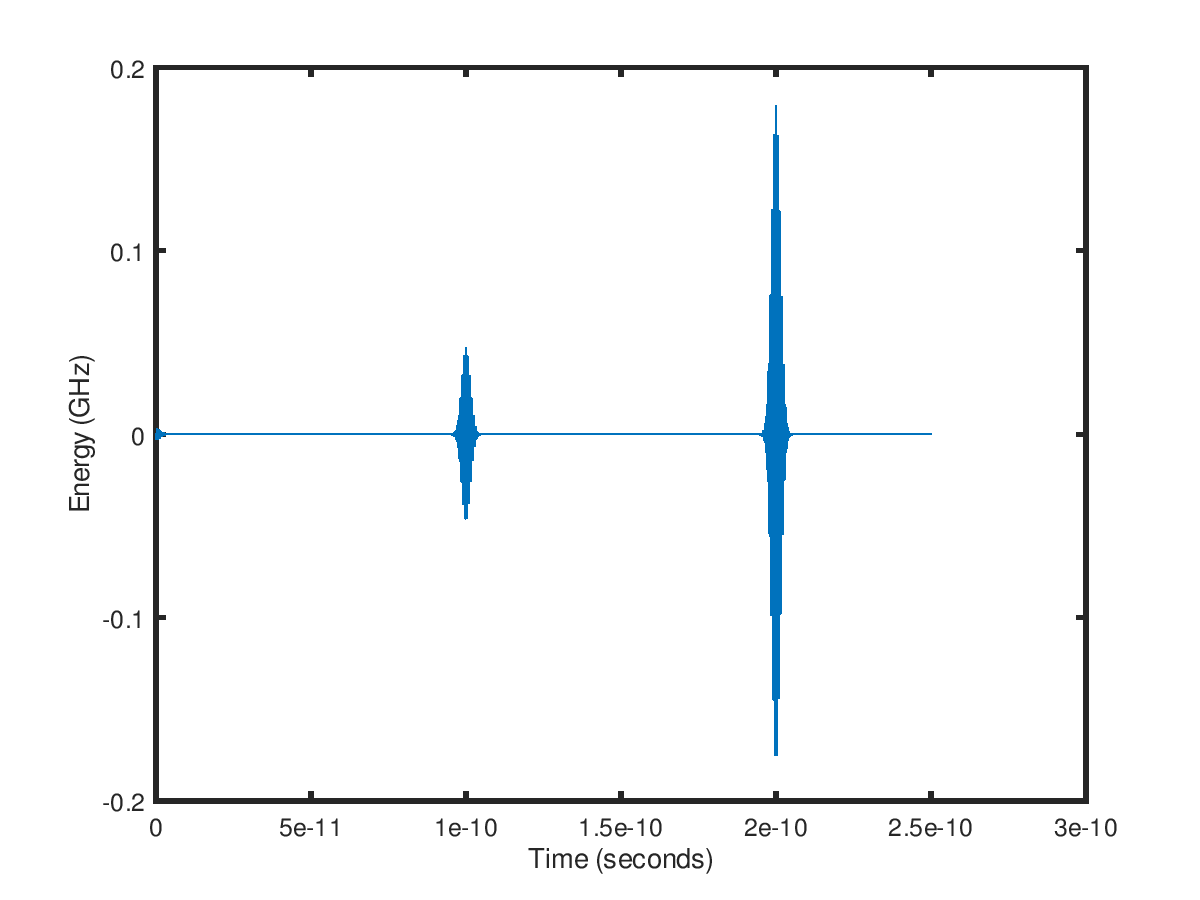}
    \caption{The real component of the Lamb-shift induced shift in the expectation of $\omega_0 a^\dag a$ self-evolution of the Tavis--Cummings model beginning with the thermal state (equation (\ref{realself})) for $N=1000$, $T=0.3$ K, $g_0=200\pi$ Hz, and  $\omega_0=20\pi$ GHz.}
    \label{nhatreal1000}
\end{figure}
\fi
\if{false}
\begin{figure}[h!]
    
    \begin{subfigure}{.3\textwidth}
        \includegraphics[width=.95\linewidth]{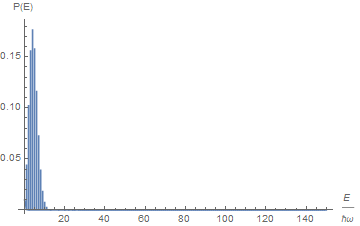}
        \caption{$T = 0.025$ K}
    \end{subfigure}
    \begin{subfigure}{.3\textwidth}
        \includegraphics[width=.95\linewidth]{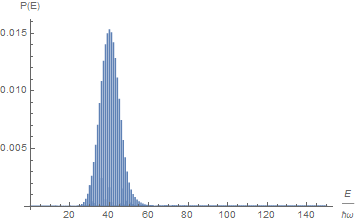}
        \caption{$T = 0.200$ K}
    \end{subfigure}
    \begin{subfigure}{.3\textwidth}
        \includegraphics[width=.95\linewidth]{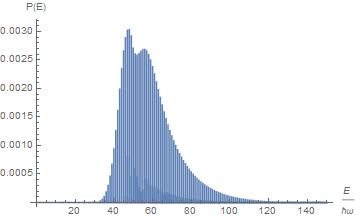}
        \caption{$T = 1.000$ K}
        \label{fig:1c}
    \end{subfigure}
    
    \caption{Thermal population distribution for the TC Hamiltonian at $\omega_0 = 10$ GHz, $g_0 = 10$ Hz, and $N=100$ spins, for selected temperatures of $T=0.025,0.200,1.000$ Kelvin. Notice that at each temperature, the ground state population is not appreciably populated for temperatures on the order of hundreds of mK or greater. \edit{These images have binning problems. The true plots look more like (b) for higher T regime}}
    \label{fig:thermal_distro}

    \end{figure}
    \fi
\if{false}
\begin{figure}[t]
    \centering
    \includegraphics[scale=0.5]{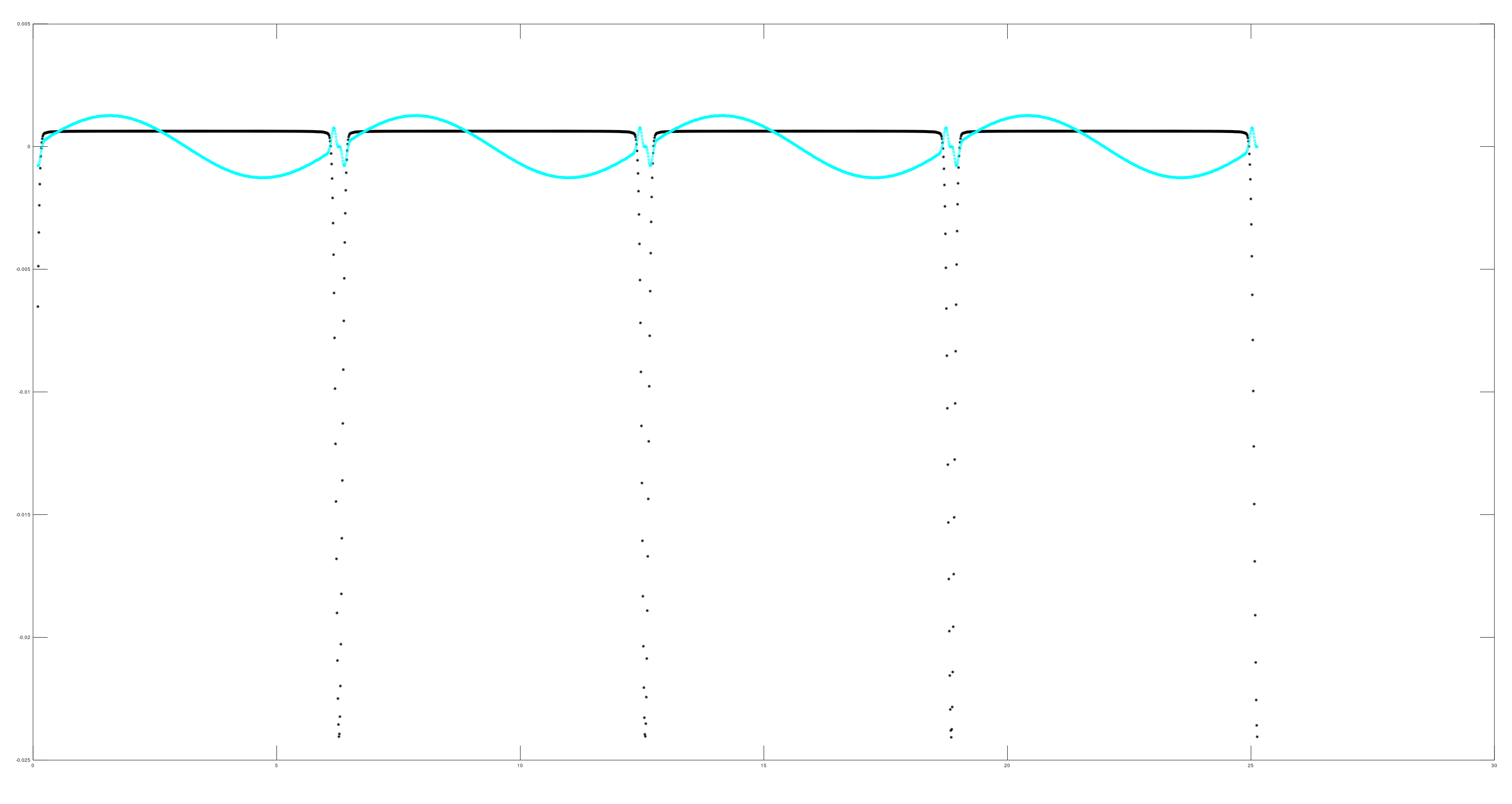}
    \caption{$N=1000$, $\omega_0=10GHz$, $T=1K$, $t\in(0,8\pi)$. Runtime: 3 minutes. The $y$-axis is very small values ($<.025$). Blue (no dips) is $Im(\langle e^{itJ_z}\rangle(t))$ and black (with dips) is $Re(\langle e^{itJ_z}\rangle(t))$. IMAGE NEEDS UPDATING, AS WELL AS SCRIPT}
\end{figure}
\fi

\if{false}

\begin{figure}[t]
    \centering
    \includegraphics[scale=0.75]{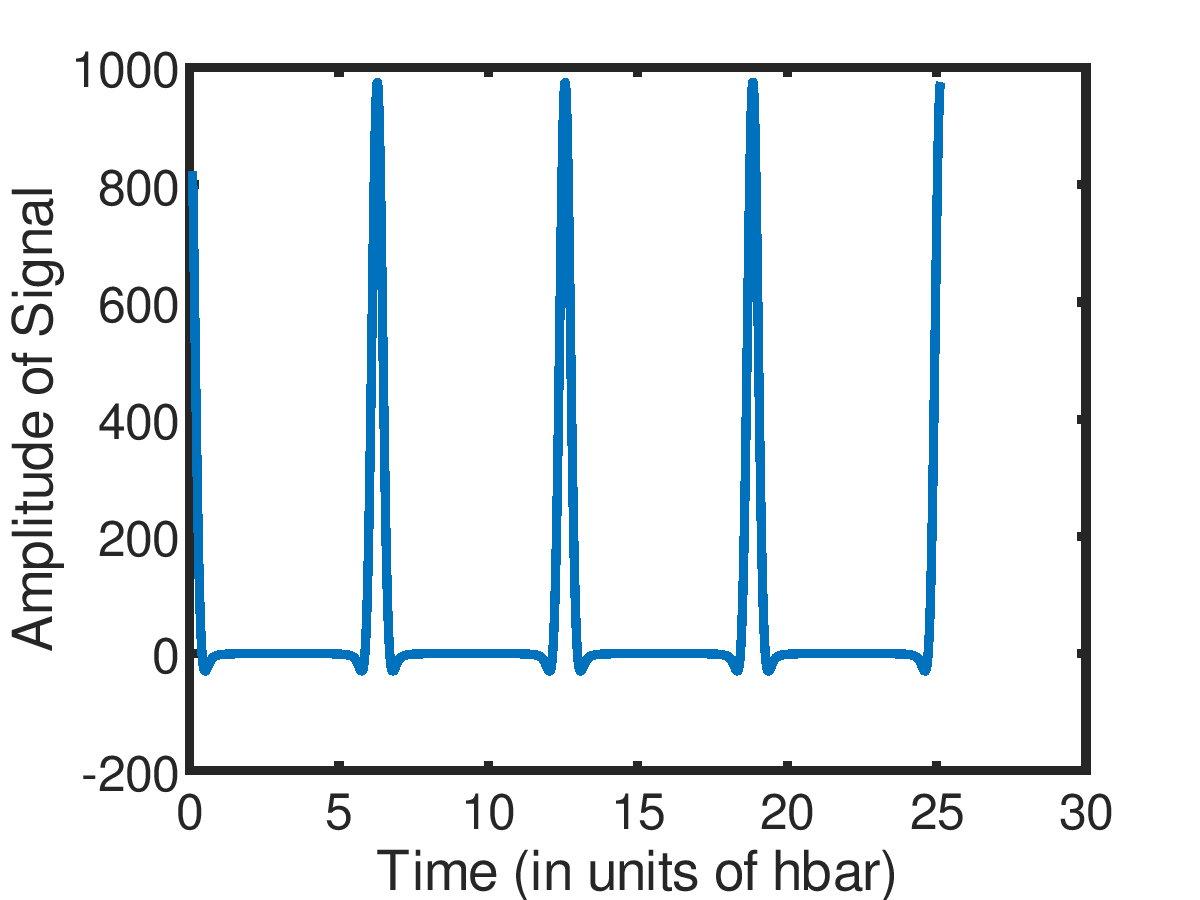}
    \caption{$N=100$, $\omega_0=10GHz$, $T=1K$, $t\in(0,8\pi)$. Runtime: 3 seconds. This shows the signal from equation (\ref{real}).}
    \label{realsignal}
\end{figure}

\begin{figure}[t]
    \centering
    \includegraphics[scale=0.75]{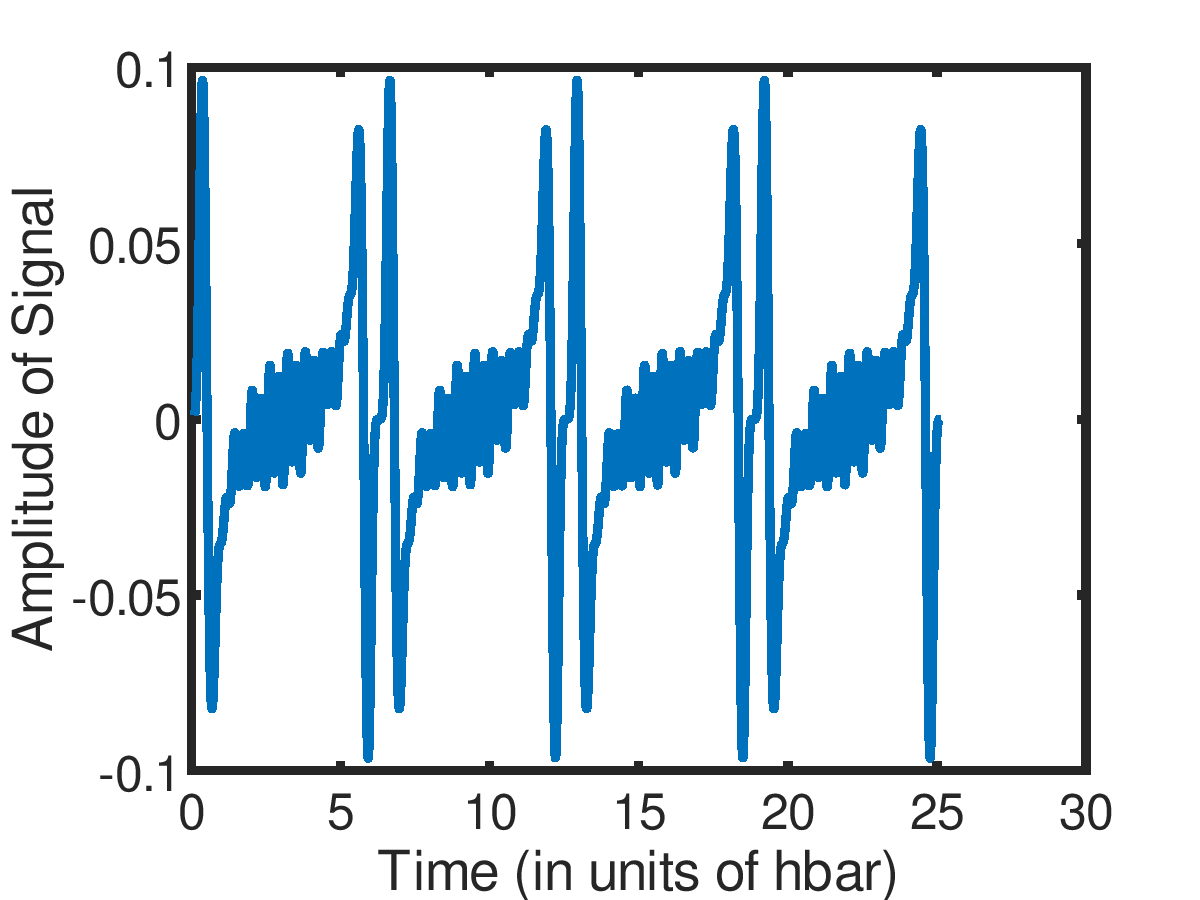}
    \caption{$N=100$, $\omega_0=10GHz$, $T=1K$, $t\in(0,8\pi)$. Runtime: 3 seconds. This shows the signal from equation (\ref{imag}).}
    \label{imagsignal}
\end{figure}

\fi

\if{false}
\begin{figure}[t]
    \centering
    \includegraphics[scale=0.5]{Visuals (32).png}
    \caption{Figures of $J_z$ due to Lamb shifts parameterized by magnetic spin value, $m$, and number of excitations, $k$, against the strength for those values. The first figure (bottom left) is a top view, with white being positive values and black being negative, while the majority of the space, as expected, is effectively zero. The middle figure is an angled perspective of the same plot allowing the amplitude and distribution to be seen more easily. The final figure (upper right) is the summation of the effect along the $k$ axis, resulting the cumulative effect as a function of $m$. For these plots $N=100$, $T=0.5$ K, $g_0=100$ Hz, and $\omega_0=10$ GHz.}
    \label{imagsignal}
\end{figure}
\fi

\if{false}

    \begin{figure}[h]
    \centering
    \includegraphics[scale=0.9]{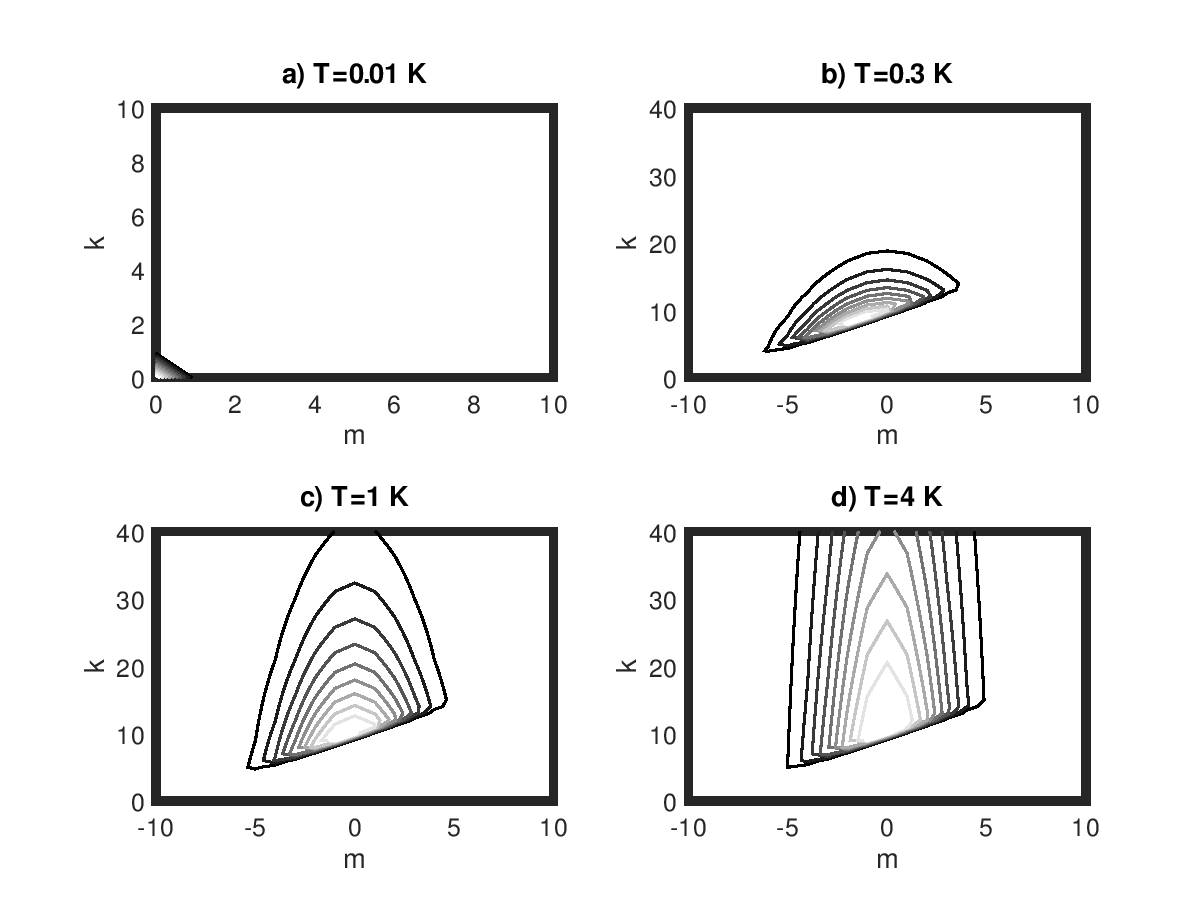}
    \caption{aaaaa}
\end{figure}
\fi
\end{widetext}

\section{Mesoscopic Experimental Predictions}

To ground the prior algorithmic results, we provide some potential use cases for our results and show experimental predictions from such experiments. This can guide experimental design as well as determine preliminary characterization of the noise impacting a given collective system. For this purpose we consider two different direct observable shifts: fractional shift in the profile of photon count and a driven experiment which could form the basis for Hahn echo sequences. Notably, with our algorithmic improvements we are able to compute values for ensembles of $1000$ spins, larger than anywhere else reported\footnote{We are limited not by runtime or computational power, but by digits of precision and more nuanced implementation details.}.



\subsection{Fractional Shift in $a^\dag a$}

As one experimentally observable signature of the impact due to Lamb shifts is the number of photons in the system at a given temperature. As the shift itself is quite small until particularly large systems are used, we will consider the fractional shift in the mean--$\langle a^\dag a\rangle_{Z_{pert}}/\langle a^\dag a\rangle_{Z_0}$. We also consider the fractional shift in the variance. Intuitively, this shift can be thought of as being due to a decrease in the average energy of the system due to Lamb shifts (as downshifted dressed states are available), which in turn reduces the number of photons in the system for a given temperature. For this we assume that we may switch on the coupling term $g_0$, which may be experimentally realized through bringing the spin system's Zeeman splitting onto resonance with the cavity's frequency.

For numerically computing $\langle a^\dag a\rangle_{Z_{pert}}/\langle a^\dag a\rangle_{Z_0}$ and $Var(a^\dag a)_{Z_{pert}}/Var( a^\dag a)_{Z_0}$, we not that the observable $a^\dag a$ (and finite powers of it) are subexponential in the parameters $j$ and $k$, which means that we may utilize our fastest algorithmic reduction. Formulae for these expressions are computed in appendix \ref{nshift}. We then compute these observables for systems with $\omega_0=20\pi GHz$, $g_0=200\pi Hz$, $T=0.3 K$, and $N$ ranging from $100$ to $1000$. From the slope with these parameters a fractional shift of a percent or more occurs for $N$ near $10^{15}$, while a fractional shift in the variance might be detected for $N$ near $10^{13}$. The results are shown in Figure \ref{nshiftlamb}.

\begin{widetext}

\begin{figure}[thb]
    \centering
    \includegraphics[scale=0.9]{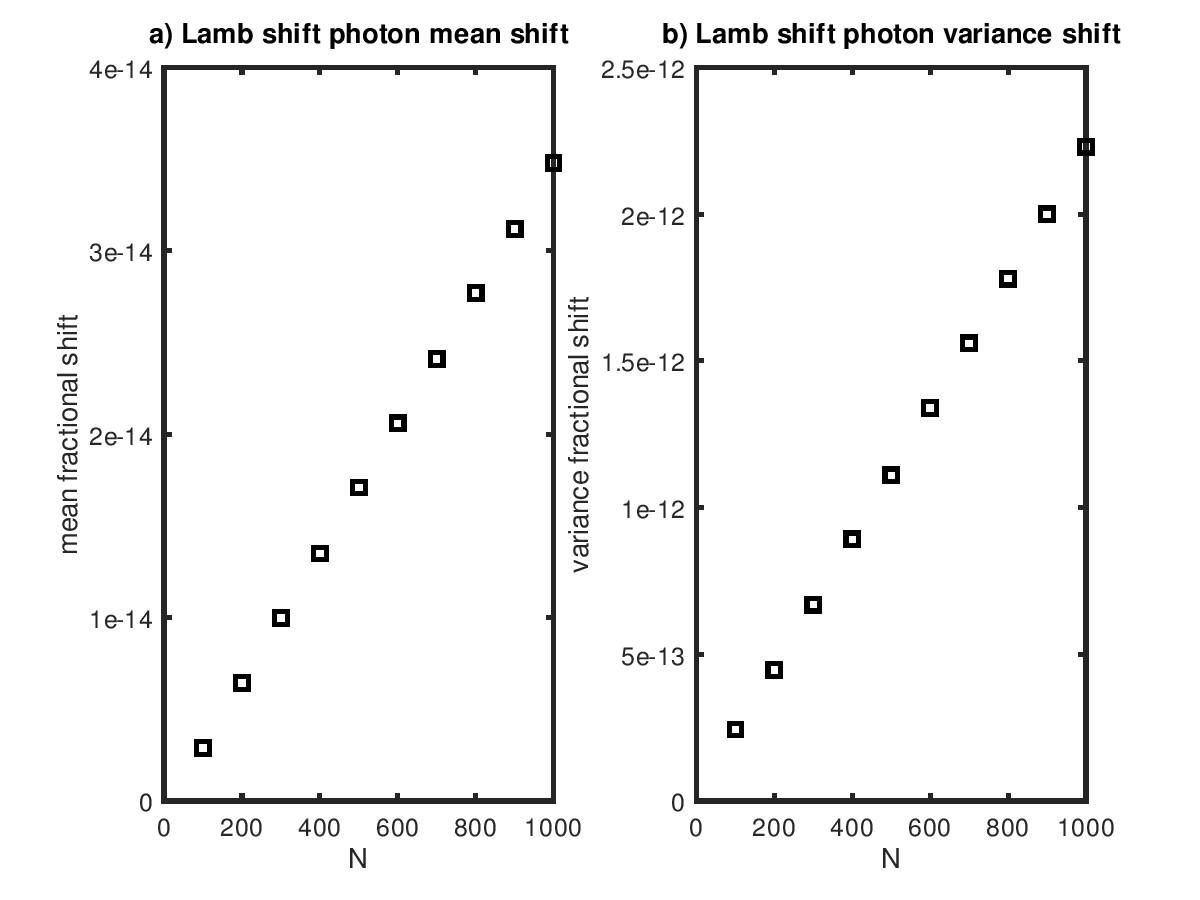}
    \caption{The fractional change in the mean (a) and variance (b) number of photons in the cavity as a function of $N$. This is plotted for $T=0.3$ K, $g_0=200\pi$ Hz, $\omega_0=20\pi$ GHz, and $N=100$ to $1000$. This is a linear trend ($r^2=1$). For (a) the slope is $3.54\cdot 10^{-17}$ and intercept is $-6.49\cdot 10^{-16}$. Once $N\approx 10^{15}$ this fraction would be appreciable. For (b) the slope is $2.21\cdot 10^{-15}$ and intercept is $9.4\cdot 10^{-15}$. Once $N\approx 10^{13}$ this fraction would be appreciable.}
    \label{nshiftlamb}
\end{figure}

\end{widetext}

These observables could serve as tests for how collective a system is or the level of discrepancy due to environmental factors. Whilst a scalar observable is a useful initial metric, having an understanding of the expected behavior under a controlling pulse would generally be more helpful. Next we explore such a case.

\subsection{Driven Rotation}

As our second application of our results, we consider evolving the system under a $J_x$ drive, along with the Hamiltonian itself, then observing $J_y$ at a given time. We find:
\begin{multline}
tr(J_y e^{it(H+2\Omega\cos(\omega_0 t) J_x)}\rho_{th}e^{-it(H+2\Omega\cos(\omega_0 t) J_x)})\\
= \sin(\Omega t) tr(J_z \rho_{th}),
\end{multline}
and compute an expression for $tr(J_z \rho_{th})$ in appendix \ref{drivensoln}. Such a sequence and associated observable could be used to form the basis of a Hahn echo based experiment, and notable deviations from the sinusoidal behavior and the predicted amplitude provide indicators for the variety and strength of noise impacting the system. This, paired with an analogous $J_z$ drive, provide  controls for traversing the hypersphere for collective spin systems, which helps open exploration for using these devices.

This observable (and associated control) also satisfies our strictest requirements for simulation and so we are able to compute parameters for systems of size $1000$ without supercomputer power. As the sinusoidal portion is not impacted without the presence of noise, we focused on the amplitude term. Considering systems with $\omega_0=20\pi$ GHz, $g_0=200\pi$ Hz, $T=0.3$ K, and $N$ ranging from $100$ to $1000$, we find a quadratic dependence on the number of spins in the system. The simulation results are shown in Figure \ref{jzshiftplot}.

    \begin{figure}[h]
    \centering
    \includegraphics[scale=0.4]{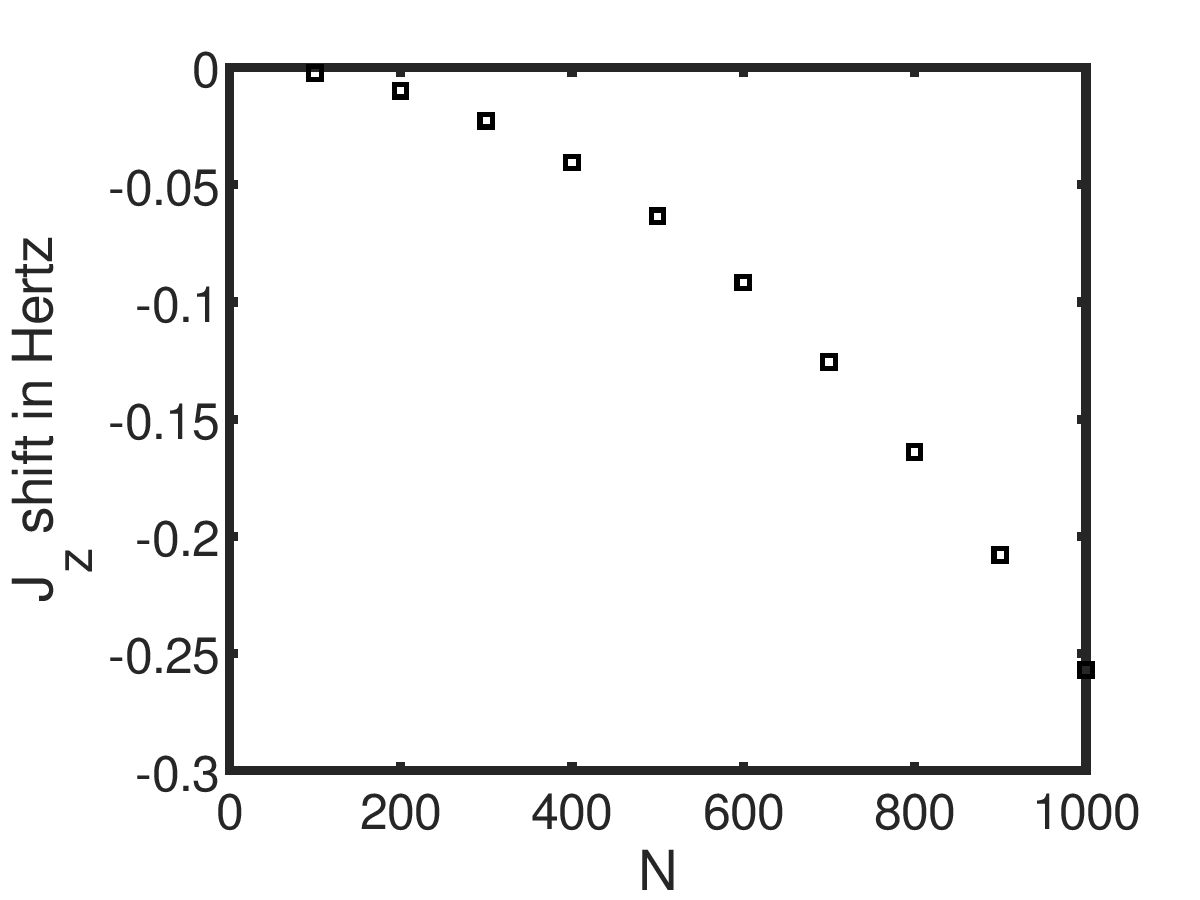}
    \caption{The shift in the value of $tr(\omega_0 J_z \rho_{th})$ due to the Lamb shifts, given by the first order correction to the trace. The negative sign indicates the availability of lower spin values, weighing the mean down. This is plotted for $T=0.3$ K, $g_0=200\pi$ Hz, $\omega_0=20\pi$ GHz, and for $N$ as $100$ to $1000$. These data points are found to follow a quadratic dependence on $N$: $-2.60\cdot 10^{-7}N^2+2.82\cdot 10^{-6}N-8.6\cdot 10^{-5}$, where a shift of $1$ GHz occurs once $N\approx 10^7$. Using a Dicke approximation and the same parameters this value rapidly tends to $0$ due to the increasing total shift in the partition function \footnote{For $N=100$ to $1000$ the shifts are $[-7\cdot 10^{-11},-3\cdot 10^{-18},-7.1\cdot 10^{-26},-1.3\cdot 10^{-33},-2.1\cdot 10^{-41},-3.1\cdot 10^{-49},-4.4\cdot 10^{-57},-5.9\cdot 10^{-65},-7.8\cdot 10^{-73},-9.9\cdot 10^{-81}]$}.} 
    \label{jzshiftplot}
\end{figure}

Unlike the fractional shift in the photon count, this provides a time-dependent signature which can be sought and leveraged to determine characteristics of the system. We conclude with these results, however, considerations of the impact of noise on these signatures will be of great interest we anticipate--the incorporation of weak noises, and when the simulation reductions are still valid, are discussed in appendix \ref{perts}.

\if{false}We will consider the changes in the number operator, $a^\dag a$, under the self-evolution of the Hamiltonian:
\begin{equation}
    \langle \omega_0 a^\dag a e^{itH}\rangle_{Z_{total}}(t)=\frac{tr(\omega_0 a^\dag a e^{-\beta H} e^{itH})}{tr(e^{-\beta H})}.
\end{equation}
In the appendix we break this into real and imaginary components and rewrite the expression in terms of sums, obtaining for the real portion:
\begin{widetext}
\if{false}\begin{multline}\label{realself}
    Real[\langle e^{itH}\rangle_{Z_{total}}(t)]=\frac{1}{2Z_{total}}(\sum_k e^{-\beta \omega_0 k}\cos(\omega_0 k t)\sum_j d_j |\mathcal{B}|(1+\frac{\beta^2-t^2}{2}Var(\Lambda(j,k)))\\
    -\sum_k e^{-\beta\omega_0 k}\sin(\omega_0 k t)\beta t \sum_j d_j |\mathcal{B}|Var(\Lambda(j,k)))
\end{multline}\fi
\begin{multline}\label{realself}
    Real[\langle \omega_0 a^\dag a e^{itH}\rangle_{Z_{total}}(t)]=\frac{1}{2Z_{total}}(\sum_k e^{-\beta \omega_0 k}\cos(\omega_0 k t)\sum_j d_j \frac{\beta^2-t^2}{2}\hat{n}_{pert}(j,k)\\
    -\sum_k e^{-\beta\omega_0 k}\sin(\omega_0 k t)\beta t \sum_j d_j \hat{n}_{pert}(j,k)
\end{multline}
\end{widetext}
where we have made the approximation that the timescale of interest is such that $tH_{pert}\ll 1$. Note that this operator is bounded by $1$ and so is subexponential in $j$ and $k$ meaning that the region of strong support is kept.

\begin{widetext}
\if{false}
    \begin{figure}[h]
    \centering
    \includegraphics[scale=0.8]{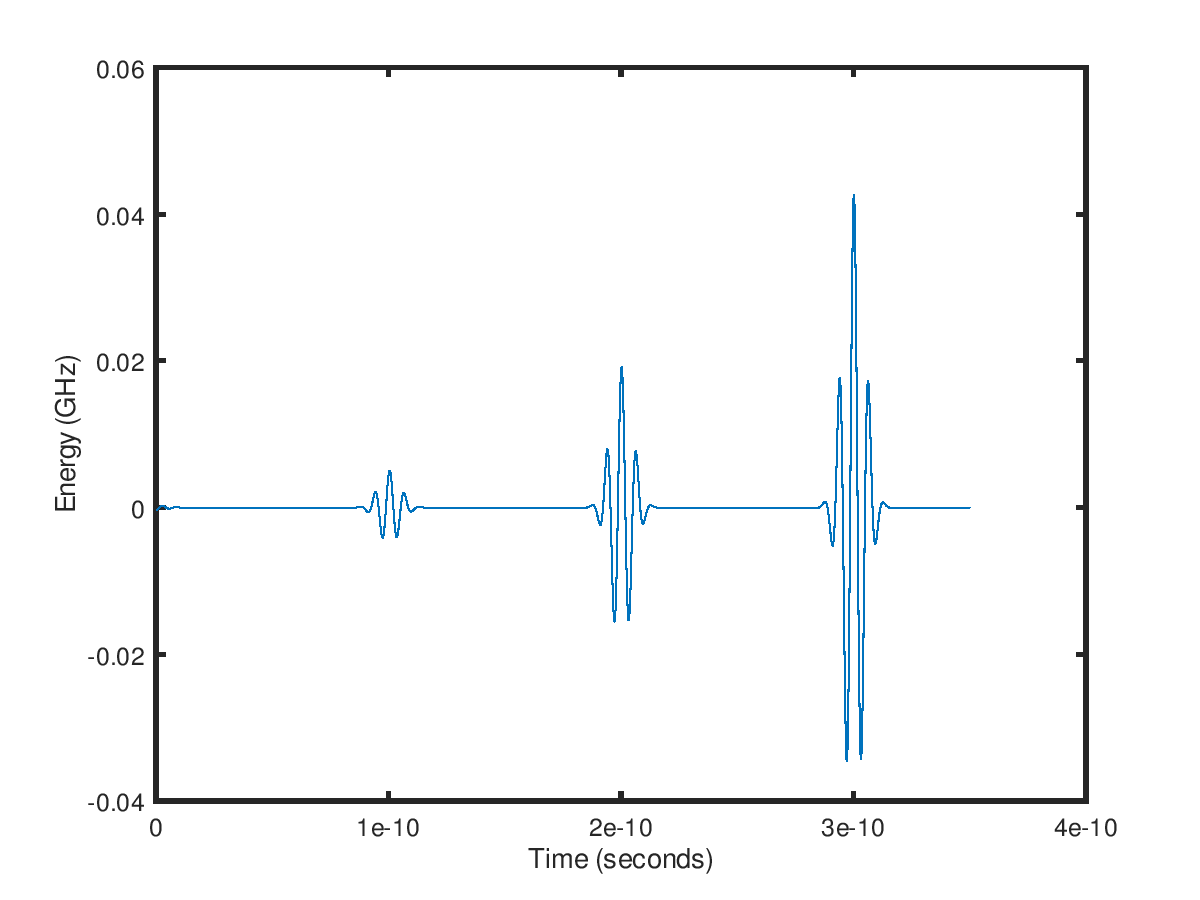}
    \caption{The real component of the Lamb-shift induced shift in the expectation of $\omega_0 a^\dag a$ self-evolution of the Tavis--Cummings model beginning with the thermal state (equation (\ref{realself})) for $N=100$, $T=0.3$ K, $g_0=200\pi$ Hz, and  $\omega_0=20\pi$ GHz. \textbf{I have plots for $N=1000$, the problem is that the packets have too high frequency to see anything. What should be done? Focus on the pack in $N=1000$? Figure \ref{nhatreal1000} has this image}}
    \label{nhatreal}
\end{figure}
\fi
\if{false}
    \begin{figure}[h]
    \centering
    \includegraphics[scale=0.8]{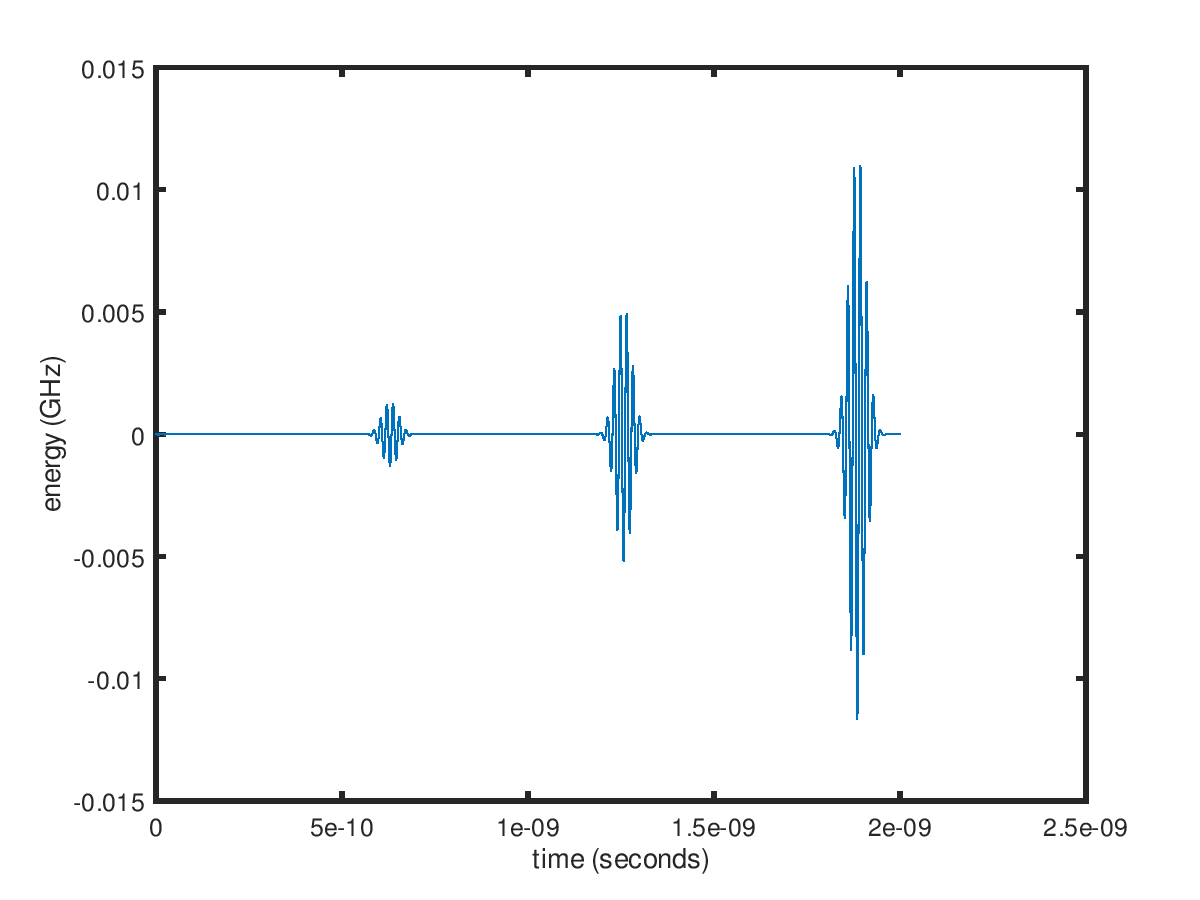}
    \caption{The real component of the Lamb-shift induced shift in the expectation of $\omega_0 a^\dag a$ self-evolution of the Tavis--Cummings model (equation (\ref{realself})) for $N=100$, $T=0.1$ K, $g_0=100$ Hz, and  $\omega_0=10$ GHz.}
    \label{nhatreal}
\end{figure}
\fi

\if{false}
    \begin{figure}[h]
    \centering
    \includegraphics[scale=0.8]{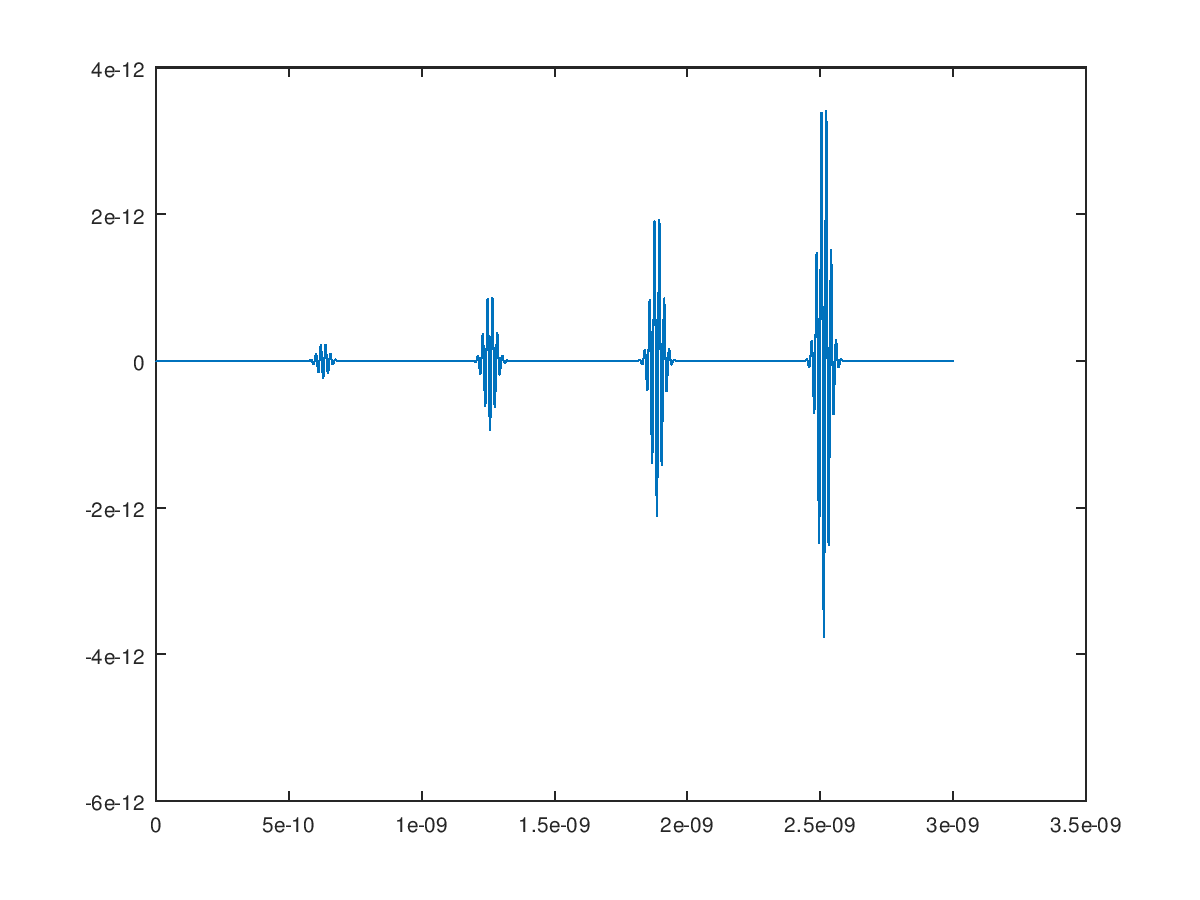}
    \caption{The real component of the Lamb-shift induced self-evolution of the Tavis--Cummings model (equation (\ref{realself}), with the $1$ in the first innermost sum removed) for $N=100$, $T=0.1$ K, $g_0=100$ Hz, and  $\omega_0=10$ GHz. The $x$-axis is time in seconds.}
    \label{selfsignalreal}
\end{figure}

    \begin{figure}[h]
    \centering
    \includegraphics[scale=0.8]{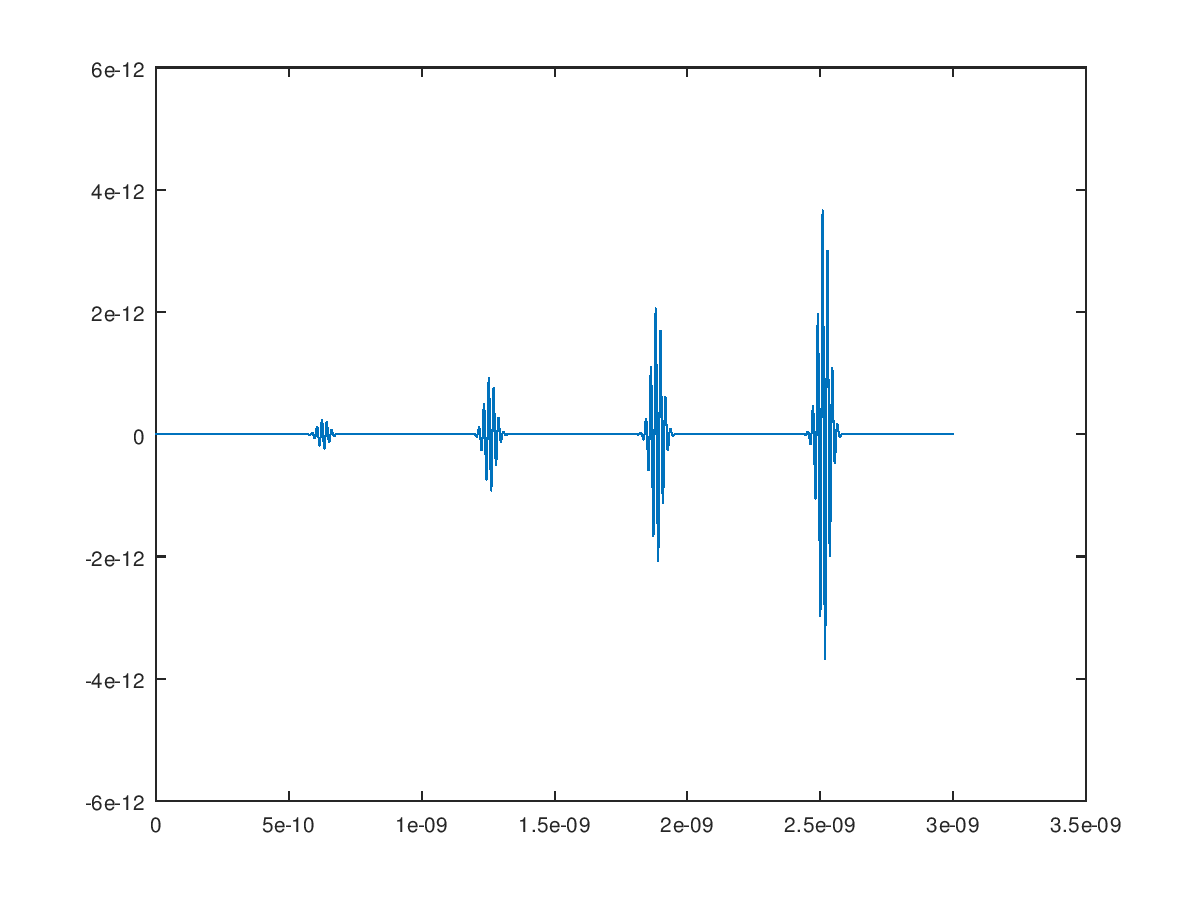}
    \caption{The imaginary component of the Lamb-shift induced self-evolution of the Tavis--Cummings model (equation (\ref{imagself})) for $N=100$, $T=0.1$ K, $g_0=100$ Hz, and  $\omega_0=10$ GHz. The $x$-axis is time in seconds.}
    \label{selfsignalimag}
\end{figure}
\fi
\end{widetext}


\if{false}To ground this discussion, we will consider the operator $e^{itJ_z}$ and its expectation. Firstly, this operator is a collective spin operator and has no dependence on $k$, so the sum structure is preserved.  Note that this operator is bounded by $1$ and so is subexponential in $j$ and $k$.\fi \if{false} As shown in appendix \ref{evolution} this has:
\begin{widetext}
\begin{equation}\label{real}
    Re(\langle e^{itJ_z}\rangle)=\frac{1}{Z_{total}}\sum_k ke^{-\beta \omega_0 k}\sum_j d_j(\sin(jt)\cot(t/2)+\cos(jt))
\end{equation}

\begin{equation}\label{imag}
    Im(\langle e^{itJ_z}\rangle)= \frac{1}{Z_{total}}\sum_k e^{-\beta \omega_0 k}\sum_j d_j 2\sin^2(t/2)((j+1)\sin(jt)-j\sin((j+1)t))
\end{equation}
\end{widetext}

Notice that since $e^{itJ_z}$ is not Hermitian, and so not an observable, we have written expressions for the real and imaginary components, which can be obtained as a sum and difference of the operator with both signs in the exponent. In these expressions the real component corresponds to observations made due to the oscillator portion of the uncoupled Hamiltonian, whereas the imaginary part is due to the collective Zeeman term in the uncoupled Hamiltonian. To compute this expression for a given time $t$ it takes $O(T\sqrt{N})$ as the same reductions apply. Considering a large collection of times, $\tau$, will not take more runtime as the trigonometric expressions can be evaluated ahead of time and the terms pulled out as needed.
\fi

These observables are weighted sums of sinusoidal functions, with the weights determined by $j$ and $k$. There is a quadratic dependency on time, $t$, showing that the Lamb shift induced self-signal will grow over time--although eventually is periodic. The spacing between the wave packets is dominantly determined by the value of $\omega_0$, while the width of the packets increases as the temperature decreases. A simulation is given in Figure \ref{nhatreal}.



\fi

\section{Conclusion}



Understanding further details about the interactions between ensembles of spins and cavities enables better control of such systems opening them for more uses. These hybrid quantum devices could be used as processors, a metrological tool, and a test bed for collective model behaviors. In this work we focused on the Tavis--Cummings model and the thermal properties of this model. In the case of this model we found that the system goes through a smooth transition from having predominantly angular momentum ground states populated to having the population spread out over a collection of $\Theta(\sqrt{N})$ collective spin values, the so-called degeneracy dominated regime. This informs a perturbative expansion for the partition function and general thermal averages in this model. From this expansion, we are able to obtain optimal simulation runtimes at least for many choices of observables, and demonstrate expectations, with up to $1000$ spins, from experiments. This provides a direct way to test the quality of the model and opens it up significantly for further exploration--both numerically, theoretically, and experimentally.

The body of this work has focused on the case of perfect on-resonance between the cavity and spins, as well as a perfectly homogeneous field. In order to improve the utility of these results the effects of such perfections removal should be analyzed and it should be determined whether current technology inhomogeneities are sufficiently small for the effects of these ensemble Lamb shifts to be observed. Considerations of these inhomogeneities are discussed in appendix \ref{perts}, and provide a benchmark with which experimental setups may be compared, as well as when the strength of these perturbations no longer allow our reductions without significant loss of accuracy. While this work does not delve into these considerations further, we provide bounds on validity which should aid further exploration.

Also of note, the complexity of this quantum system has been determined, even in the presence of slight perturbations. This means that within the complexity of quantum systems community and experimental realizations of early-term quantum devices, the Tavis--Cummings model, as considered here, is not sufficiently complex of a system. Either larger perturbations or dipolar couplings should be considered in order to generate a truly "hard" problem for a classical computer to solve.

In this work we were readily able to compute expectations for parameters for up to $N=1000$ spins on a laptop computer. As these generally took only a few minutes to generate the data, computational power is not the limiting constraint here but rather the digits in standard programming environments. Given the algorithmic scaling, with sufficiently careful programming there is no reason to not expect that expectations for systems of size $N\approx 10^7$ or larger could be computed, if desired. We leave this as another direction for others to explore.

Lastly these results must be translated fully into an experiment and a prediction for the experimental signal. The prior section provided clear expected experimental results which can provide insights into the exact regimes experimentally considered as well as how well this theory matches experiments. This will hopefully inform the parameters in the experiment and permit the resolution of this collective effect and the associated Lamb shifts, as well as better understanding of spin ensemble systems in this regime. In some regards this fundamental model is an ideal instance for collective properties and herein we have provided insights making it significantly more numerically accessible.

    \if{false}

\section{Results} \label{sec:methods}
    In the Tavis--Cummings model we usually find ourselves restricting to the Dicke subspace, so that we can find analytical and tractable numerical solutions to the functions of interest. However, these methods often fail to accurately represent the states and dynamics of the model, especially when the Dicke subspace is not sufficiently populated, as we have shown in our previous work \cite{gunderman2022lamb}. A common approximation procedure is the semi--classical approximation, which we have included in the appendix for completeness.
    
    In this work we are interested in the thermal states of the Tavis--Cummings model compared to the case where the cavity and spin ensemble are non--interacting. We begin by providing the behaviors of the non--interacting model and will proceed to the more general Tavis--Cummings case. The cavity and spin ensembles are said to be non--interacting if the coupling strength $g_0$ is zero or if the spin and cavity are sufficiently out of resonance. So that this case more closely matches that of the interacting Tavis--Cummings model we will assume that the coupling strength $g_0$ is zero.
    
    The thermal state of the spin ensemble is given by
    \begin{eqnarray}
         \hat{\rho}_{\text{th}}^{\text{spins}} &=& \otimes_{i=1}^N \frac{e^{-\beta \omega_0/2}|\uparrow\rangle\langle\uparrow|+e^{\beta\omega_0/2}|\downarrow\rangle\langle\downarrow|}{e^{-\beta \omega_0/2}+e^{\beta\omega_0/2}} \\
        &=& \otimes_{i=1}^N \frac{1}{2}\bigg(\openone - \tanh\big(\frac{\beta \omega_0}{2}\big) \hat{\sigma}_z^{(i)}\bigg). 
    \end{eqnarray}
    The thermal state of a cavity system is given by
    \begin{equation}
        \hat{\rho}_{\text{th}}^{\text{cavity}} = \big(1-e^{-\beta\omega_0}\big)\sum_{k=0}^{\infty}e^{-k \beta \omega_0}\ketbra{k}{k}.
    \end{equation}
    
    As there is no coupling term being considered at this time the thermal state of the two systems, having thermalized separately, is given by the product state
        $\hat{\rho} = \hat{\rho}_{\text{th}}^{\text{c}} \otimes \hat{\rho}_{\text{th}}^{\text{s}}$,
    so the partition function of this joint system can then be given as:
    \begin{equation}\label{z0}
        Z_0=(1-e^{-\beta\omega_0})^{-1}(e^{-\beta\omega_0/2}+e^{\beta\omega_0/2})^N.
    \end{equation}
    We also provide the average energy of this system as we will compare this to that of the coupled Tavis--Cummings model:
    \begin{equation}
        \langle E\rangle_{0}=\omega_0 (e^{\beta \omega_0}-1)^{-1}+\frac{\omega_0 N}{2}\left(1 - \tanh\left(\frac{\beta \omega_0}{2}\right)\right).
    \end{equation}
    This has a limit of $0$ at $T=0$, as we have shifted the Zeeman term by $\frac{N}{2}\openone$, and a high temperature limit of:
    \begin{equation}
        T+\frac{\omega_0}{2}(N-1)-\frac{\omega_0^2}{4}(N-\frac{1}{3})T^{-1}+O(T^{-3}).
    \end{equation}
    Having discussed the $g_0=0$ case, we are now prepared to make remarks about the $g_0> 0$ case.
    
    \subsection{Quasicritical temperature and low temperature information}
    
    STATE SUPER CLEARLY TEMP RANGE
    
    We begin by remarking on the comparative simplicity of the very low-temperature regime for this problem. Due to the utility of the analytic calculations in the Dicke subspace, it is useful to have a temperature dependent condition for determining whether the system is well approximated by a restriction to the Dicke subspace. Notably, since angular momentum is a conserved quantity when restricting to collective operators, any population that starts outside of the Dicke subspace cannot be transferred into it by collective dynamics. This has made the analysis of cavity cooling quite difficult, as one needs to break the collective symmetry in order to cool to the Dicke ground state \cite{wood2014cavity, wood2016cavity}. Further, as we have shown in figure \ref{fig:thermal_distro}, the thermal distributions indicate that we require the system to thermalize with an environment described by temperatures on the order of tens of mK or lower in order to appreciably populate the Dicke subspace.
    
    Using the block diagonal structure of the Hamiltonian we may write the partition function as:
    \begin{equation}
        Z_{total}(\beta) = \sum_{k=0}^\infty \sum_{j=j_0(k)}^{N/2}\bigg( d_j e^{-\beta k \omega_0} \sum_{\lambda\in\Lambda(j,k)} e^{-\beta \lambda g}\bigg),
    \end{equation}
    where we have defined
    \begin{equation}
        j_0(k) = \max \lbrace N/2-k, \ N/2 - \floor{N/2}\rbrace.
    \end{equation}
    For low temperatures, the degeneracy, $d_j$ plays an extremely important role within the partition function. We find a quasicritical temperature in our system, with value roughly $T_c=\omega_0/\log N$ (see appendix). Below $T_c$ a fair portion of the population is contained in the lowest lying energy states with the largest collective angular momentum, such as the Dicke subspace, but once we are above $T_c$ the rapid degeneracy increase across angular momentum subspaces dominates and the population is strongly biased toward the available angular momentum space closest to $j^*$.
    
    
    \subsection{Higher Temperature Perturbative Expansions}
    
    We now turn our attention to the case of higher temperatures. We find that we may expand the partition function in a perturbative series expansion. So long as $T\gtrsim g_0N\sqrt{\kappa}$, where $\kappa$ is the largest excitation with non-negligible population, (or generally $\beta\omega_0<1$, which is around $0.8$ Kelvin for a typical value of $\omega_0$), we may approximate $Z_{total}(\beta)\approx Z_0(\beta)+Z_{pert}(\beta)$ with:
    \begin{equation}
        Z_0(\beta)=\sum_k e^{-\beta k\omega}\sum_j d_j|\mathcal{B}_{j,k}|
    \end{equation}
    which is the same as equation (\ref{z0}), and
    \begin{equation}
        Z_{pert}(\beta)\approx \frac{(\beta g_0)^2}{2}\sum_k e^{-\beta k\omega}\sum_j d_j|\mathcal{B}_{j,k}| Var(\Lambda(j,k)).
    \end{equation}
    This then leads to the following expression for the average energy in this system:
    \begin{equation}
        \langle E\rangle\approx \langle k\rangle _{Z_0}\omega_0+((\langle k\rangle_{Z_{pert}}-\langle k\rangle_{Z_0}) \omega_0-\frac{2}{\beta})\frac{Z_{pert}}{Z_0}.
    \end{equation}
    Alternatively we can express this in terms of the change in the average energy of our spin-cavity system:
    \begin{equation}
        \Delta\langle E\rangle =((\langle k\rangle_{Z_{pert}}-\langle k\rangle_{Z_0}) \omega_0-\frac{2}{\beta})\frac{Z_{pert}}{Z_0},
    \end{equation}
    where $\langle k\rangle_{Z_0}$ is the average number of excitations with respect to the distribution $Z_0(\beta,k)/Z_0$, and $\langle k\rangle_{Z_{pert}}$ is the analogous expectation over the perturbative partition function. Additionally, we have as the shift induced in the Helmholtz free energy:
    \begin{equation}
        \Delta \langle A\rangle=\frac{Z_{pert}}{Z_0}
    \end{equation}
    
    \subsection{Computational Improvements}
     We note that the computational complexity of computing the first order correction is faster than the next order as we have an analytic expression for $Var(\Lambda(j,k))$--computing the first order correction costs $O(NT)$. Including the next term would cost $O(N^2T)$ (although this cost expression would not increase if another order or two were included as well), however, for machine precision level $O(N^2 T\log N)$ time is needed. If, however, we also utilize our knowledge that the degeneracies have a relatively small region of strong support with size $O(\sqrt{N})$, these time complexities reduce to: $O(\sqrt{N} T)$, $O(NT)$, and $O(N^{3/2}T\log N)$, respectively, where the extra half power for the last one comes from having to solve matrices of size up to $N\times N$. Note that in all the prior expressions we have suppressed the dependency on $\omega_0$ as this is just a scaling of the variable $T$ in effect. As thermodynamic properties are computed using the partition function and this is the bottle neck for performing those computations, most thermodynamic functions should take the same run time.
     
     Beyond these numerical speedups, we may make two general remarks on $\Delta\langle E\rangle$. Firstly, as $g_0$ tends toward zero, so does this expression. Secondly, as $\langle k\rangle_{Z_{pert}}<\langle k\rangle_{Z_0}$ [due to it being a higher moment of an exponential distribution], $\Delta \langle E\rangle\leq 0$, meaning that the introduction of the coupling reduces the average energy in system. Further simplifications of the expression for $\Delta\langle E\rangle$ are beyond the scope of this work, however, we provide some illustrative figures as well as analysis in the next subsection.
     \begin{widetext}
     \begin{center}
\begin{tabular}{||c| c| c| c| c||} 
 \hline
 Error & $O(g_0^4)$ & $O(g_0^{2c}),\ c> 2$ & Machine Precision & Naive \\ [0.5ex] 
 \hline\hline
 $Z_{pert}$, $Z_{tot}$, $\Delta\langle A\rangle$, $\Delta \langle E\rangle$ & $O(T\sqrt{N})$ & $O(NT)$ & $O(TN\log N)$ & $O(TN^3)$ \\ 
 \hline
 $\mathcal{O}(j,k)$ (no region of strong support promise) & $O(NT)$ & $O(N^2T)$ & $O(TN^2\log N)$ & $O(TN^3)$ \\
 \hline
 $\mathcal{O}(j,k)$ (region of strong support promised) & $O(T\sqrt{N})$ & $O(NT)$ & $O(TN\log N)$ & $O(TN^3)$ \\ [1ex]
 \hline
\end{tabular}
\end{center}
 \end{widetext}   
    \subsection{Numerical Computation of the Thermal Distribution}
    While the expressions in the previous sections provided some insights, they also caused other aspects to be overshadowed. Here we describe how we compute the thermal distribution, instead of the partition function, as well as the energy shift induced from the coupling $\Delta\langle E\rangle$.
    
    By decomposing the TC system into subspaces labeled by angular momentum $j$ and excitations $k$, we can efficiently solve for the energy eigenvalues. We can write the thermal state populations as a probability distribution over the allowed energy levels of the system. That is, we determine the set,
    \begin{equation}
        \mathcal{E}(\beta) = \lbrace (E, P(E \ | \ \beta)) \,|\, \forall E, \H \ket{E} = E \ket{E} \rbrace,
    \end{equation}
    where the probability of the label $E$ is determined by the number of eigenstates with energy $E$, given by $d(E)$, and the temperature of the system:
    \begin{equation}
        P(E \ | \ \beta) = d(E) e^{-\beta E} / Z(\beta).
    \end{equation}
    
    We find that the eigenvalue problem needs to only be solved once for a given value of $N$ and by storing the energy spectrum we can efficiently compute the distribution of the thermal state for a given new temperature or temperature range. The results are shown in figure \ref{fig:thermal_distro}.

    \begin{figure}[h!]
    
    \begin{subfigure}{.3\textwidth}
        \includegraphics[width=.95\linewidth]{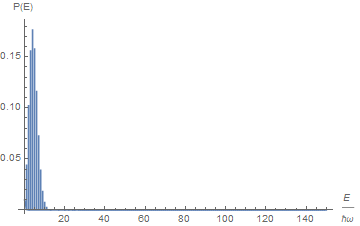}
        \caption{$T = 0.025$ K}
    \end{subfigure}
    \begin{subfigure}{.3\textwidth}
        \includegraphics[width=.95\linewidth]{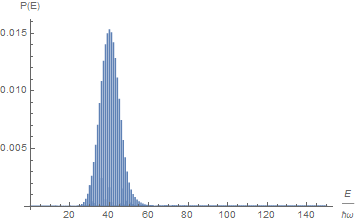}
        \caption{$T = 0.200$ K}
    \end{subfigure}
    \begin{subfigure}{.3\textwidth}
        \includegraphics[width=.95\linewidth]{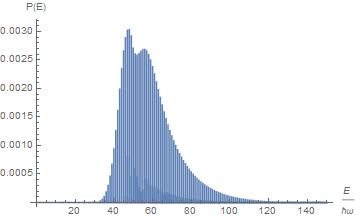}
        \caption{$T = 1.000$ K}
        \label{fig:1c}
    \end{subfigure}
    
    \caption{Thermal population distribution for the TC Hamiltonian at $\omega_0 = 10$ GHz, $g_0 = 10$ Hz, and $N=100$ spins, for selected temperatures of $T=0.025,0.200,1.000$ Kelvin. Notice that at each temperature, the ground state population is not appreciably populated for temperatures on the order of hundreds of mK or greater.}
    \label{fig:thermal_distro}

    \end{figure}

   To demonstrate this numerically, we consider the parameters $\omega_0 = 10$ GHz, $g_0 = 10$ Hz, and temperatures at or less than $1$ Kelvin. These parameters can be feasibly realized in an ESR system \cite{}. We demonstrate our ability to determine the thermal distribution by particularizing to $N=100$ identical spin--1/2 particles.

    It may be somewhat surprising to see that temperatures near $200$ mK result in a fairly high energy distribution, populating subspaces with 30 to 60 excitations. Even in the case of $25$ mK, it appears that the system has not fully cooled to its ground state.
    
    Plot for $\Delta\langle E\rangle$ instead! I want a plot of this!!! One curve with and one without, as well as just a plot of it.

    We see that while the existence of the coupling terms inhibits our ability to determine analytic solutions for the thermal state of the TC model, we find that we can still draw strong parallels on its behavior to the underlying subsystems. This concept is possibly best illustrated in figure \ref{fig:1c}, where the thermal distribution has two peaks, indicating that temperatures near 1 K form an intermediate regime, dominated by neither spin contributions, nor photonic contributions.
    

\section{Conclusion} \label{sec:conclusion}
    
    
    
 
    In this work we have described the thermal state for the Tavis--Cummings model, as well as provided a technique for perturbatively computing various thermal properties, as well as improvements to numerical methods. While these are useful aspects to know, converting these properties into other observables such as a free-induction decay (FID) would aid in verifying these results. In addition, this work assumes that the spin-cavity system is a closed system. Thus, while we allow thermalization with a reference bath, we do not allow for this interaction to induce any dissipation operators on the spins or the cavity. In a real experimental setting, there are a variety of noise sources which will inevitably break some of our idealizations and understanding how these will impact observations is crucial. For instance, a direct consequence will be that the computed Fourier spectrum of the FID will be composed not of a sum of delta functions, but rather, a sum of broadened curves. A full open quantum system treatment would be needed in order to properly predict the spectrum in this setting, which leads to new challenges, but also greater utility of these systems.

    \fi
    
\section*{Acknowledgements} \label{sec:acknowledgements}

We thank Andrew Stasiuk for help in the initial observation of regime change of the system, and Mohamed El Mandouh for helpful discussions. L.G. current affiliation: Department of Electrical and Computer Engineering, University of Illinois--Chicago, Chicago, Illinois, 60607.

\section*{Funding}
We gratefully thank the financial contributions of the Canada First Research Excellence Fund, Industry Canada, CERC (215284), NSERC (RGPIN-418579), CIFAR, the Province of Ontario, and Mike and Ophelia Lazaridis.

\bibliographystyle{unsrt}
\bibliography{refs}
\clearpage

\appendix*

\newpage
\begin{widetext}
\section{} \label{sec:app}
\if{false}
\subsection{The Semi--Classical Approximation to the Tavis--Cummings Model}\label{sec:app1}
    A tractable way to describe features of the Tavis--Cummings (TC) model for a large number of spins is by restricting ourselves to the Dicke subspace and taking a semi--classical approach, by which we can produce analytic results for the thermal behaviour of the spin--cavity system. In the paper we are mainly interested in the thermal properties of the TC model by way of computing the partition function. Since the partition function is the Laplace transform of the density of states, we concern ourselves with calculating the density of states here.
    
    The TC Hamiltonian is given by equation \eqref{eq.tc}:
    \begin{equation}
        \H = \omega_0 (\hat{a}^\dagger \hat{a} + \hat{J}_z) + g_0(\hat{a}\hat{J}_+ + \hat{a}^\dagger \hat{J}_-).
    \end{equation}
    The semi--classical density of states of the TC model in the Dicke subspace well approximates the fully quantum density of states in the thermodynamic limit (is this the Holtstein-Primakoff approximation?--Lane). The semi--classical Hamiltonian is given by substituting the spin operators $\hat{J}_i$ with classical angular momentum variables $j_i$ and substituting the cavity operators
    with classical harmonic variables $q$ and $p$ \cite{bastarrachea2014comparative, de1992chaos}. The semi--classical Hamiltonian for our system of interest is then 
    \begin{equation}
        \mathcal{H} =\frac{\omega_0}{2} (q^2+p^2) +  \omega_0 j_z + g_0\sqrt{j} \sqrt{1- \frac{j_z^2}{j^2}}\left[q \cos\phi -p\sin\phi\right],
    \end{equation}
    where $\phi = \arctan(j_y/j_x)$. 
    
    The density of states equation is computed with \cite{gutzwiller2013chaos} 
    \begin{equation}
        \nu(E)=\frac{1}{(2 \pi)^{2}} \int \dd q \ \dd p \ \dd \phi \ \dd j_{z} \ \delta\left(E-\mathcal{H}\left(q, p, \phi, j_{z}\right)\right).
    \end{equation}
    We now proceed to evaluate this expression. We will use the identity $\delta(f(x))= \sum_i \frac{\delta(x-x_i)}{|f'(x_i)|}$, where $x_i$ are the real roots of $f(x)$. We first integrate over $q$, so we can rewrite $\delta\left(E-\mathcal{H}\left(q, p, \phi, j_{z}\right)\right) = \frac{\delta\left(q-q_{+}\right)}{\left| \partial \mathcal{H} / \partial q \right|_{q_{+}}} + \frac{\delta\left(q-q_{-}\right)}{|\partial \mathcal{H} / \partial q|_{q_{-}}}$, where $q_+$ and $q_-$ are the real roots of the quadratic equation $E- \mathcal{H}\left(q, p, \phi, j_{z}\right)$, such that 
    \begin{equation}
        \nu(E)=\frac{1}{(2 \pi)^{2}} \int \dd j_{z} \ \dd \phi \ \dd p \ \dd q \ \left(\frac{\delta\left(q-q_{+}\right)}{|\partial \mathcal{H} / \partial q|_{q_{+}}}+\frac{\delta\left(q-q_{-}\right)}{|\partial \mathcal{H} / \partial q|_{q_{-}}}\right).
    \end{equation}
    The roots are 
    \begin{align}
        \omega_0 q_{\pm} &= -g_0\sqrt{j}\cos\phi \sqrt{\left(1- \frac{j_z^2}{j^2}\right)}
        \pm \sqrt{-\omega_0^2 p^2 + bp + c},
    \end{align}
    where 
    \begin{align}
        b &= 2\omega_0 g_0\sqrt{j} \sin\phi \sqrt{\left(1- \frac{j_z^2}{j^2}\right)},\\
        c &= g_0^2 j \cos{\phi}^2\left(1- \frac{j_z^2}{j^2}\right) + 2\omega_0 (E -\omega_0 j_z).
    \end{align}
    In addition, we have 
    \begin{equation}
        |\partial \mathcal{H} / \partial q|_{q_{+}} =  |\partial \mathcal{H} / \partial q|_{q_{-}} = \sqrt{-\omega_0^2 p^2 + bp + c}.
    \end{equation}
    We then integrate with respect to $q$ from $q_-$ to $q_+$ and obtain simply
    \begin{equation}
        \nu(E)=\frac{1}{(2 \pi)^{2}} \int \dd j_{z} \ \dd \phi \ \dd p \ \frac{2}{\sqrt{-\omega_0^2 p^2 + bp + c}}.
    \end{equation}
    Next, we repeat the procedure once more for $p$. The limits of the integration must be such that denominator has real roots in $p$, so we must have 
    \begin{equation}
        b^2+4\omega_0^2c \geq 0.
    \end{equation}
    But first, we can rewrite the integral as 
    \begin{equation}
         \nu(E)=\frac{1}{(2 \pi)^{2}} \int \dd j_{z} \ \dd \phi \  \dd p \ \frac{1}{\sqrt{-\omega_0^2 (p-p_+)(p-p_-)}},
    \end{equation}
    and integrate from $p_-$ to $p_+$ and obtain 
    \begin{equation}
        \nu(E)=\frac{1}{\omega_0(2 \pi)^{2}} \int \dd j_{z} \ \dd \phi.
    \end{equation}
    The roots $p_{\pm}$ are real when 
    \begin{equation}
        \frac{jg_0^2}{2\omega_0} - \frac{j_z^2 g_0^2}{j\omega_0} - \omega_0 j \geq -E,
    \end{equation}
    which leaves us with three cases in $E$. First, when $E_{\text{min}} \leq E < -\omega_0 j$ [what is $E_{min}$? Add a quickie.--Lane], we have that $j_z \in [a_- , a_+]$ holds, where 
    \begin{equation}
        a_\pm = \frac{jg_0^2}{2\omega_0} \pm \frac{g_0}{\omega_0}\sqrt{2(E-E_{\text{min}})}.
    \end{equation}

    Second case is for $-\omega_0 j < E < \omega_0 j$, we have $j_z \in [-\omega_0 j, a_+]$. Finally, for $E > \omega_0 j$ we have $j_z \in [-\omega j , \omega j]$. We have no restrictions on $\phi$, so $\phi \in [0, 2\pi)$. Integrating with respect to $j_z$ and $\phi$ we obtain
    \begin{equation}
       \nu(E) = \begin{dcases*}
        \frac{2jg_0}{\omega_0^2}\sqrt{2(E-E_{\text{min}})}, & $E_{\text{min}} \leq E < -\omega_0 j $\\
        \frac{j}{\omega_0} \left[1- \omega_0^2/g_0^2 +\omega_0/g_0 \sqrt{2(E-E_{\text{min}})}\right], & $-\omega_0 j < E < \omega_0 j$\\
        \frac{2j}{\omega_0},  & $E > \omega_0 j$.\\
        \end{dcases*}
    \end{equation}
    This technique can be extended to the rest of the angular momentum subspaces, since in the TC model each angular momentum sector is decoupled from the rest in the block diagonal basis $\ket{\alpha_{j,l}}$. Thus we can write the semi--classical density of states for all the angular momentum sectors as
    \begin{equation}
        \nu(E) = \sum_{j=0}^{N/2} d_j \nu(E, j).
    \end{equation}
    
    However, the semi--classical approximation is only good for large $N$, so the approximation begins to break down for small $j$. So if the state of the system is populated by states in the lowest lying angular momentum sectors then the approximation will fail. This may not be of concern in the thermodynamic limit $N\longrightarrow \infty$. However, for finite $N$ and general dynamics, the semi--classical approximation will not suffice. 
    
\fi
\subsection{Computation of $g_0=0$ Thermalization Case}\label{g0}
Here we show the steps involved in computing the non-interacting partition function and average energy for a spin ensemble and cavity. This involves solving the problem for a cavity and a collection of $N$ collective spins.

For a cavity, the thermal state is given by:
\begin{equation}
    \rho_{\text{th}}^{\text{cavity}}=(1-e^{-\beta \omega_0})\sum_{k=0}^\infty e^{-\beta \omega_0 k} |k\rangle\langle k|
\end{equation}
The partition function for the cavity can then be given by:
\begin{equation}
    Z_c=(1-e^{-\beta \omega_0})^{-1}
\end{equation}
The average energy for a thermalized cavity is then given by:
\begin{eqnarray}
    \langle E\rangle_c&=&(1-e^{-\beta \omega_0})\sum_{k=0}^\infty k\omega_0 e^{-\beta \omega_0 k}\\
    &=& (1-e^{-\beta \omega_0})\sum_{k=0}^\infty -\frac{\partial}{\partial \beta} e^{-\beta \omega_0 k}\\
    &=& -(1-e^{-\beta \omega_0}) \frac{\partial}{\partial \beta}\sum_{k=0}^\infty  e^{-\beta \omega_0 k}\\
    &=& \omega_0 e^{\beta \omega_0} (1-e^{-\beta \omega_0})(1-e^{\beta \omega_0})^{-2}\\
    &=& \omega_0 (e^{\beta \omega_0}-1)^{-1}
\end{eqnarray}
\if{false}
The limits for this are:
\begin{equation}
    \lim_{T\rightarrow 0} \langle E\rangle_c=0
\end{equation}
\begin{eqnarray}
    \lim_{T\rightarrow \infty} \langle E\rangle_c&\approx & \omega_0 \frac{1}{1+\beta \omega_0-1}\\
    &=& \beta^{-1} -\frac{1}{2}\omega_0+\frac{\beta\omega_0^2}{12}+O(\beta^3)
\end{eqnarray}
\fi
For a collection of $N$ uniformly interacting spins the Hamiltonian is given by $\omega_0 (J_z+\frac{N}{2}\openone)$. Note that the eigenstates of this are equivalent to those of $\otimes_{i=1}^N \sigma_z^{(i)}$. As such, we have:
\begin{align}
         \hat{\rho}_{\text{th}}^{\text{spins}} &= \otimes_{i=1}^N \frac{e^{-\beta \omega_0}|\uparrow\rangle\langle\uparrow|+|\downarrow\rangle\langle\downarrow|}{1+e^{-\beta\omega_0}} 
    \end{align}
The partition function for this portion can be given by:
\begin{equation}
    Z_s=(1+e^{-\beta\omega_0})^N
\end{equation}
The average energy is then given by:
\begin{equation}
        \avg{E}_s = N\omega_0 \frac{e^{-\beta\omega_0}}{1+e^{-\beta\omega_0}}.
    \end{equation}
    \if{false}
The limits of this are:
\begin{equation}
    \lim_{T\rightarrow 0}\langle E\rangle_s=0,\quad \lim_{T\rightarrow\infty}\langle E\rangle_s=\frac{\omega_0 N}{2}(1-\frac{\beta \omega_0}{2}+O(\beta^3)
\end{equation}
\fi
The total energy when $g_0=0$ is then given by:
\begin{eqnarray}
    \langle E\rangle_{0}&=&\langle k\rangle_{Z_0}\omega_0\\
    &=&\omega_0 (e^{\beta \omega_0}-1)^{-1}+N\omega_0 \frac{e^{-\beta\omega_0}}{1+e^{-\beta\omega_0}}
\end{eqnarray}

\subsection{Cutoff Temperature Argument}\label{regimes}


\if{false}
The following Lemma will be useful in our argument below:
\begin{lemma}\label{tay}
In the Taylor series for $e^k$, with $k\geq 1$, very, very nearly, the largest term is the $k/3$-th term.
\end{lemma}

\begin{proof}
We begin with the Taylor series:
\begin{equation}
    e^k=\sum_{t=0}^\infty \frac{k^t}{t!}
\end{equation}
We note that for $k\geq 1$, $k^0/0!<k^2/2!$ and $\lim_{t\rightarrow\infty} k^t/t!=0$. We then take the $t$ derivative by using the analytic continuation of the factorial function:
\begin{equation}
    \frac{d}{dt} \left(\frac{k^t}{t!}\right)=\frac{k^t (\log (k)-H_t+\gamma)}{\Gamma(t+1)} 
\end{equation}
where $\gamma$ is the Euler-Mascheroni constant, $H_t$ is the $t$-th Harmonic number, and $\Gamma$ is the standard $\Gamma$ function. Equating this to 0 we see that there's a single root:
\begin{eqnarray}
    \log(k)-\gamma&=&H_t\\
    &\approx& \gamma+\log(t)
\end{eqnarray}
So picking $t\approx e^{-2\gamma} k$ is the largest term. $e^{-2\gamma}\approx 0.315$. In the case of a truncated series, if the truncation includes fewer than this many terms, the largest term is the last one.
\end{proof}
\fi

In this appendix we show that for temperatures above $\omega_0/\log N$, the population within the Dicke subspace is minimal, as well as the population in the low-excitation manifolds being minimal. These will combine into our argument for the number of angular momentum subspaces that must be considered for a given temperature in order to capture the majority of the population in the system.

For our analysis we neglect the Lamb shifts, and at the end argue why they do not change the results, except in high order. To begin with, consider the partial partition function $\tilde{Z}$ composed of only considering the contributions to the partition function of the two angular momentum subspaces with the largest value ($j=N/2$ and $j=N/2-1$). This partial partition function is given by:
\begin{eqnarray}
    \tilde{Z}&=&(1+2e^{-\beta \omega_0}+3e^{-2\beta \omega_0}+\cdots)+N(e^{-\beta\omega_0}+2e^{-2\beta\omega_0}+3e^{-3\beta\omega_0}+\cdots)\\
    &=& \sum_{j=0}^{"\infty "} (j+1)e^{-j\beta\omega_0}+N\sum_{j=0}^{"\infty "} je^{-j\beta\omega_0},
\end{eqnarray}
where $"\infty "$ signifies that these sums are not truly taken to infinity, but rather to $N+1$ and $N$ respectively, then the remaining terms are $N+1$ times the exponential factor and $N$ times the exponential factor, however, since $N$ is large we take it as $\infty$ for these weighted geometric sums. Note that $\tilde{Z}<Z_0<Z_{total}$, and so we may bound the probability of being in the $N/2$ subspace by:
\begin{equation}
    p(Dicke)< \frac{\sum_{j=0}^{"\infty "} (j+1)e^{-j\beta\omega_0}}{\tilde{Z}}.
\end{equation}
This will be negligible when $p(Dicke)\ll 1$, which occurs when the second term in $\tilde{Z}$ dominates. This is promised when:
\begin{equation}
    \sum_{j=0}^{"\infty "} e^{-j\beta \omega_0}[(j+1)-Ne^{-\beta\omega_0}j]<0.
\end{equation}
This can be guaranteed if $Ne^{-\beta\omega_0}>1+1/j$ for all $j$, which provides $T_c=\omega_0/\log N$. This means that for $T>T_c$ the population in the Dicke subspace must be negligible. As discussed in the main text, while in the rotating-wave approximation regime, the Lamb shifts are small enough to not alter the asymptotic behavior in these expressions, as $g_0\max \lambda\ll \omega_0$ is required.


Consider the case where the system is cold enough that the maximum number of excitations appreciably excited in the system is given by some $k_{max} \ll N$. With this condition, we can approximate the binomial term in the degeneracy function using the following expression \cite{das2016brief}:
    \begin{equation}
        \binom{N}{k} = (1+o(1))\frac{N^k}{k!}.
    \end{equation}
    
    Motivated by our restriction of the number of excitations to be small and bounded, we consider the degeneracy of the angular momentum ground state with $k$ excitations. Given $N$ spins, the ground state of the subspace with angular momentum $j$ has $k=N/2-j$ spin excitations. 














    Then, the degeneracy of the ground state with $k$ excitations is given by
    \begin{equation}
        d_g(k) = \frac{N - (2k-1)}{N - (k-1)}\binom{N}{N-k}.
    \end{equation}
    Rearranging this term and utilizing the approximation of the binomial coefficient, we find that
    \begin{equation}
        d_g(k) = \frac{1}{1- \frac{k}{N - (2k-1)}} \frac{N^k}{k!} (1+o(1)).
    \end{equation}
    Since $k \ll N$, we can safely approximate the degeneracy of the ground state with $k$ excitations as  $d_g(k) \approx N^k/k!$.
    
    The dimension of the ground state subspace is 1, and has a zero energy splitting. This implies that the ground state contributions to the partition function are given by,
    \begin{eqnarray}
       \nonumber Z_g(\beta)&=&\sum_{k=0}^{k_{\max}}e^{-\beta k\omega_0} \frac{N^k}{k!}\\
        &=&\sum_{k=0}^{k_{\max}} \frac{1}{k!}(e^{-\beta\omega_0}e^{\log N})^k\\
        &\approx & e^{e^{\log N-\beta\omega_0}} 
    \end{eqnarray}
    \if{false} Consider the ratio of sequential terms $m$ and $m+1$ in this series above:
    \begin{equation}
        \frac{1}{m+1}e^{\log N-\beta\omega_0},
    \end{equation}
    when $\log N<\beta \omega_0$ the sum will converge rapidly. However, when $\log N>\beta\omega_0$ the sum will converge far slower.
    \fi
    
    Then the population in these states is given by:
    
    \begin{equation}
        \frac{e^{e^{\log N-\beta\omega_0}} }{(1-e^{-\beta\omega_0})^{-1}(e^{-\beta\omega_0/2}+e^{\beta\omega_0/2})^N},
    \end{equation}
    which is minimal for $T>T_c$.
    
\if{false}    In order to satisfy the assumption that the thermal state occupation is contained below some $k_{max} \ll N$, it must be the case that the contributions of higher excitation spaces is decreasing. In order to satisfy this condition, it must be the case that the prior series have a ratio less than 1. We can use this to compute a benchmark for determining whether or not the ground state is populated, where we re-introduce $\hbar$:
    \begin{equation}
        e^{-\hbar\beta\omega_0} N < 1 \implies T < \frac{\hbar\omega_0}{k_b \log N}.
    \end{equation}
    
    \fi
    Given these properties for temperatures above $\omega_0/\log N$, we then define the \textit{ground state cutoff temperature} to be
    \begin{equation}
        T_c =  \frac{\hbar\omega_0}{k_b \log N}.
    \end{equation}
    Above this temperature the population in the Dicke subspace and in the lowest energy states will be minimal, so the more full structure of the system must be considered.

    \subsection{Algorithmic scaling}\label{ndepen}
    
    
    \begin{theorem}\label{size}
    For an arbitrary constant error, $\delta$, the runtime to compute a parameter which is a function of $j$ and $k$, and subexponential in $j$, is $\Theta (\sqrt{N})$.
    \end{theorem}
    
    
    
    
    To prove the dependence on $N$, we begin by noting that $d_j$ has a unique maximum at $j^*$, about which the ratio of adjacent terms is $1+O(N^{-3/2})$ \cite{gunderman2022lamb}. 
    Generally the ratio between two adjacent angular momentum subspaces is given by:
    \begin{eqnarray}
        \frac{d_j}{d_{j+1}}=\frac{2j+1}{2j+3}\cdot \frac{N/2+j+2}{N/2-j}.
    \end{eqnarray}
    Differentiating this, we see that there are no valid $j$ values that extremize this expression, and it begins with a ratio of $N$ for $j=N/2-1$ and tends to $\frac{1}{3}$ for $j=0$. Therefore, this ratio is monotonic, and bijective, over the domain of possible $j$ values, having a ratio between $1/3$ and $N-1$. 
    
    Next we will need the following lemma, a well-known result of geometric series:
    
    
    \begin{lemma}
    The number of terms that must be summed such that the geometric series with ratio $(1+\epsilon)^{-1}$ has a fractional error of $\delta$ is given by $O(\log(\delta^{-1})\epsilon^{-1})$, for $\delta\ll 1$ and $\epsilon\ll 1$. 
    \end{lemma}
    
    \begin{proof}
    We begin by noting that:
    \begin{equation}
        \sum_{i=0}^\infty (1+\epsilon)^{-i}=\frac{1}{1-\frac{1}{1+\epsilon}}=1+\epsilon^{-1}
    \end{equation}
    So we wish to find $t$ such that:
    \begin{equation}
        \sum_{i=0}^t (1+\epsilon)^{-i}\geq (1-\delta)(1+\epsilon^{-1})
    \end{equation}

    Note that this sum is equivalent to:
    \begin{equation}
        \sum_{i=0}^t 2^{-i\log (1+\epsilon)}\geq (1-\delta)(1+\epsilon^{-1})
    \end{equation}
    Then this is given by:
    \begin{eqnarray}
    t &=& O(\log_{(1+\epsilon)^{-1}}(\delta))\\
    &=& O(-\log (\delta)/\log(1+\epsilon))\\
    &\approx & O(\log(\delta^{-1})/\epsilon)
    \end{eqnarray}
    \end{proof}
    
    \begin{corollary}
    For a series of real numbers where the terms monotonically decrease by at least a constant ratio, $1+\epsilon$, the number of terms needed to approximate their sum to relative error $\delta$ is $O(\log(\delta^{-1})/\epsilon)$.
    \end{corollary}
    
    \begin{proof}
    This is the same as the above, except now we allow the ratio to possibly become more severe than $1+\epsilon$, as our summations will have.
    \end{proof}
    
    
    
 \begin{proof}[Proof of Theorem \ref{size}] Fix $k$. Then the set of possible $j$ values falls in the range $[\max\{j_{min},\frac{N}{2}-k\},\frac{N}{2}]$. The ratio of populations between angular momentum subspace $j$ and $j+1$, with $j$ in the range aforementioned, is then given by:
    \begin{equation}
        r=\frac{d_j}{d_{j+1}}\frac{\mu}{\mu+1},
    \end{equation}
    where $\mu$ is the number of states available for subspace $j$ (if $j$ is not the smallest value permitted for a given $k$ choice, one could have say the lower $3$ levels for $j$ and the lower $4$ for $j+1$).
    
    Generally $\frac{d_j}{d_{j+1}}\in (1/3, N-1]$, but suppose we restrict ourselves to $j$ such that $\frac{d_j}{d_{j+1}}\in (1+N^{-1/2},N-1]$, which occurs for $j=\Theta (j^*)=\Theta(\sqrt{N})$. Now, we wish for $r=1+\Omega(N^{-1/2})$ to be satisfied. Note that $\mu/(\mu+1)$ begins at $\frac{1}{2}$ and tends to $1$, so if we can ensure $\mu/(\mu+1)=1-O(N^{-1/2})$ then $r$ will satisfy our desired value. Once $\mu\approx \sqrt{N}$, then:
    \begin{eqnarray}
        \frac{\mu}{\mu+1}&=& \frac{\sqrt{N}}{\sqrt{N}+1}\\
        &=& \frac{1}{1+N^{-1/2}}\\
        &\approx & 1-N^{-1/2}
    \end{eqnarray}
    Putting this together, it means that to obtain the sum for some fixed $k$ value, assuming $d_j/d_{j+1}=1+N^{-1/2}$, we sum from $\max\{ j_{min},\frac{N}{2}-k\}$ for $O(\sqrt{N})$ terms, then for another $O(\sqrt{N})$ terms whereby these terms converge to within a $\delta$ error within $O(\sqrt{N})$ terms from the prior lemma and corollary. Lastly, note that $d_j/d_{j+1}=1+N^{-1/2}$ for $j=O(\sqrt{N})$ so when the ratio is below this, there will be at most $O(\sqrt{N})$ more $j$ values to sum. Putting these together, we obtain a $\delta$ error, for any constant $\delta$, with only $O(\sqrt{N})$ terms summed.
    
    Lastly, for the matching lower-bound, we can use the result from our prior work, \cite{gunderman2022lamb}, where we showed that $\Omega(\sqrt{N})$ angular momentum values contain the far majority of the degeneracies in the system. This provides a lower-bound for higher excitations whereby most of the degeneracies are available. Together this means that this procedure is effectively optimal in the spin system size, $N$, with a runtime of $\Theta(\sqrt{N})$.
    \end{proof}

\if{false}    
    \subsection{$T$ dependence}
    
    \begin{equation}
        e^{-\beta\omega_0} P_k+D_j\leq P_{k+1}\leq e^{-\beta\omega_0} 2 P_k+D_j
    \end{equation}
    Leads to $O(T)$ bound, although requires caution since this is a limit of large $T$ and you can't get that high without breaking RWA.
    
        The ratio in the populations in angular momentum subspace $j$ and $j+1$ is given by:
    \begin{equation}
    \frac{d_j}{d_{j+1}}\frac{\mu}{\mu+1} e^{-\beta\omega_0}
    \end{equation}
    where $\mu$ ranges from $1$ to $2j+1$. If we require this ratio to be at least $2$ (or any value separated from $1$) then the ratio between $j$ and $j+1$ subspaces will be $2$, and the ratio of $j$ and $j+2$ will be $4$, and so forth, doubling each time. This occurs for:
    \begin{equation}
        T> \frac{\omega_0}{\log (\frac{1}{2}\frac{\mu}{\mu+1}(d_j/d_{j+1}))}.
    \end{equation}
    As shorter hand we can take the $\mu$ limit of $1$, reducing this to $\omega_0/\log(d_j/(4d_{j+1}))$. Note that for $j$ such that $d_j<4d_{j+1}$, there is no temperature for which these states can have a larger population.
    
\fi
    
    

    

\subsection{Monotonicity of Population Ratio Temperature}\label{ratio}
    Here we show that the temperature needed to have more population in a low $j$ value than the Dicke value slowly rises to a limiting temperature requirement of $\omega_0/(2\log 2)$. Note that for small $j$, the following approximation may be made:
    \begin{eqnarray}
        d_j &=& {N\choose N/2+j}\cdot\frac{2j+1}{N/2+j+1}\\
        &\approx & \frac{2^N}{\sqrt{N\pi/2}} e^{-(2j)^2/(2N)}\cdot\frac{2j+1}{N/2+j+1}
    \end{eqnarray}
    where the approximation is due to \cite{spencer2014asymptopia}. This is to leading order given by $2^N$ and so:
    \begin{equation}
        \lim_{N\rightarrow \infty} \frac{\omega_0 (N/2-j)}{\log d_j}=\frac{\omega_0}{2\log 2}.
    \end{equation}
    Lastly the expression $\omega_0 (N/2-j)/\log d_j$ has a unique root, the above, and is always positive up to $j^*$. After this temperature the other angular momentum subspaces are populated according to the main text.
    
\if{false}  
     Setting the Dicke ground state energy to $0$, by shifting the Hamiltonian, we can see that the population of the Dicke ground state is given by $1/Z(\beta)$. Now, consider the $k$--th partial partition function, formed by summing over the first $k$ excitation subspaces only:
    \begin{equation}
        Z_k(\beta) = \sum_{k'=0}^{k} \sum_{j=j_0(k')}^{N/2}\bigg( d_j e^{-\beta k' \omega_0} \sum_{\lambda\in\Lambda(j,k')} e^{-\beta \lambda g}\bigg).
    \end{equation}
    Since each term in the representation of the partition function is positive, we have that $Z_k(\beta) < Z(\beta)$, which implies that the population of the Dicke ground state must be bounded above by $1/Z_k(\beta)$ for all $k$. We have that we can compute the eigenvalues exactly for any $j,k$ subspace with dimension less than or equal to 9, which means we could in principle give exact values for $Z_k(\beta)$ up to $k = 9$ total excitations \cite{PREV PAPER}.
\fi
    
    
    
    

\subsection{Thermal Expectation Expansions}\label{pertexp}

Here we show the error in truncating the expansion for the partition function. The partition function can be used to compute averages directly related to the partition function and also by taking the expectation of the partition function we can compute other averages.

The exact expression for the partition function for our system of interest is given by:
\begin{eqnarray}
    Z&=&\sum_{j,k,\ \lambda \in \Lambda(j,k)} d_j e^{-\beta (k\omega_0 \openone+g_0L(j,k))}\\
    &=&\sum_k e^{-\beta k\omega_0 \openone} \sum_j d_j \sum_{\lambda\in \Lambda(j,k)} e^{-\beta g_0 \lambda}.
\end{eqnarray}
We focus on truncating the inner sum. We set the inner summation as $\tau(j,k,\beta)$. Then $Z$ is:
\begin{equation}\label{tauz}
    Z=\sum_k e^{-\beta k\omega_0} \tau(j,k,\beta).
\end{equation}
We may then expand $\tau$ in a power series. This will allow us to find the bounded error in a truncated expansion. Before beginning, we note that we may assume we have $\|\Lambda(j,k)\|_\infty\ll \omega_0/g_0$, as in this regime the rotating-wave approximation holds, which is needed to have the pair of good quantum numbers utilized in the Tavis--Cummings Hamiltonian. 

We begin with an expansion of the innermost sum in $\tau$:
\begin{equation}
    \sum_{\lambda\in \Lambda(j,k)} e^{-\beta g_0 \lambda}\approx \sum_\lambda 1-\beta g_0 \sum_\lambda \lambda + \frac{(\beta g_0)^2}{2}\sum_{\lambda} \lambda^2-\frac{(\beta g_0)^3}{3!}\sum_{\lambda} \lambda^3+\frac{(\beta g_0)^4}{4!}\sum_{\lambda} \lambda^4+\ldots
\end{equation}
As we showed in our prior work all odd moments of the eigenvalues are zero since eigenvalues come in positive-negative pairs \cite{gunderman2022lamb}. The number of eigenvalues is $|\mathcal{B}_{j,k}|$, so this expansion is then given by:
\begin{equation}
    |\mathcal{B}_{j,k}|(1+\frac{(\beta g_0)^2}{2}Var(\Lambda(j,k)))+O((\beta g_0)^4 \sum_\lambda \lambda^4)
\end{equation}
This expansion says that the partition function can be considered as a summation over the states of the uncoupled system plus contributions from second order transitions from the coupling Hamiltonian, and the fourth order, and so forth. We wish for this to truncate, so let's consider what temperatures we can promise the higher order terms to be negligible. Note that $(\frac{1}{2}(\beta g_0)^2\sum \lambda^2)^2>\frac{1}{4}(\beta g_0)^4\sum \lambda^4$, and generally $(\frac{1}{2}(\beta g_0)^2\sum \lambda^2)^t>\frac{1}{2^t}(\beta g_0)^{2t}\sum \lambda^{2t}$, and so if $\frac{1}{2}(\beta g_0)^2\sum \lambda^2<1$ then the higher order terms are all smaller and can be neglected if sufficiently small. This means that the first order correction is dominant so long as $|\mathcal{B}_{j,k}|\frac{(\beta g_0)^2}{2}Var(\Lambda(j,k))\ll 1$. 

\if{false}

We begin by upper bounding the sums of powers of the Lamb shifts:
\begin{equation}
    \sum_\lambda \lambda^{2t} \leq |\mathcal{B}_{j,k}| \|\Lambda (j,k)\|_\infty^{2t}\ll |\mathcal{B}_{j,k}| (\omega_0/g_0)^{2t},
\end{equation}
where the last expression is from being in the weak-coupling regime and being allowed to make the rotating-wave approximation. The coefficient for this term is $\frac{(\beta g_0)^{2t}}{(2t)!}$ and so the additional powers are able to be dropped so long as $\beta \omega_0<1$. This is around $0.8$ Kelvin for $\omega_0=10$ GHz. This cutoff can be reduced further by utilizing the expression for $\|\Lambda(j,k)\|_\infty$ from our prior paper--generally we just need $T>g_0\max_{j,k}\|\Lambda(j,k)\|_\infty$, which (if $k\ll N$) is roughly $T\gtrsim g_0N\sqrt{k}$.

Assuming we have $\beta\omega_0<1$, we have:\fi

Assuming we satisfy this condition we have:
\begin{equation}\label{orderone}
    \tau(j,k,\beta)\approx \sum_j d_j|\mathcal{B}_{j,k}|(1+\frac{(\beta g_0)^2}{2} Var(\Lambda(j,k))).
\end{equation}
Including this back in the summation over $k$, we have:
\begin{equation}
    Z_{total}=Z_0+Z_{pert}
\end{equation}
where $Z_0$ is the partition function when $g_0=0$ (just a collection of spins uncoupled with a cavity), and $Z_{pert}$ is the contribution due to the interaction between the cavity and spins.

This also means that the normalization factor, $Z^{-1}$, can be expanded as:
\begin{eqnarray}
    \frac{1}{Z_{total}}&=&\frac{1}{Z_0+Z_{pert}}\\
    &\approx& \frac{1}{Z_0}(1-\frac{Z_{pert}}{Z_0}),
\end{eqnarray}
where the last expression comes from utilizing a truncated geometric series expansion.


\subsection{Average Internal Energy Shift}\label{epert}
We may now compute the average energy of this system, as well as the associated shift induced due to the coupling:
\begin{equation}
    \langle E\rangle=-\frac{1}{Z}\frac{\partial}{\partial \beta}Z.
\end{equation}
Applying this to our expression for $Z_{total}$ in (\ref{tauz}) and applying the chain rule for differentiating, then using the first order correction given by (\ref{orderone}) we obtain:
\begin{equation}
    = \langle k \rangle_{Z} \omega_0-\frac{2}{\beta}\frac{Z_{pert}}{Z}.
\end{equation}
We may now break this into perturbation powers to obtain a more physically meaningful expression. For each line we plug in our expansions and drop higher order terms:
\begin{eqnarray}
    &=& \langle k\rangle_{Z_0}Z_0 \omega_0 (\frac{1}{Z_0}(1-\frac{Z_{pert}}{Z_0}))+[Z_{pert}\langle k\rangle_{Z_{pert}} \omega_0-\frac{2}{\beta}Z_{pert}](\frac{1}{Z_0}(1-\frac{Z_{pert}}{Z_0}))\\
    &= & \langle k\rangle _{Z_0} \omega_0 (1-\frac{Z_{pert}}{Z_0})+\langle k\rangle_{Z_{pert}}\frac{Z_{pert}}{Z_0} \omega_0-\frac{2}{\beta}\frac{Z_{pert}}{Z_0}\\
    &=& \langle k\rangle _{Z_0}\omega_0+((\langle k\rangle_{Z_{pert}}-\langle k\rangle_{Z_0}) \omega_0-\frac{2}{\beta})\frac{Z_{pert}}{Z_0}
\end{eqnarray}
This means that the shift in the average internal energy is given by $\Delta\langle E\rangle=((\langle k\rangle_{Z_{pert}}-\langle k\rangle_{Z_0}) \omega_0-\frac{2}{\beta})\frac{Z_{pert}}{Z_0}$. This expression can be considered as the average excitation difference between the perturbed portion of the system and the uncoupled system, $(\langle k\rangle_{Z_{pert}}-\langle k\rangle_{Z_0})$, with the addition of a thermally scaling factor of the perturbation expansion variable $-\frac{2}{\beta} \frac{Z_{pert}}{Z_0}$, and will always be negative, so long as the expansions are valid. This agrees with intuition as having the lower-energy Lamb shifted energy states to place some population into can only decrease the average energy of the system.

\if{false}
\subsection{Helmholtz Free Energy Shift}
The same computation can be done for the Helmholtz free energy. The expression for the Helmholtz free energy in terms of the partition function provides:
\begin{eqnarray}
    -\beta \langle A\rangle&=& \log Z\\
    &=& \log (Z_0+Z_{pert})\\
    &=& \log Z_0 + \log (1+Z_{pert}/Z_0)\\
    &=& \log Z_0+\frac{Z_{pert}}{Z_0}+O((\frac{Z_{pert}}{Z_0})^2)
\end{eqnarray}
This then provides $\Delta \langle A\rangle=-\frac{1}{\beta}\frac{Z_{pert}}{Z_0}$ as the shift induced in the Helmholtz free energy.
\fi
\if{false}

\subsection{Upper bound on $Z_{pert}/Z_0$}

As we already have an exact expression for $Z_0$, we focus on upper bounding $Z_{pert}$. We recall the expression for $Z_{pert}$ from above:
\begin{equation}
    Z_{pert} :=\sum_k e^{-\beta k\omega_0 \openone} \sum_j d_j [\frac{(\beta g_0)^2}{2}|\mathcal{B}_{j,k}|Var(\Lambda(j,k))+O((\beta g_0)^4)].
\end{equation}
We will only derive an asymptotic expression for the upper bound and so will drop the higher order term in the above expression. Noticing that the expression for $Var(\Lambda(j,k))$ is a sub-exponential function in $j$ and $k$ we may utilize the region of strong support for $d_j$ to reduce the summation from being over all $j$ to just $j$ roughly in $[0,2j^*]$, where $j^*$ is the maximally degenerate collective angular momentum subspace at $\sqrt{N}/2$ to leading order.

Since we are only finding an upper bound we note that we may use $d_{j^*}$ for all $d_j$ in this region as $d_{j^*}\geq d_j$ by definition. We now utilize the prior result that $d_{j^*}=O(2^N/N)$. Additionally, since we are above $T_c$ we are strongly biased toward values of $k$ whereby $j^*$ is populated due to the massive degeneracy of this regime, and so we may take $k\geq N/2-2j^*$. Using these approximations we have reduced the problem to:
\begin{equation}
    Z_{pert}\leq O\left((\beta g_0)^2\sum_{k=N/2-\sqrt{N}}^\infty e^{-\beta k\omega_0} \sum_{j=0}^{\sqrt{N}}\frac{2^N}{N} |\mathcal{B}_{j,k}|Var(\Lambda(j,k))\right).
\end{equation}
The next challenge is to simplify the argument of the innermost summation. The expressions for $|\mathcal{B}_{j,k}|$ and $Var(\Lambda(j,k))$ have a conditional structure based on whether all spin $J_z$ values have been saturated. We denote the variance for $k$ below this switch with $Var_1(\Lambda(j,k))$ and for $k$ above we use $Var_2(\Lambda(j,k)$. The same notation is used for $|\mathcal{B}_{j,k}|$. This switching happens when:
\begin{equation}
    k+j-\frac{N}{2}\geq 2j,
\end{equation}
or equivalently $k\geq \frac{N}{2}+j$. Since the variance function is continuous even at this value and we are only interested in the leading term, we will take the approximation that the dominant contributing factors are for the case above this transition so set $Var(\Lambda(j,k))=Var_2(\Lambda(j,k))$ and $|\mathcal{B}_{j,k}|=2j+1$. Plugging this into the expression for the variance, we obtain as the summand in the innermost sum:
\begin{eqnarray}
    \nonumber |\mathcal{B}_{j,k}|\operatorname{Var}(\Lambda(j,k)) &\approx & \frac{1}{2}\abs{\mathcal{B}_{j,k}}^4 -\frac{1}{3}\abs{\mathcal{B}_{j,k}}^3(2k'+4j+7)\\
    & &+ \abs{\mathcal{B}_{j,k}}^2(2jk'+2k'+4j+7/2) - \frac{1}{3}|\mathcal{B}_{j,k}|(6jk' + 8j + 4k' + 5)\\
    &=& \frac{1}{2}(2j+1)^4 -\frac{1}{3}(2j+1)^3(6j+2k-N+7)\\
    & &+ (2j+1)^2(2j^2+2jk-jN+6j+2k-N+7/2)\\
    & &- \frac{1}{3}(2j+1)(6j^2+6jk-3jN+12j+4k-2N+5)\\
    &=& (\frac{8}{3}j^3+4j^2+2j+\frac{1}{3}-4j^3+8j^2-5j-1+2j^2+\frac{7}{3}j+\frac{2}{3})N\\
    & & + (-\frac{16}{3}j^3-8j^2-4j-\frac{2}{3}+8j^3+16j^2+10j+2-4j^2-\frac{14}{3}j-\frac{4}{3})k\\
    & & + [-8j^4-\frac{80}{3}j^3-28j^2-12j-\frac{11}{6}]+[8j^4+32j^3+40j^2+20j+\frac{7}{2}]\\
    & &+[-4j^3-10j^2-\frac{22}{3}j-\frac{5}{3}]
\end{eqnarray}
where we have used $k':=k-k_0(j)=k+j-\frac{N}{2}$. Simplifying this expression we have:
\begin{equation}
    (-\frac{4}{3}j^3+14j^2-\frac{2}{3}j)N+(\frac{8}{3} j^3+4j^2+\frac{4}{3}j)k+(\frac{4}{3}j^3+2j^2+\frac{2}{3}j).
\end{equation}
We will drop the lower order terms to take this as $(-\frac{4}{3}j^3+14j^2)N+(\frac{8}{3} j^3+4j^2)k$. 

We now return to the original expression for $Z_{pert}$ and perform the innermost summation:
\begin{eqnarray}
    \sum_{j=0}^{\sqrt{N}} |\mathcal{B}_{j,k}|Var(\Lambda(j,k))&\approx&   \sum_{j=0}^{\sqrt{N}} (-\frac{4}{3}j^3+14j^2)N+(\frac{8}{3} j^3+4j^2)k\\
    &=& -\frac{1}{3}(\sqrt{N}+1)N^{3/2}(N-13\sqrt{N}-7)+\frac{2}{3}k(\sqrt{N}+1)\sqrt{N}(N+3\sqrt{N}+1)\\
    &\approx & -\frac{1}{3}N^{2}(N-13\sqrt{N})+\frac{2}{3}kN(N+3\sqrt{N}).
\end{eqnarray}
Our expression is now able to be computed:
\begin{eqnarray}
    Z_{pert}&=& O\left((\beta g_0)^2\sum_{k=N/2-\sqrt{N}}^\infty e^{-\beta k\omega_0}\frac{2^N}{N}[2kN(N+3\sqrt{N})-N^{2}(N-13\sqrt{N})]\right)\\
    &=& O\left((\beta g_0)^2\sum_{\tilde{k}=0}^\infty e^{-\beta (\tilde{k}+N/2-\sqrt{N})\omega_0}\frac{2^N}{N}[2(\tilde{k}+N/2-\sqrt{N})N(N+3\sqrt{N})-N^{2}(N-13\sqrt{N})] \right)\\
    &\approx & O\left((\beta g_0)^2\sum_{\tilde{k}=0}^\infty e^{-\beta (\tilde{k}+N/2-\sqrt{N})\omega_0}2^N[2\tilde{k} N+14N^{3/2}] \right)\\
    &=& O\left((\beta g_0)^22^Ne^{-\beta(N/2-\sqrt{N})\omega_0}[2N\frac{e^{\beta\omega_0}}{(e^{\beta\omega_0}-1)^2}+14N^{3/2}\frac{1}{1-e^{-\beta\omega_0}}]\right)
\end{eqnarray}
This expression alone provides the thermal state's variance in the Lamb shifts. In particular this is an expression for the degeneracy averaged variance as a function of particle number, $N$, and temperature, $T$ (given in terms of $\beta=T^{-1}$). We may finally bound the ratio:
\begin{eqnarray}
    \frac{Z_{pert}}{Z_0}&=&O\left(\frac{(\beta g_0)^22^Ne^{-\beta(N/2-\sqrt{N})\omega_0}[2N\frac{e^{\beta\omega_0}}{(e^{\beta\omega_0}-1)^2}+14N^{3/2}\frac{1}{1-e^{-\beta\omega_0}}]}{(1-e^{-\beta\omega_0})^{-1}(e^{-\beta\omega_0/2}+e^{\beta\omega_0/2})^N} \right)\\
    &=&O\left((\beta g_0)^2\left(\frac{2 e^{-\beta\omega_0 /2}}{e^{-\beta\omega_0/2}+e^{\beta\omega_0/2}}\right)^N\frac{[N\frac{e^{\beta\omega_0}}{(e^{\beta\omega_0}-1)^2}+7N^{3/2}\frac{1}{1-e^{-\beta\omega_0}}]}{(1-e^{-\beta\omega_0})^{-1}} \right)\\
    &=&O\left((\beta g_0)^2\left(\frac{2 e^{-\beta\omega_0 /2}}{e^{-\beta\omega_0/2}+e^{\beta\omega_0/2}}\right)^N[N(e^{\beta\omega_0}-1)^{-1}+7N^{3/2}] \right)\\
    &=& O\left((\beta g_0)^2\left(\frac{2}{1+e^{\beta\omega_0}}\right)^N[7N^{3/2}+N(e^{\beta\omega_0}-1)^{-1}] \right).
\end{eqnarray}
When in the regime that $\beta\omega_0\ll 1$ as well, then this is simplified to
\begin{equation}
    O\left((\beta g_0)^2\left(1-\frac{1}{2}\beta\omega_0\right)^N[7N^{3/2}+\frac{N}{\beta\omega_0}] \right).
\end{equation}
These expressions for the ratio of the partition functions provides a bound on the shift in the Helmholtz free energy in the presence of the coupling and is one element in the value of the average internal energy shift.

\fi

\subsection{Rapid Computation of Expectations with Certain Invariances}

Here we outline the analytical results implications for bounding expectations of observables of thermal states of the Tavis--Cummings model as well as the runtime improvements for numerical methods evaluating these expectations. For this we focus on observables with particular invariances, and list them as cases. For each case, $\sum$ means that the variable(s) must be summed over, whereas $\int$ means that that variable is either invariant with respect to that variable or can analytically have its dependence removed. In what follows $k$ is the number of excitations, $j$ is the angular momentum subspace value, and $m$ is the spin value within that angular momentum subspace value.

\begin{enumerate}
    \item $\sum j,k,m$: in this case $\mathcal{O}(k,j,m)$ depends on all three parameters. In this case the full summation must be considered which makes the problem nearly intractable and has runtime of $O(TN^2)$. If one can ensure that the region of strong support is satisfied, then this can be reduced to $O(TN)$.
    \item $\sum k,j \int m$: $\mathcal{O}(j,k)$ is the observable. The probabilities are only functions of $j$ and $k$ in this case so $Var(\Lambda(j,k))$ may be used in the expectation computation. Additionally, so long as $\mathcal{O}(j,k)$ is sub-exponential in $j$ we may use the region of strong support within the degeneracies to simplify the expression and improve the runtime of computing this expectation. This results in $O(T\sqrt{N})$ time.
    \item $\sum k, \int j,m$: in this case the observable only depends on the number of excitations $\mathcal{O}(k)$. Since the dependency on $j$ is removed and can be tabulated, this loses all dependency on $N$ and requires only $O(T)$ time.
    \item $\sum j, \int k,m$: this follows similar reasoning as the above. The temperature dependence can be removed and summed over independently, meaning just the spin states matter, taking $O(N)$ unless the region of strong support can be provided, in which case this reduces to $O(\sqrt{N})$.
\end{enumerate}

For observables such as $J_z$ the pair of quantum numbers remain valid and so the trace value is basis independent so we just take the summation over $|k,j,m\rangle$. Some examples of these categories include: $a^\dag a$ and $a+a^\dag$ in (c) and $J_z$ and $J_x$ in (b). In fact, any function of cavity operators ($a$ and $a^\dag$) and collective spin operators ($J_+$, $J_-$, and $J_z$) is covered in case (b), with any function that grows subexponentially in $j$ also being able to exploit the region of strong support. Some examples of subexponentially growing functions include $J_x$, $e^{itJ_z}$, and $e^{(a^\dag a^\dag- aa) it}$.

\if{false}\subsection{Simulation Equations Corrected}

For this section we mirror some of the transformations from \cite{wood_cavity_2014}. Our Hamiltonian with observable is given by:
\begin{equation}
    \omega_0 (a^\dag a +J_z)+g_0 (a^\dag J_-+aJ_+)+\eta J_z,
\end{equation}
where $\eta J_z$ generates the signal we wish to see and $\eta \ll \omega_0$. We move into a rotating frame given by $(\omega_0+\eta) (a^\dag a +J_z)$:
\begin{equation}
    H'(t)= \eta a^\dag a +g_0 (a^\dag J_-+aJ_+).
\end{equation}
Moving to an interaction picture given by $\eta a^\dag a$, this becomes:
\begin{eqnarray}
    H_I(t) &=& g_0 (e^{-i\eta t}a^\dag J_-+e^{i\eta t} J_+a)
\end{eqnarray}
Note that $(e^{-i\eta t}a^\dag J_-+e^{i\eta t} J_+a)$, within values of $j$ and $k$, is very similar to $\Lambda_{j,k}$. In particular the off diagonal entries are scaled by factors of $e^{\pm \eta t}$, and in particular the diagonal entries are all zero for odd powers, and the same as $\Lambda_{j,k}^2$ for even powers, regardless of the value for $\eta$. Since we will be taking the trace of this, and trace is basis invariant, this is equivalent to finding $\langle tr_m (\cos(t g_0 \Lambda_{j,k}))\rangle_{Z_{total}} $. Since we wish to preserve the sinusoidal behavior of this function, we stop simplifications here.

\if{false}
Taylor expanding this to the first nontrivial order gives:
\begin{equation}
    \langle \mathcal{B}_{j,k}(1-\frac{1}{2}g_0^2 Var(\Lambda_{j,k}) t^2)\rangle_{Z_{total}}
\end{equation}
The difference induced by the Lamb shift is then:
\begin{equation}
    \frac{1}{Z_0} (Z_{pert}-\sum_{j,k}\frac{1}{2} g_0^2 \mathcal{B}_{j,k}Var(\Lambda_{j,k}) t^2 Z_0(j,k))
\end{equation}
Slowly step through arguments to take careful account. This is likely mostly right.

where did $\eta$ go?? Do I just need $\eta$ small for the approximations?
\fi
\if{false}

While $[\Lambda_{j,k},\Gamma_{j,k}]\neq 0$, it is the case that $\Gamma_{j,k}\ll \Lambda_{j,k}$. Explicitly we have:
\begin{equation}
    tr(-\Gamma_{j,k}^2)=tr(\Lambda_{j,k}^2)
\end{equation}

We wish to compute:
\begin{equation}
    \langle e^{it (H+\eta J_z)}\rangle_{Z_{total}}
\end{equation}
with $\eta\ll \omega_0$.
We expand this out as:
\begin{equation}
    \langle k,j,m |e^{it (H_0+\eta J_z+H_{pert})}|k,j,m\rangle
\end{equation}
We begin by moving into the interaction frame of $H$, which provides:
\begin{equation}
    H'=\eta e^{it H} J_z e^{-it H}=\eta e^{itH_{pert}}J_z e^{-it H_{pert}}
\end{equation}
Note that $H_0$ cancels out in the above as it commutes with $J_z$. 
\if{false}Next, observe that:
\begin{equation}
    [J_-a^\dag+J_+ a,J_z]=J_-a^\dag -J_+ a,\quad [J_-a^\dag -J_+ a,J_z]=J_-a^\dag +J_+ a.
\end{equation}
This means that the commutation relations... nope...
\fi

While we could commute the operators, instead we take another approach. Since we will be computing the trace for various choices of time, we require that the net change in $m$ is zero. Only $H_{pert}$ can induce a change in $m$, so we need the power of $H_{pert}$ to be even between the left and right terms. [should this instead be the imaginary terms cancelling? since we could have a shift then shift back in the later Taylor series for time evolution. now I'm back again... is basis independence of trace gonna help me here?] Taylor expanding action of $e^{it H_{pert}}$ we obtain:
\begin{equation}
    \langle k,j,m | H' |k,j,m\rangle=\langle k,j,m | \cos(\Lambda_{j,k}t)  J_z  \cos(\Lambda_{j,k}t)+ \sin(\Lambda_{j,k}t) J_z  \sin(\Lambda_{j,k}t) |k,j,m\rangle 
\end{equation}
While $[H_{pert},J_z]\neq 0$, the action of $\Lambda_{j,k}$ in terms of the parameters $k,j,m$ is able to be rearranged. This means that we may reduce the expression to:
\begin{eqnarray}
    &=& \langle k,j,m | (\cos^2(\Lambda_{j,k}t)+\sin^2(\Lambda_{j,k}t))  J_z  |k,j,m\rangle\\
    &=& \langle k,j,m | J_z  |k,j,m\rangle
\end{eqnarray}
The Lamb shift induced signal is then the below expressions averaged over $Z_{pert}$.

\if{false}
[can we reorder this and cancel it out?]

[think to see if there's a way to view this without destroying exp/cyclical behavior. rewrite cos and sin in terms of exp then do like before. This is closer... sin is the problem here since parity and it doesn't keep things fixed]
Since we wish to integrate out the variable $m$, we will will only keep terms that are $O(g_0^2)$:
\begin{equation}
    \langle k,j,m | H' |k,j,m\rangle=\langle k,j,m | (I+\frac{1}{2}(g_0\Lambda_{j,k}t)^2)  J_z  (I+\frac{1}{2}(g_0\Lambda_{j,k}t)^2)+ (g_0\Lambda_{j,k}t)  J_z (g_0\Lambda_{j,k}t) |k,j,m\rangle 
\end{equation}
seems this is the best I can do... which is fair since we required the parameter in the above to be $\ll 1$ for expansion of partition function anyway.
\begin{equation}
    \langle k,j,m | H' |k,j,m\rangle=\langle k,j,m | J_z +2J_z(g_0\Lambda_{j,k}t)^2|k,j,m\rangle 
\end{equation}
$\Lambda^2~O(m)$ doesn't it? Can I just take it as $e^{if(m)}=\cos(f(m))+i\sin(f(m))$?
\fi
\fi
\fi
\if{false}
\subsection{Simulation Equations}\label{evolution}

We now derive expressions for $\langle e^{itJ_z}\rangle $. First, we note that $\langle e^{itJ_z}\rangle_{Z_{pert}}=0$, as in the basis $|k,j,m\rangle$ the interaction Hamiltonian has overlap $0$, and since the trace is basis independent this component will not alter the expectation value. We then only need to consider the uncoupled Hamiltonian's expectation and the normalization from the interaction term.

\if{false}
\begin{eqnarray}
    \langle J_z\rangle&=&\omega_0 \langle a^\dag a J_z\rangle +\omega_0 \langle J_z^2\rangle\\
    &=&\sum_{k,j,m}\langle k,j,m|a^\dag a J_z|k,j,m\rangle+\langle k,j,m|J_z^2|k,j,m\rangle\\
    &=& \sum_{k,j,m}km+d_j m^2\\
    &=& 0+2\sum_{k,j} d_j \sum_{m=0}^j m^2\\
    &=& \sum_k e^{-\beta \omega_0 k} \sum_{j=N/2}^{N/2-k} \frac{1}{3}d_j j(j+1)(2j+1) 
\end{eqnarray}

Whereas we have:
\begin{eqnarray}
    \langle J_z^2\rangle &=& \omega_0\langle a^\dag a J_z^2\rangle+\omega_0\langle J_z^3\rangle\\
    &=&\sum_{k,j,m}\langle k,j,m|a^\dag a J_z^2|k,j,m\rangle+\langle k,j,m|J_z^3|k,j,m\rangle\\
    &=& \sum_{k,j,m}km^2+d_j m^3\\
    &=& 2\sum_{k,j} kd_j \sum_{m=0}^j m^2+0\\
    &=& \sum_k ke^{-\beta \omega_0 k} \sum_{j=N/2}^{N/2-k} \frac{1}{3}d_j j(j+1)(2j+1)
\end{eqnarray}
\fi
The real part of $\langle e^{itJ_z}\rangle$ is due to $a^\dag a$, whereas the imaginary part is due to the Zeeman portion. We may not truncate the Taylor series for $e^{itJ_z}$ without losing the periodicity of the function, so we will instead perform the sums using the basis $|k,j,m\rangle$:
\begin{eqnarray}
    Re(\langle e^{itJ_z}\rangle)&=& \sum_k ke^{-\beta \omega_0 k}\sum_j d_j Re(\sum_{m=-j}^j e^{itm})\\
    &=& \sum_k ke^{-\beta \omega_0 k}\sum_j d_j Re(e^{-ijt}\sum_{m=0}^{2j} e^{itm})\\
    &=& \sum_k ke^{-\beta \omega_0 k}\sum_j d_j Re(\frac{e^{it(2j+1)}-1}{e^{it}-1}e^{-ijt})\\
    &=& \sum_k ke^{-\beta \omega_0 k}\sum_j d_j(\sin(jt)\cot(t/2)+\cos(jt))
\end{eqnarray}
\if{false}
\begin{eqnarray}
    Re(\langle e^{itJ_z}\rangle)&=& \sum_k ke^{-\beta \omega_0 k}\sum_j d_j Re(\sum_{m=-j}^j e^{itm})\\
    &=& 2\sum_k ke^{-\beta \omega_0 k}\sum_j d_j Re(\sum_{m=j_{min}}^j e^{itm})\\
    &=...& 2\sum_k ke^{-\beta \omega_0 k}\sum_j d_j Re(\frac{e^{itj}-1}{e^{it}-1}e^{itj_{min}})\\
    &\approx& \sum_k ke^{-\beta \omega_0 k}\sum_j d_j Re((e^{itj}-1)(e^{-it}-1)\frac{e^{itj_{min}}}{1-\cos(t)})\\
    &\approx& \sum_k ke^{-\beta \omega_0 k}\sum_j d_j \frac{1}{1-\cos(t)} Re(1-\cos(t)-\cos(jt)+\cos((j-1)t))\\
    &\approx& \sum_k ke^{-\beta \omega_0 k}\sum_j d_j \frac{1}{1-\cos(t)} (1-\cos(t)+2\sin(t/2)\sin((j-1/2)t))\\
    &\approx& \sum_k ke^{-\beta \omega_0 k}\sum_j d_j (1-\csc(t/2)\sin((j-1/2)t))
\end{eqnarray}
\fi

The other expression is then given by:
\begin{eqnarray}
    Im(\langle e^{itJ_z}\rangle)&=& \sum_k e^{-\beta \omega_0 k}\sum_j d_j Im(\sum_{m=-j}^j e^{itm}m)\\
    &=& \sum_k e^{-\beta \omega_0 k}\sum_j d_j Im(e^{-ijt}\sum_{m=0}^{2j} e^{itm}(m-j))\\
    &=& \sum_k e^{-\beta \omega_0 k}\sum_j d_j 2\sin^2(t/2)((j+1)\sin(jt)-j\sin((j+1)t))
\end{eqnarray}
\if{false}
\begin{eqnarray}
    Im(\langle e^{itJ_z}\rangle)&=& \sum_k e^{-\beta \omega_0 k}\sum_j d_j Im(\sum_{m=-j}^j e^{itm}m)\\
    &\approx& \sum_k e^{-\beta \omega_0 k}\sum_j d_j Im(\sum_{m=j_{min}}^j e^{itm}m)\\
    &\approx& \sum_k e^{-\beta \omega_0 k}\sum_j d_j Im(\frac{1}{i}\frac{d}{dt}[\frac{e^{itj}-1}{e^{it}-1}e^{itj_{min}}])\\
    &\approx& \sum_k e^{-\beta \omega_0 k}\sum_j d_j Im(\frac{1}{i}\frac{d}{dt}[(e^{itj}-1)(e^{-it}-1)\frac{e^{itj_{min}}}{1-\cos(t)}])\\
    &\Rightarrow & \frac{1}{(1-\cos(t))^2}[(1-\cos(t))(\sin(t)+\sin((j-1)t)-j(\sin((j-1)t)-\sin(jt)))\\
    & & +\sin(t)(1-\cos(t)-\cos(jt)+\cos((j-1)t))]\\
    &= & \frac{1}{(1-\cos(t))}[j\sin(jt)+\sin(t)-(j-1)\sin((j-1)t)\\
    & & +\sin(t)(1-\csc(t/2)\sin((j-1/2)t))]\\
    &=& \frac{1}{1-\cos(t)}[j\sin(jt)+2\sin(t)-(j-1)\sin((j-1)t)-2\cos(t/2)\sin((j-1/2)t)]\\
    &=& \frac{1}{1-\cos(t)}((j-1)\sin(jt)-j\sin((j-1)t)+2\sin(t))
\end{eqnarray}
\fi
\fi
\subsection{Perturbations and Inhomogeneities}\label{perts}

The following theorem is used, it is a direct consequence from the Bauer-Fike Theorem and Weyl's inequality \cite{bauer1960norms,eisenstat1998three}:
\begin{theorem}\label{inhombound}
Let $H$ and $H+\delta H$ be Hermitian matrices with eigenvalues $\{\lambda_1,\lambda_2,\ldots\}$ and $\{\mu_1,\mu_2,\ldots\}$, respectively, in non-increasing order. Then for all eigenvalues, with index $i$, we have:
\begin{equation}
    |\lambda_i-\mu_i|\leq \|\delta H\|_2
\end{equation}
\end{theorem}

This result naturally extends to the case of block preserving perturbations. The summation structure for our Hamiltonian remains the same so long as $\|\delta H\|_2\ll \omega_0$. The perturbation expansion needs the variance updated with $tr((H+\delta H)^2)-tr(H+\delta H)^2$, which slightly increases the value, but if small enough perturbations, the net change is effectively nothing.

It is worth noting that formally $H$, and $\delta H$, act on Hilbert spaces with countably infinite bases, so without the thermal effects the norm may not converge. If one requires block preserving perturbations, then the level of perturbation needs only be bounded by $\omega_0$ within each block, and so the infiniteness of the Hilbert space is circumvented. Some examples of block preserving perturbations include those which preserve the angular momentum and total excitation quantum numbers, which would include most forms of field and coupling inhomogeneities, but would not include dipolar interactions or individual spin orientation perturbations.

Formally these operators are semi-bounded Hermitian operator on a countably infinite space and the results are rigorous so long as $\|\delta H\|_2$ is bounded \cite{kato2013perturbation}. Fortunately, the Hilbert spaces of $H$ and $H+\delta H$ are effectively truncated at large $k$ since the rotating-wave approximation breaks down, which results in the loss of the good quantum numbers and the lack of validity for the direct sum decomposition.

As an explicit example, let us consider dipolar flip-flop coupling between pairs of spins:
\begin{equation}
    \delta H_{ff}=d_{ij}(\sigma_+^{(i)} \sigma_-^{(j)}+\sigma_-^{(i)}\sigma_+^{(j)}),\quad i\neq j.
\end{equation}
Note that this is independent of the number of excitations in the cavity so the Hilbert space can be taken as a direct sum over the number of excitations in the cavity, meanning that we only require $\|\delta H_{ff}\|_2\ll \omega_0$ within each block. We then have:
\begin{eqnarray}
    \|\delta H_{ff}\|_2 &=& \sqrt{tr(\delta H_{ff} \delta H_{ff})}\\
    &=& \sqrt{ \sum_{ij} 2 d_{ij}^2}.
\end{eqnarray}
Then this says that interaction due to dipolar flip-flops does not change the summation structure so long as $\sum_{ij} 2 d_{ij}^2\ll \omega_0^2$. If one wished to bound this in some lattice-like scenario where each spin couples predominantly with $O(1)$ spins, then the total summation will have $\sum_{ij} 2d_{ij}^2$ going as $O(N)$. If no such restriction is made, then there will be $O(N^2)$ terms in the summation.


\if{false}
\subsection{Number Operator Self-evolution}\label{nshift}

To compute the expectation of the number operator under self-evolution, we will need to compute $\hat{n}_{pert}(j,k):=tr(\omega_0 a^\dag a L^2(j,k))$.
\if{false}
\begin{eqnarray}
    & &tr(a^\dag a L^2(j,k))\\
    &=& \langle 1_{j,k} |a^\dag a|1_{j,k}\rangle l_1^2(j,k)+\langle |\mathcal{B}_{j,k}| |a^\dag a||\mathcal{B}|_{j,k}\rangle l_{|\mathcal{B}_{j,k}|}^2(j,k)+\sum_{\alpha=2}^{|\mathcal{B}_{j,k}|-1} \langle \alpha_{j,k} |a^\dag a|\alpha_{j,k}\rangle (l_\alpha^2(j,k)+l_{\alpha+1}^2(j,k))\\
    &=& (k+j-\frac{N}{2}-1)(2j)(k+j-\frac{N}{2})+(k+j-\frac{N}{2}-|\mathcal{B}|)(|\mathcal{B}|(2j-(|\mathcal{B}|-1))(k+j-\frac{N}{2}+1-|\mathcal{B}|)\\
    & & +\sum_{\alpha=2}^{|\mathcal{B}|-1} (k+j-\frac{N}{2}-\alpha)(\alpha (2j-\alpha+1)(k+j-\frac{N}{2}-\alpha+1)+(\alpha+1)(2j-\alpha)(k+j-\frac{N}{2}-\alpha))\\
    &=& (k+j-\frac{N}{2}-1)(2j)(k+j-\frac{N}{2})+(k+j-\frac{N}{2}-|\mathcal{B}|)(|\mathcal{B}|(2j-(|\mathcal{B}|-1))(k+j-\frac{N}{2}+1-|\mathcal{B}|)\\
    & & + \sum_{\alpha=2}^{|\mathcal{B}|-1}[ (-2)\alpha^4+ (1+8j +4k-2N)\alpha^3\\
    & &+(-1-10j^2-j-12jk+6jN-2k^2-k+2kN-N^2/2+N/2)\alpha^2\\
    & &+(4j^3-2j^2+8j^2k-4j^2N+j+4jk^2-2jk-4jkN+jN^2+jN+k-N/2)\alpha\\
    & &+(2j^3+4j^2k-2j^2N+2jk^2+jN^2/2)]\\
    &=& (k+j-\frac{N}{2}-1)(2j)(k+j-\frac{N}{2})+(k+j-\frac{N}{2}-|\mathcal{B}|)(|\mathcal{B}|(2j-(|\mathcal{B}|-1))(k+j-\frac{N}{2}+1-|\mathcal{B}|)\\
    & & + A\frac{1}{30}(6|\mathcal{B}|^5-15|\mathcal{B}|^4+10|\mathcal{B}|^3-|\mathcal{B}|-30)+B\frac{1}{4}(|\mathcal{B}|^4-2|\mathcal{B}|^3+|\mathcal{B}|^2-4)\\
    & &+C\frac{1}{6}(2|\mathcal{B}|^3-3|\mathcal{B}|^2+|\mathcal{B}|-6)+D\frac{1}{2}(|\mathcal{B}|^2-|\mathcal{B}|-2)+E(|\mathcal{B}|-2)
\end{eqnarray}
\fi
Written using $k'=k-k_0(j)$ and following the notations used in \cite{gunderman2022lamb}, we may compute:
\begin{eqnarray}
    & &tr(a^\dag a L^2(j,k))\\
    &=& \langle 1_{j,k} |a^\dag a|1_{j,k}\rangle l_1^2(j,k)+\langle |\mathcal{B}_{j,k}| |a^\dag a||\mathcal{B}|_{j,k}\rangle l_{|\mathcal{B}_{j,k}|}^2(j,k)+\sum_{\alpha=2}^{|\mathcal{B}_{j,k}|-1} \langle \alpha_{j,k} |a^\dag a|\alpha_{j,k}\rangle (l_\alpha^2(j,k)+l_{\alpha+1}^2(j,k))\\
    &=& (k'-1)(2j)(k')+(k'-|\mathcal{B}|)(|\mathcal{B}|(2j-(|\mathcal{B}|-1))(k'+1-|\mathcal{B}|)\\
    & & +\sum_{\alpha=2}^{|\mathcal{B}|-1} (k'-\alpha)(\alpha (2j-\alpha+1)(k'-\alpha+1)+(\alpha+1)(2j-\alpha)(k'-\alpha))\\
    &=& (k'-1)(2j)(k')+(k'-|\mathcal{B}|)(|\mathcal{B}|(2j-(|\mathcal{B}|-1))(k'+1-|\mathcal{B}|)\\
    & & + \sum_{\alpha=2}^{|\mathcal{B}|-1}[ (-2)\alpha^4+ (1+4j+4k')\alpha^3+(-1-8jk'-2k'^2-k')\alpha^2+(4jk'^2-2jk'+k')\alpha+(2jk'^2)]\\
    &=& 2jk'(k'-1)+[-|\mathcal{B}|^4+(2j+2k')|\mathcal{B}|^3+(-4jk'-2j-k'^2-3k'-1)|\mathcal{B}|^2+(2jk'^2+2jk'+k'^2+k')|\mathcal{B}|\\
    & & + A\frac{1}{30}(6|\mathcal{B}|^5-15|\mathcal{B}|^4+10|\mathcal{B}|^3-|\mathcal{B}|-30)+B\frac{1}{4}(|\mathcal{B}|^4-2|\mathcal{B}|^3+|\mathcal{B}|^2-4)\\
    & &+C\frac{1}{6}(2|\mathcal{B}|^3-3|\mathcal{B}|^2+|\mathcal{B}|-6)+D\frac{1}{2}(|\mathcal{B}|^2-|\mathcal{B}|-2)+E(|\mathcal{B}|-2)\\
    &=& \frac{A}{5}|\mathcal{B}|^5+(-1-\frac{A}{2}+\frac{B}{4})|\mathcal{B}|^4+(2j+2k'+\frac{A}{3}-\frac{B}{2}+\frac{C}{3})|\mathcal{B}|^3\\
    & &+(-4jk'-2j-k'^2-3k'-1+\frac{B}{4}-\frac{C}{2}+\frac{D}{2})|\mathcal{B}|^2+(2jk'^2+2jk'+k'^2+k'-\frac{A}{30}+\frac{C}{6}-\frac{D}{2}+E)|\mathcal{B}|\\
    & & +(2jk'(k'-1)-A-B-C-D-2E)\\
    &=& -\frac{2}{5}|\mathcal{B}|^5+\frac{1}{4}(1+4j+4k')|\mathcal{B}|^4+\frac{1}{6}(-16jk'-4k'^2-2k'-9)|\mathcal{B}|^3+\frac{1}{4}(8jk'^2-4jk'-4j-4k'-1)|\mathcal{B}|^2\\
    & & +\frac{1}{30}(60jk'^2+50jk'+20k'^2+10k'-3)|\mathcal{B}|+2(-3jk'^2+4jk'-2j+k'^2-2k'+1)
\end{eqnarray}
Evaluating this at the two regimes $|\mathcal{B}|=\min \{2j+1,k'+1\}$, we obtain:
\begin{equation}
    \frac{16}{5}j^5-\frac{16}{3}j^4k'+4j^4+\frac{8}{3}j^3k'^2-\frac{20}{3}j^3k'-16j^3+4j^2k'^2-\frac{2}{3}j^2k'-25j^2-\frac{14}{3}jk'^2+\frac{26}{3}jk'-\frac{81}{5}j+2k'^2-4k'
\end{equation}
and
\begin{equation}
    \frac{1}{3} jk'^4+jk'^3-\frac{16}{3}jk'^2+8jk'-4j-\frac{1}{15}k'^5-\frac{1}{12}k'^4-\frac{11}{6}k'^3-\frac{47}{12}k'^2-\frac{101}{10}k',
\end{equation}
respectively. These are analytical expressions still.

\if{false}
We wish to find: NEED CONJ BY H? Something isn't right here........
\begin{equation}
    \frac{tr(\omega_0a^\dag a e^{-\beta H}e^{it H})}{Z_{total}}=\frac{tr(e^{(it-\beta)H})}{Z_{total}}
\end{equation}
Note that due to the parity of the eigenvalues of $H_{pert}$, this can be written as:
\begin{equation}
    \frac{tr(\omega_0a^\dag a e^{(it-\beta)H_0}\cosh((it-\beta)H_{pert}))}{2Z_{total}}
\end{equation}
Expanding the $\cosh$, this becomes:
\begin{equation}
    \frac{tr(\omega_0 a^\dag a e^{(it-\beta)H_0} [\cosh(\beta H_{pert})\cos(t H_{pert})+i\sinh(\beta H_{pert})\sin(t H_{pert})]}{2Z_{total}}
\end{equation}
We will consider $t$ such that $tH_{pert}\ll 1$ and the perturbative expansion only requires the first term. Then this reduces to:
\begin{equation}
    \frac{tr(\omega_0a^\dag a e^{(it-\beta)H_0} [1+(\beta H_{pert})^2/2-(t H_{pert})^2/2+i(\beta H_{pert})(t H_{pert})])}{2Z_{total}}
\end{equation}
Breaking this into real and imaginary components and rewriting the expression in terms of sums, we obtain as the perturbation:
\if{false}
\begin{multline}
    Real(t)=\frac{1}{2Z_{total}}(\sum_k e^{-\beta \omega_0 k}\cos(\omega_0 k t)\sum_j d_j |\mathcal{B}|(1+\frac{\beta^2-t^2}{2}Var(\Lambda(j,k)))\\
    -\sum_k e^{-\beta\omega_0 k}\sin(\omega_0 k t)\beta t \sum_j d_j |\mathcal{B}|Var(\Lambda(j,k)))
\end{multline}
\begin{multline}
    Imag(t)=\frac{1}{2Z_{total}}(\sum_k e^{-\beta \omega_0 k}\sin(\omega_0 k t)\sum_j d_j |\mathcal{B}|(1+\frac{\beta^2-t^2}{2}Var(\Lambda(j,k))) \\
    + \sum_k e^{-\beta\omega_0 k}\cos(\omega_0 k t)\beta t \sum_j d_j |\mathcal{B}|Var(\Lambda(j,k)))
\end{multline}
The perturbation is given by:
\fi
\begin{multline}
    Real(t)=\frac{1}{2Z_{total}}(\sum_k e^{-\beta \omega_0 k}\cos(\omega_0 k t)\sum_j d_j \frac{\beta^2-t^2}{2}\hat{n}_{pert}(j,k)-\sum_k e^{-\beta\omega_0 k}\sin(\omega_0 k t)\beta t \sum_j d_j \hat{n}_{pert}(j,k)
\end{multline}
\begin{multline}
    Imag(t)=\frac{1}{2Z_{total}}(\sum_k e^{-\beta \omega_0 k}\sin(\omega_0 k t)\sum_j d_j \frac{\beta^2-t^2}{2}\hat{n}_{pert}(j,k)     + \sum_k e^{-\beta\omega_0 k}\cos(\omega_0 k t)\beta t \sum_j d_j \hat{n}_{pert}(j,k))
\end{multline}
\fi

\subsubsection{Second moment}
We begin by defining $\hat{n}_{pert}^{(2)}(j,k):=tr(\omega_0^2 (a^\dag a)^2 L^2(j,k))$. Using a similar method as just before, we may compute:
\begin{eqnarray}
    & &tr((a^\dag a)^2 L^2(j,k))\\
    &=& \langle 1_{j,k} |(a^\dag a)^2|1_{j,k}\rangle l_1^2(j,k)+\langle |\mathcal{B}_{j,k}| |(a^\dag a)^2||\mathcal{B}|_{j,k}\rangle l_{|\mathcal{B}_{j,k}|}^2(j,k)+\sum_{\alpha=2}^{|\mathcal{B}_{j,k}|-1} \langle \alpha_{j,k} |(a^\dag a)^2|\alpha_{j,k}\rangle (l_\alpha^2(j,k)+l_{\alpha+1}^2(j,k))\\
    &=& (k'-1)^2(2j)(k')+(k'-|\mathcal{B}|)^2(|\mathcal{B}|(2j-(|\mathcal{B}|-1))(k'+1-|\mathcal{B}|)\\
    & & +\sum_{\alpha=2}^{|\mathcal{B}|-1} (k'-\alpha)^2(\alpha (2j-\alpha+1)(k'-\alpha+1)+(\alpha+1)(2j-\alpha)(k'-\alpha))\\
    &=& (k'-1)^2(2j)(k')+(k'-|\mathcal{B}|)^2(|\mathcal{B}|(2j-(|\mathcal{B}|-1))(k'+1-|\mathcal{B}|)\\
    & & + \sum_{\alpha=2}^{|\mathcal{B}|-1}[ 2\alpha^5+(-1-4j-6k')\alpha^4+(1+12jk'+6k'^2+2k')\alpha^3\\
    & &+(-12jk'+2jk'-2k'^3-k'^2-2k')\alpha^2+(4jk'^3-4jk'^2+k'^2)\alpha+(2jk'^3)]\\
    &=& |\mathcal{B}|^5+(-2j-3k'-2)|\mathcal{B}|^4+(6jk'+2j+3k'^2+5k'+1)|\mathcal{B}|^3+(-6jk'^2-4jk'-k'^3-4k'^2-2k')|\mathcal{B}|^2\\
    & &+(2jk'^3+2jk'^2+k'^3+k'^2)|\mathcal{B}|+(2jk'^3-4jk'^2+2jk')\\
    & & +A\frac{1}{12}(2|\mathcal{B}|^6-6|\mathcal{B}|^5+5|\mathcal{B}|^4-|\mathcal{B}|^2-12)+B\frac{1}{30}(6|\mathcal{B}|^5-15|\mathcal{B}|^4+10|\mathcal{B}|^3-|\mathcal{B}|-30)\\
    & & +C\frac{1}{4}(|\mathcal{B}|^4-2|\mathcal{B}|^3+|\mathcal{B}|^2-4) +D\frac{1}{6}(2|\mathcal{B}|^3-3|\mathcal{B}|^2+|\mathcal{B}|-6)+E\frac{1}{2}(|\mathcal{B}|^2-|\mathcal{B}|-2)+F(|\mathcal{B}|-2)\\
    &=& \frac{A}{6}|\mathcal{B}|^6+(1-\frac{A}{2}+\frac{B}{5})|\mathcal{B}|^5+(-2j-3k'-2+\frac{5}{12}A-\frac{B}{2}+\frac{C}{4})|\mathcal{B}|^4\\
    & & +(6jk'+2j+3k'^2+5k'+1+\frac{B}{3}-\frac{C}{2}+\frac{D}{3})|\mathcal{B}|^3\\
    & &+(-6jk'^2-4jk'-k'^3-4k'^2-2k'-\frac{A}{12}+\frac{C}{4}-\frac{D}{2}+\frac{E}{2})|\mathcal{B}|^2\\
    & &+(2jk'^3+2jk'^2+k'^3+k'^2-\frac{B}{30}+\frac{D}{6}-\frac{E}{2}+F)|\mathcal{B}|+(2jk'^3-4jk'^2+2jk'-A-B-C-D-E-2F)\\
    &=& \frac{1}{3}|\mathcal{B}|^6+\frac{1}{5}(-1-4j-6k')|\mathcal{B}|^5+(3jk'+\frac{3}{2}k'^2+\frac{1}{2}k'-\frac{5}{12})|\mathcal{B}|^4\\
    & &+(-\frac{10}{3}jk'+\frac{2}{3}j-\frac{2}{3}k'^3-\frac{1}{3}k'^2+\frac{4}{3}k'+\frac{1}{6})|\mathcal{B}|^3+(2jk'^3-8jk'^2+4jk'-\frac{3}{2}k'^2-\frac{1}{2}k'+\frac{1}{12})|\mathcal{B}|^2\\
    & &+(2jk'^3+4jk'^2-\frac{5}{3}jk'+\frac{2}{3}k'^3+\frac{1}{3}k'^2-\frac{1}{3}k'-\frac{1}{15})|\mathcal{B}|+(-jk'^3+4j+2k'^3-6k'^2+6k'-2)
\end{eqnarray}

Evaluating this at the two regimes $|\mathcal{B}|=\min \{2j+1,k'+1\}$, we obtain:
\begin{multline}
    -\frac{64}{15}j^6+\frac{48}{5}j^5k'-\frac{32}{5}j^5+24j^4k'^2-\frac{56}{3}j^4k'-\frac{4}{3}j^4+\frac{8}{3}j^3k'^3+\frac{40}{3}j^3k'^2-\frac{64}{3}j^3k'+\frac{4}{3}j^3\\
    +4j^2k'^2+2j^2k'^2-\frac{16}{3}j^2k'+\frac{1}{3}j^2+\frac{1}{3}jk'^3+\frac{2}{3}jk'^2-\frac{2}{3}jk'+\frac{56}{15}j+2k'^3-6k'^2+\frac{29}{5}k'-\frac{21}{10}
\end{multline}
and
\begin{equation}
    -\frac{13}{10}k'^6+\frac{21}{5}jk'^5-\frac{1}{30}k'^5+\frac{8}{3}jk'^4+\frac{1}{12}k'^4-\frac{13}{3}jk'^3+2k'^3-\frac{5}{3}jk'^2-\frac{25}{4}k'^2+\frac{86}{15}k'+\frac{58}{15}j-\frac{21}{10}
\end{equation}
respectively. These are analytical expressions still.

\subsection{Variance of $\hat{n}$ due to the Lamb shifts}

\begin{equation}
    Var(a^\dag a)=\langle (a^\dag a)^2\rangle-\langle (a^\dag a)\rangle^2
\end{equation}
As per appendix \ref{pertexp}, the density matrix is a sum of the $g_0=0$ density matrix, then an increasing series in powers of $g_0^2$. By the linearity of trace these expectations will also be sums of these forms. \if{false} Since we are interested in the variance in $a^\dag a$ induced by the Lamb shifts we will subtract out the variance that naturally exists in the system, resulting in needing to evaluate the following to obtain the leading perturbation:
\begin{multline}
    tr((a^\dag a)^2 e^{(it-\beta)H_0}[\frac{\beta^2-t^2}{2} H_{pert}^2+i\beta t H_{pert}^2])-2tr((a^\dag a e^{(it-\beta)H_0}))tr(a^\dag a e^{(it-\beta)H_0}[\frac{\beta^2-t^2}{2} H_{pert}^2+i\beta t H_{pert}^2])\\
    -tr(a^\dag a e^{(it-\beta)H_0}[\frac{\beta^2-t^2}{2} H_{pert}^2+i\beta t H_{pert}^2])^2
\end{multline}
Replacing $H_{pert}^2$ with $\hat{n}$ and $\hat{n}^{(2)}$ and dropping the last term as it is comparable to $g_0^4$, this becomes:
\begin{equation}
    tr( e^{(it-\beta)H_0}[\frac{\beta^2-t^2}{2} \hat{n}^{(2)}+i\beta t \hat{n}^{(2)}])-2tr((a^\dag a e^{(it-\beta)H_0}))tr( e^{(it-\beta)H_0}[\frac{\beta^2-t^2}{2} \hat{n}+i\beta t \hat{n}]),
\end{equation}
where the subscripts and arguments for $\hat{n}$ and $\hat{n}^{(2)}$ have been suppressed for notational clarity. Breaking this into real and imaginary components then expressing the real part as sums, we obtain:
\begin{multline}
    Real[ Var(a^\dag a)] (t)=\frac{1}{2Z_{total}}(\sum_k e^{-\beta \omega_0 k}\cos(\omega_0 k t)\sum_j d_j \frac{\beta^2-t^2}{2}\hat{n}^{(2)}-\sum_k e^{-\beta\omega_0 k}\sin(\omega_0 k t)\beta t \sum_j d_j \hat{n}^{(2)})\\
    -2\frac{1}{(2Z_{total})^2} (\sum_k ke^{-\beta\omega_0}\cos(\omega_0 k t)\sum_j d_j|\mathcal{B}|) (\sum_k e^{-\beta \omega_0 k}\cos(\omega_0 k t)\sum_j d_j \frac{\beta^2-t^2}{2}\hat{n}-\sum_k e^{-\beta\omega_0 k}\sin(\omega_0 k t)\beta t \sum_j d_j \hat{n})\\
    -2\frac{1}{(2Z_{total})^2} (\sum_k ke^{-\beta\omega_0}\sin(\omega_0 k t)\sum_j d_j|\mathcal{B}|) (\sum_k e^{-\beta \omega_0 k}\sin(\omega_0 k t)\sum_j d_j \frac{\beta^2-t^2}{2}\hat{n}+\sum_k e^{-\beta\omega_0 k}\cos(\omega_0 k t)\beta t \sum_j d_j \hat{n})
\end{multline}

\if{false}
\subsection{Partition Functions Distributions}
This section contains plots of the partition function as a plot...
\begin{figure}[th]
    \centering
    \includegraphics[scale=0.5]{Visuals (44).png}
    \caption{Schematic figure showing $j$ values that need summing over for any given $k$ choice. The width of each shape is $\Theta (\sqrt{N})$, aside from the tapered tip of the triangle for which there are insufficiently many angular momentum values available, so all of them are used there.}
\end{figure}

    \begin{figure}[th]
    \centering
    \includegraphics[scale=0.5]{Visuals (45).png}
    \caption{Schematic figure showing $j$ values that need summing over for any given $k$ choice. The width of each shape is $\Theta (\sqrt{N})$, aside from the tapered tip of the triangle for which there are insufficiently many angular momentum values available, so all of them are used there.}
\end{figure}

\fi
\fi
\clearpage
\subsection{Fractional changes in $\langle a^\dag a\rangle$ and $Var(a^\dag a)$ due to the Lamb shifts}

From the prior pages we found expressions for the expectations and variances of the uncoupled system and the coupled system. Taking the difference of the values for $a^\dag a$ and $Var(a^\dag a)$ and dividing by the uncoupled system value, we obtain the following pair of plots, both of which are linear. In effect, these plot the change in the number of photons in the cavity due strictly to the Lamb shifts, and likewise for the variance in the number of photons in the cavity.
\fi

\if{false}
    \begin{figure}[h]
    \centering
    \includegraphics[scale=0.5]{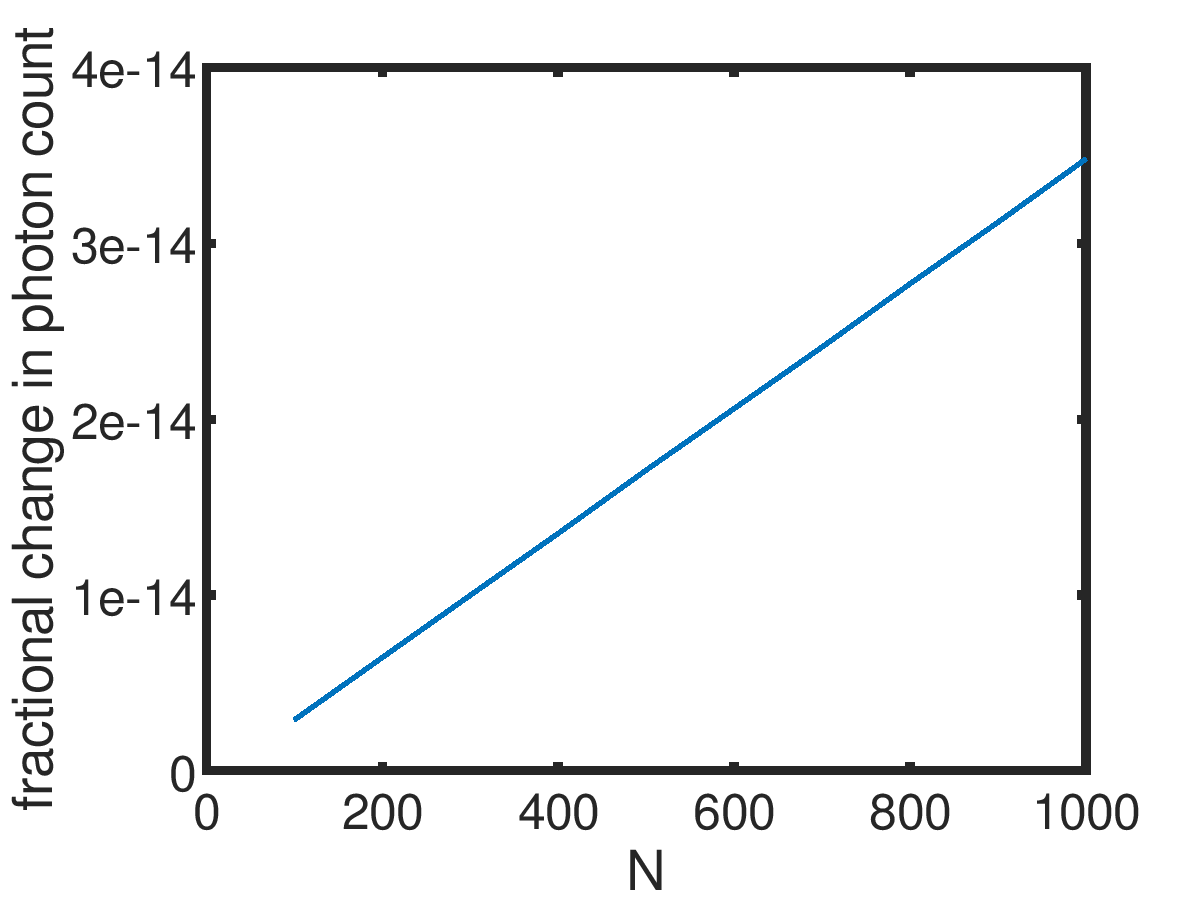}
    \caption{The fractional change in the mean number of photons in the cavity as a function of $N$. This is plotted for $T=0.3$ K, $g_0=200\pi$ Hz, $\omega_0=20\pi$ GHz, and $N=100$ to $1000$. This is a linear trend ($r^2=1$) with slope $3.54\cdot 10^{-17}$ and intercept $-6.49\cdot 10^{-16}$. Once $N\approx 10^{15}$ this fraction would be appreciable.}
    \label{nmeanshift}
\end{figure}

    \begin{figure}[h]
    \centering
    \includegraphics[scale=0.5]{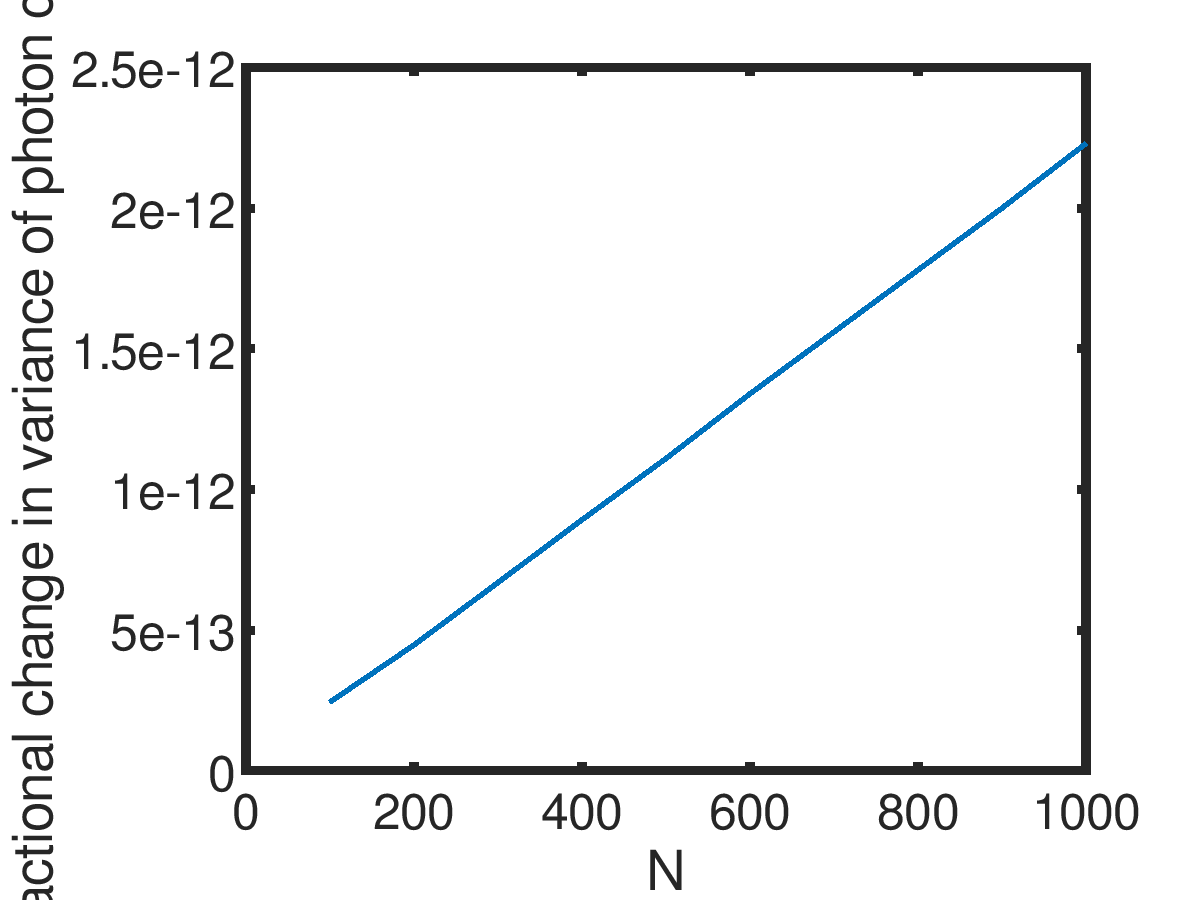}
    \caption{The fractional change in the variance of the number of photons in the cavity as a function of $N$. This is plotted for $T=0.3$ K, $g_0=200\pi$ Hz, $\omega_0=20\pi$ GHz, and $N=100$ to $1000$. This is a linear trend ($r^2=1$) with slope $2.21\cdot 10^{-15}$ and intercept $9.4\cdot 10^{-15}$. Once $N\approx 10^{13}$ this fraction would be appreciable.}
    \label{nvarshift}
\end{figure}
\fi

\if{false}
\clearpage
\subsection{Driven Solution}\label{driven}

We will consider the result of:
\begin{equation}
    tr(J_y e^{it(H+2\Omega\cos(\omega_0 t) J_x)}\rho_{th}e^{-it(H+2\Omega\cos(\omega_0 t) J_x)})
\end{equation}
Note that since $\rho_{th}\propto e^{-\beta H}$ we may commute the evolution of $H_{int}$ around it as $H_{int}$ (without the rotating-wave approximation) commutes with $J_x$. Further the term $a^\dag a$ in $H_0$ commutes with $J_y$, so we may cancel that term as well. Then this reduces to:
\begin{equation}
    tr(J_y e^{it(\omega_0J_z+2\Omega\cos(\omega_0 t)J_x)}\rho_{th}e^{-it(\omega_0J_z+2\Omega\cos(\omega_0 t)J_x)}).
\end{equation}
We will now move into the rotating frame defined by $\omega_0 J_z$ for the Hamiltonian $\omega_0 J_z+ 2\Omega\cos(\omega_0 t) J_x$. This becomes $\Omega J_x$ \cite{wood_cavity_2014}. Our expression is now:
\begin{equation}
    tr(e^{-it\Omega J_x}J_y e^{it\Omega J_x} \rho_{th}).
\end{equation}
If we only keep the terms with constant number of excitations, we obtain from \cite{wood_cavity_2014} that:
\begin{eqnarray}
    tr(e^{-it\Omega J_x}J_y e^{it\Omega J_x} \rho_{th})&=& tr(\frac{1}{2} (e^{-i\Omega t} (J_y+iJ_z)+e^{i\Omega t}(J_y-iJ_z))\rho_{th})\\
    &=& tr(\cos(\Omega t)J_y \rho_{th})+tr(\sin(\Omega t)J_z\rho_{th})\\
    &=& \sin(\Omega t) tr(J_z \rho_{th}), \label{behavior}
\end{eqnarray}
where the first trace is zero due to not preserving the number of excitations.

Computing $tr(\omega_0 J_z\rho_{th})$ first non-zero change:
\begin{eqnarray}
    tr(J_z L^2(j,k))&=& (1-j)l_1^2(j,k)+(|\mathcal{B}|-j)l_{|\mathcal{B}|}^2(j,k)+\sum_{\alpha=2}^{|\mathcal{B}|-1}(\alpha-j)(l_\alpha^2(j,k)+l_{\alpha+1}^2(j,k))\\
    &=& (1-j)(2j)(k')+(|\mathcal{B}|-j)(|\mathcal{B}|(2j-(|\mathcal{B}|-1))(k'+1-|\mathcal{B}|)\\
    & & +\sum_{\alpha=2}^{|\mathcal{B}|-1}(\alpha-j)(\alpha (2j-\alpha+1)(k'-\alpha+1)+(\alpha+1)(2j-\alpha)(k'-\alpha))\\
    &=& |\mathcal{B}|^4+(-3j-k'-2)|\mathcal{B}|^3+(2j^2+3jk'+4j+k'+1)|\mathcal{B}|^2\\
    & & +(-2j^2k'-2j^2-jk'-j)|\mathcal{B}|+(-2j^2k'+2jk')\\
    & & +\sum_{\alpha=2}^{|\mathcal{B}|-1} 2\alpha^4+(-1-6j-2k')\alpha^3+(1+4j^2+j+6jk')\alpha^2+(-4j^2k'-j+2jk')\alpha+2j^2k'\\
    &=& |\mathcal{B}|^4+(-3j-k'-2)|\mathcal{B}|^3+(2j^2+3jk'+4j+k'+1)|\mathcal{B}|^2\\
    & & +(-2j^2k'-2j^2-jk'-j)|\mathcal{B}|+(-2j^2k'+2jk')\\
    & & +A\frac{1}{30}(6|\mathcal{B}|^5-15|\mathcal{B}|^4+10|\mathcal{B}|^3-|\mathcal{B}|-30)+B\frac{1}{4}(|\mathcal{B}|^4-2|\mathcal{B}|^3+|\mathcal{B}|^2-4)\\
    & & +C\frac{1}{6}(2|\mathcal{B}|^3-3|\mathcal{B}|^2+|\mathcal{B}|-6)+D\frac{1}{2}(|\mathcal{B}|^2-|\mathcal{B}|-2)+E(|\mathcal{B}|-2)\\
    &=& \frac{A}{5}|\mathcal{B}|^5+(1-\frac{A}{2}+\frac{B}{4})|\mathcal{B}|^4+(-3j-k'-2+\frac{A}{3}-\frac{B}{2}+\frac{C}{3})|\mathcal{B}|^3\\
    & &+ (2j^2+3jk'+4j+k'+1+\frac{B}{4}-\frac{C}{2}+\frac{D}{2})|\mathcal{B}|^2\\
    & &+ (-2j^2k'-2j^2-2jk'-j-\frac{A}{30}+\frac{C}{6}-\frac{D}{2}+E)|\mathcal{B}|\\
    & &+(-2j^2k'+2jk'-A-B-C-D-2E)\\
    &=& \frac{2}{5}|\mathcal{B}|^5+\frac{1}{4}(-1-6j-2k')|\mathcal{B}|^4+\frac{1}{6}(8j^2+12jk'+2j-3)|\mathcal{B}|^3\\
    & & +\frac{1}{4}(-8j^2k'+4jk'+6j+2k'+1)|\mathcal{B}|^2+\frac{1}{30}(60j^2k'-40j^2-60jk'-10j+3)|\mathcal{B}|\\
    & & +(-2j^2k'-4j^2-6jk'+6j+2k'-2)
\end{eqnarray}
\fi

\if{false}
    \begin{figure}[h]
    \centering
    \includegraphics[scale=0.8]{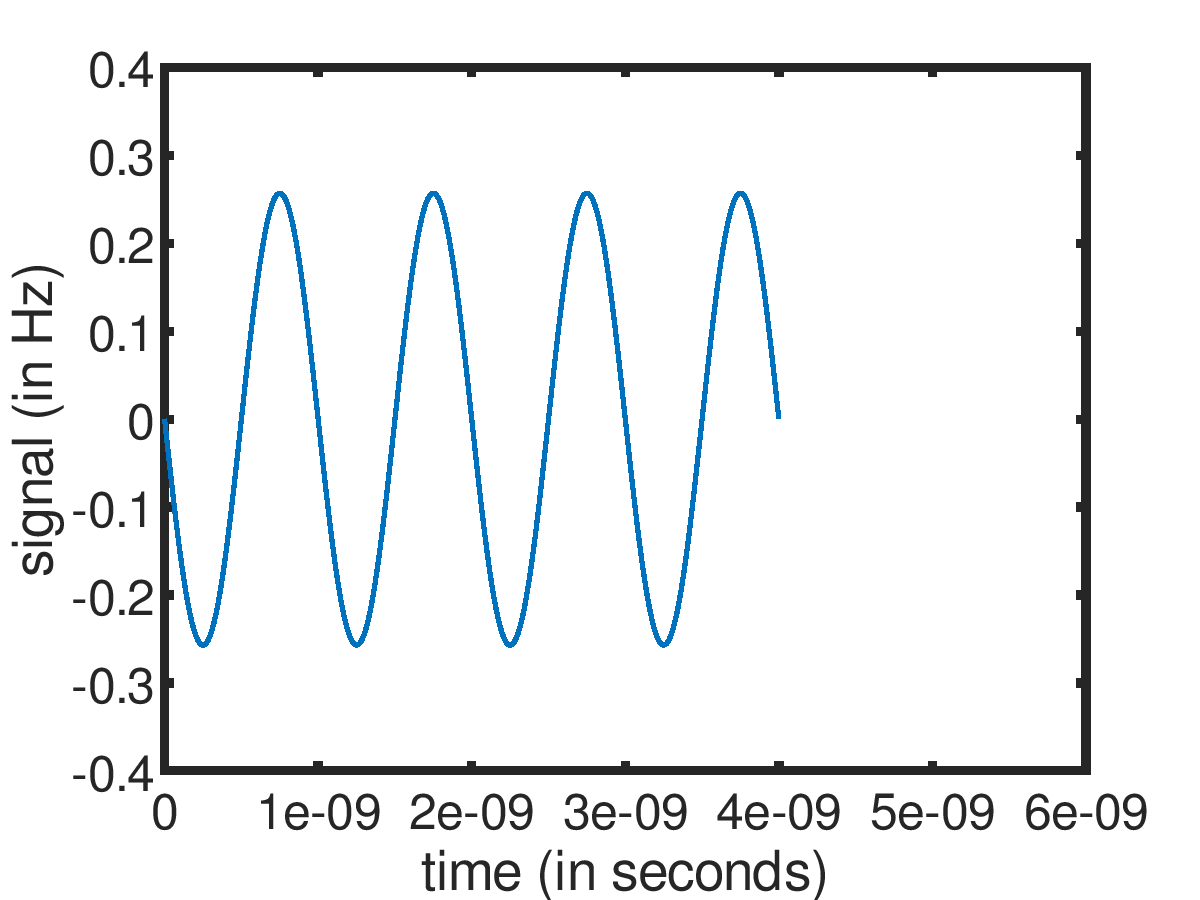}
    \caption{The signal of $\omega_0 J_y$ after time $t$ in the rotating-frame of $J_z$ for the driven Hamiltonian $H+\Omega J_x$ with $\Omega=2\pi$ GHz. This is plotted for $T=0.3$ K, $g_0=200\pi$ Hz, $\omega_0=20\pi$ GHz, and $N=1000$. As stated in equation (\ref{behavior}), this is a simple sine times the shift from the prior figure.}
    \label{jzshiftplot}
\end{figure}
\fi

\subsection{Number Operator Expectation Fractional Shift and Variance Fractional Shift}\label{nshift}
To compute the expectation of the number operator under self-evolution, we will need to compute $\hat{n}_{\text{pert}}(j,k):=\text{tr}(\omega_0 a^\dag a L^2(j,k))$.
\if{false}
\begin{eqnarray}
    & &\text{tr}(a^\dag a L^2(j,k))\\
    &=& \langle 1_{j,k} |a^\dag a|1_{j,k}\rangle l_1^2(j,k)+\langle |\mathcal{B}_{j,k}| |a^\dag a||\mathcal{B}|_{j,k}\rangle l_{|\mathcal{B}_{j,k}|}^2(j,k)+\sum_{\alpha=2}^{|\mathcal{B}_{j,k}|-1} \langle \alpha_{j,k} |a^\dag a|\alpha_{j,k}\rangle (l_\alpha^2(j,k)+l_{\alpha+1}^2(j,k))\\
    &=& (k+j-\frac{N}{2}-1)(2j)(k+j-\frac{N}{2})+(k+j-\frac{N}{2}-|\mathcal{B}|)(|\mathcal{B}|(2j-(|\mathcal{B}|-1))(k+j-\frac{N}{2}+1-|\mathcal{B}|)\\
    & & +\sum_{\alpha=2}^{|\mathcal{B}|-1} (k+j-\frac{N}{2}-\alpha)(\alpha (2j-\alpha+1)(k+j-\frac{N}{2}-\alpha+1)+(\alpha+1)(2j-\alpha)(k+j-\frac{N}{2}-\alpha))\\
    &=& (k+j-\frac{N}{2}-1)(2j)(k+j-\frac{N}{2})+(k+j-\frac{N}{2}-|\mathcal{B}|)(|\mathcal{B}|(2j-(|\mathcal{B}|-1))(k+j-\frac{N}{2}+1-|\mathcal{B}|)\\
    & & + \sum_{\alpha=2}^{|\mathcal{B}|-1}[ (-2)\alpha^4+ (1+8j +4k-2N)\alpha^3\\
    & &+(-1-10j^2-j-12jk+6jN-2k^2-k+2kN-N^2/2+N/2)\alpha^2\\
    & &+(4j^3-2j^2+8j^2k-4j^2N+j+4jk^2-2jk-4jkN+jN^2+jN+k-N/2)\alpha\\
    & &+(2j^3+4j^2k-2j^2N+2jk^2+jN^2/2)]\\
    &=& (k+j-\frac{N}{2}-1)(2j)(k+j-\frac{N}{2})+(k+j-\frac{N}{2}-|\mathcal{B}|)(|\mathcal{B}|(2j-(|\mathcal{B}|-1))(k+j-\frac{N}{2}+1-|\mathcal{B}|)\\
    & & + A\frac{1}{30}(6|\mathcal{B}|^5-15|\mathcal{B}|^4+10|\mathcal{B}|^3-|\mathcal{B}|-30)+B\frac{1}{4}(|\mathcal{B}|^4-2|\mathcal{B}|^3+|\mathcal{B}|^2-4)\\
    & &+C\frac{1}{6}(2|\mathcal{B}|^3-3|\mathcal{B}|^2+|\mathcal{B}|-6)+D\frac{1}{2}(|\mathcal{B}|^2-|\mathcal{B}|-2)+E(|\mathcal{B}|-2)
\end{eqnarray}
\fi
Written using $k'=k-k_0(j)$ and following the notations used in \cite{gunderman2022lamb}, we may compute:
\begin{eqnarray}
    & &tr(a^\dag a L^2(j,k))\\
    &=& \langle 1_{j,k} |a^\dag a|1_{j,k}\rangle l_1^2(j,k)+\langle |\mathcal{B}_{j,k}| |a^\dag a||\mathcal{B}|_{j,k}\rangle l_{|\mathcal{B}_{j,k}|}^2(j,k)\\
    &&+\sum_{\alpha=2}^{|\mathcal{B}_{j,k}|-1} \langle \alpha_{j,k} |a^\dag a|\alpha_{j,k}\rangle (l_\alpha^2(j,k)+l_{\alpha+1}^2(j,k))\\
    &=& (k'-1)(2j)(k')+(k'-|\mathcal{B}|)(|\mathcal{B}|(2j-(|\mathcal{B}|-1))(k'+1-|\mathcal{B}|)\\
    & & +\sum_{\alpha=2}^{|\mathcal{B}|-1} (k'-\alpha)(\alpha (2j-\alpha+1)(k'-\alpha+1)+(\alpha+1)(2j-\alpha)(k'-\alpha))\\
    &=& (k'-1)(2j)(k')+(k'-|\mathcal{B}|)(|\mathcal{B}|(2j-(|\mathcal{B}|-1))(k'+1-|\mathcal{B}|)\\
    & & + \sum_{\alpha=2}^{|\mathcal{B}|-1}[ (-2)\alpha^4+ (1+4j+4k')\alpha^3+(-1-8jk'-2k'^2-k')\alpha^2\\
    & &+(4jk'^2-2jk'+k')\alpha+(2jk'^2)]\\
    &=& 2jk'(k'-1)+[-|\mathcal{B}|^4+(2j+2k')|\mathcal{B}|^3+(-4jk'-2j-k'^2-3k'-1)|\mathcal{B}|^2\\
    & &+(2jk'^2+2jk'+k'^2+k')|\mathcal{B}|\\
    & & + A\frac{1}{30}(6|\mathcal{B}|^5-15|\mathcal{B}|^4+10|\mathcal{B}|^3-|\mathcal{B}|-30)+B\frac{1}{4}(|\mathcal{B}|^4-2|\mathcal{B}|^3+|\mathcal{B}|^2-4)\\
    & &+C\frac{1}{6}(2|\mathcal{B}|^3-3|\mathcal{B}|^2+|\mathcal{B}|-6)+D\frac{1}{2}(|\mathcal{B}|^2-|\mathcal{B}|-2)+E(|\mathcal{B}|-2)\\
    &=& \frac{A}{5}|\mathcal{B}|^5+(-1-\frac{A}{2}+\frac{B}{4})|\mathcal{B}|^4+(2j+2k'+\frac{A}{3}-\frac{B}{2}+\frac{C}{3})|\mathcal{B}|^3\\
    & &+(-4jk'-2j-k'^2-3k'-1+\frac{B}{4}-\frac{C}{2}+\frac{D}{2})|\mathcal{B}|^2\\
    && +(2jk'^2+2jk'+k'^2+k'-\frac{A}{30}+\frac{C}{6}-\frac{D}{2}+E)|\mathcal{B}|\\
    & & +(2jk'(k'-1)-A-B-C-D-2E)\\
    &=& -\frac{2}{5}|\mathcal{B}|^5+\frac{1}{4}(1+4j+4k')|\mathcal{B}|^4+\frac{1}{6}(-16jk'-4k'^2-2k'-9)|\mathcal{B}|^3\\
    &&+\frac{1}{4}(8jk'^2-4jk'-4j-4k'-1)|\mathcal{B}|^2\\
    & & +\frac{1}{30}(60jk'^2+50jk'+20k'^2+10k'-3)|\mathcal{B}|\\
    &&+2(-3jk'^2+4jk'-2j+k'^2-2k'+1)
\end{eqnarray}
Evaluating this at the two regimes $|\mathcal{B}|=\min \{2j+1,k'+1\}$, we obtain:
\begin{multline}
    \frac{16}{5}j^5-\frac{16}{3}j^4k'+4j^4+\frac{8}{3}j^3k'^2-\frac{20}{3}j^3k'-16j^3+4j^2k'^2-\frac{2}{3}j^2k'-25j^2\\
    -\frac{14}{3}jk'^2+\frac{26}{3}jk'-\frac{81}{5}j+2k'^2-4k'
\end{multline}
and
\begin{equation}
    \frac{1}{3} jk'^4+jk'^3-\frac{16}{3}jk'^2+8jk'-4j-\frac{1}{15}k'^5-\frac{1}{12}k'^4-\frac{11}{6}k'^3-\frac{47}{12}k'^2-\frac{101}{10}k',
\end{equation}
respectively. These are analytical expressions still.

We wish to find, where $\rho_{th}$ is the coupled thermal state, the perturbative term of which is computed using $\hat{n}_{\text{pert}}(j,k)$:
\begin{equation}
    \text{fractional mean shift}:=\frac{\text{tr}(\omega_0a^\dag a \rho_{\text{th}})-\text{tr}(\omega_0 a^\dag a \rho_0)}{\text{tr}(\omega_0a^\dag a \rho_0)}
\end{equation}
\if{false}
Note that due to the parity of the eigenvalues of $H_{pert}$, this can be written as:
\begin{equation}
    \frac{tr(\omega_0a^\dag a e^{(it-\beta)H_0}\cosh((it-\beta)H_{pert}))}{2Z_{total}}
\end{equation}

Expanding the $\cosh$, this becomes:
\begin{equation}
    \frac{tr(\omega_0 a^\dag a e^{(it-\beta)H_0} [\cosh(\beta H_{pert})\cos(t H_{pert})+i\sinh(\beta H_{pert})\sin(t H_{pert})]}{2Z_{total}}
\end{equation}
We will consider $t$ such that $tH_{pert}\ll 1$ and the perturbative expansion only requires the first term. Then this reduces to:
\begin{equation}
    \frac{tr(\omega_0a^\dag a e^{(it-\beta)H_0} [1+(\beta H_{pert})^2/2-(t H_{pert})^2/2+i(\beta H_{pert})(t H_{pert})])}{2Z_{total}}
\end{equation}
Breaking this into real and imaginary components and rewriting the expression in terms of sums, we obtain as the perturbation:
\if{false}
\begin{multline}
    Real(t)=\frac{1}{2Z_{total}}(\sum_k e^{-\beta \omega_0 k}\cos(\omega_0 k t)\sum_j d_j |\mathcal{B}|(1+\frac{\beta^2-t^2}{2}Var(\Lambda(j,k)))\\
    -\sum_k e^{-\beta\omega_0 k}\sin(\omega_0 k t)\beta t \sum_j d_j |\mathcal{B}|Var(\Lambda(j,k)))
\end{multline}
\begin{multline}
    Imag(t)=\frac{1}{2Z_{total}}(\sum_k e^{-\beta \omega_0 k}\sin(\omega_0 k t)\sum_j d_j |\mathcal{B}|(1+\frac{\beta^2-t^2}{2}Var(\Lambda(j,k))) \\
    + \sum_k e^{-\beta\omega_0 k}\cos(\omega_0 k t)\beta t \sum_j d_j |\mathcal{B}|Var(\Lambda(j,k)))
\end{multline}
The perturbation is given by:
\fi
\begin{multline}
    Real(t)=\frac{1}{2Z_{total}}(\sum_k e^{-\beta \omega_0 k}\cos(\omega_0 k t)\sum_j d_j \frac{\beta^2-t^2}{2}\hat{n}_{pert}(j,k)\\-\sum_k e^{-\beta\omega_0 k}\sin(\omega_0 k t)\beta t \sum_j d_j \hat{n}_{pert}(j,k)
\end{multline}
\begin{multline}
    Imag(t)=\frac{1}{2Z_{total}}(\sum_k e^{-\beta \omega_0 k}\sin(\omega_0 k t)\sum_j d_j \frac{\beta^2-t^2}{2}\hat{n}_{pert}(j,k)\\     + \sum_k e^{-\beta\omega_0 k}\cos(\omega_0 k t)\beta t \sum_j d_j \hat{n}_{pert}(j,k))
\end{multline}
\fi

\subsubsection{Second moment}
We begin by defining $\hat{n}_{\text{pert}}^{(2)}(j,k):=\text{tr}(\omega_0^2 (a^\dag a)^2 L^2(j,k))$. Using a similar method as just before, we may compute:
\begin{eqnarray}
    & &\text{tr}((a^\dag a)^2 L^2(j,k))\\
    &=& \langle 1_{j,k} |(a^\dag a)^2|1_{j,k}\rangle l_1^2(j,k)+\langle |\mathcal{B}_{j,k}| |(a^\dag a)^2||\mathcal{B}|_{j,k}\rangle l_{|\mathcal{B}_{j,k}|}^2(j,k)\\
   \nonumber & &+\sum_{\alpha=2}^{|\mathcal{B}_{j,k}|-1} \langle \alpha_{j,k} |(a^\dag a)^2|\alpha_{j,k}\rangle (l_\alpha^2(j,k)+l_{\alpha+1}^2(j,k))\\
    &=& (k'-1)^2(2j)(k')+(k'-|\mathcal{B}|)^2(|\mathcal{B}|(2j-(|\mathcal{B}|-1))(k'+1-|\mathcal{B}|)\\
    \nonumber& & +\sum_{\alpha=2}^{|\mathcal{B}|-1} (k'-\alpha)^2(\alpha (2j-\alpha+1)(k'-\alpha+1)+(\alpha+1)(2j-\alpha)(k'-\alpha))\\
    &=& (k'-1)^2(2j)(k')+(k'-|\mathcal{B}|)^2(|\mathcal{B}|(2j-(|\mathcal{B}|-1))(k'+1-|\mathcal{B}|)\\
    \nonumber& & + \sum_{\alpha=2}^{|\mathcal{B}|-1}[ 2\alpha^5+(-1-4j-6k')\alpha^4+(1+12jk'+6k'^2+2k')\alpha^3\\
    \nonumber& &+(-12jk'+2jk'-2k'^3-k'^2-2k')\alpha^2+(4jk'^3-4jk'^2+k'^2)\alpha+(2jk'^3)]\\
    &=& |\mathcal{B}|^5+(-2j-3k'-2)|\mathcal{B}|^4+(6jk'+2j+3k'^2+5k'+1)|\mathcal{B}|^3\\
    \nonumber& &+(-6jk'^2-4jk'-k'^3-4k'^2-2k')|\mathcal{B}|^2\\
    \nonumber& &+(2jk'^3+2jk'^2+k'^3+k'^2)|\mathcal{B}|+(2jk'^3-4jk'^2+2jk')\\
    \nonumber & & +A\frac{1}{12}(2|\mathcal{B}|^6-6|\mathcal{B}|^5+5|\mathcal{B}|^4-|\mathcal{B}|^2-12)\\
    \nonumber & &+B\frac{1}{30}(6|\mathcal{B}|^5-15|\mathcal{B}|^4+10|\mathcal{B}|^3-|\mathcal{B}|-30) +C\frac{1}{4}(|\mathcal{B}|^4-2|\mathcal{B}|^3+|\mathcal{B}|^2-4)\\
    \nonumber & &+D\frac{1}{6}(2|\mathcal{B}|^3-3|\mathcal{B}|^2+|\mathcal{B}|-6)+E\frac{1}{2}(|\mathcal{B}|^2-|\mathcal{B}|-2)+F(|\mathcal{B}|-2)
    \end{eqnarray}
    Grouping by powers of $|\mathcal{B}|$, we obtain:
    \begin{eqnarray}
    &=& \frac{A}{6}|\mathcal{B}|^6+(1-\frac{A}{2}+\frac{B}{5})|\mathcal{B}|^5+(-2j-3k'-2+\frac{5}{12}A-\frac{B}{2}+\frac{C}{4})|\mathcal{B}|^4\\
    & & +(6jk'+2j+3k'^2+5k'+1+\frac{B}{3}-\frac{C}{2}+\frac{D}{3})|\mathcal{B}|^3\\
    \nonumber & &+(-6jk'^2-4jk'-k'^3-4k'^2-2k'-\frac{A}{12}+\frac{C}{4}-\frac{D}{2}+\frac{E}{2})|\mathcal{B}|^2\\
    \nonumber & &+(2jk'^3+2jk'^2+k'^3+k'^2-\frac{B}{30}+\frac{D}{6}-\frac{E}{2}+F)|\mathcal{B}|\\
    \nonumber & &+(2jk'^3-4jk'^2+2jk'-A-B-C-D-E-2F)\\
    &=& \frac{1}{3}|\mathcal{B}|^6+\frac{1}{5}(-1-4j-6k')|\mathcal{B}|^5+(3jk'+\frac{3}{2}k'^2+\frac{1}{2}k'-\frac{5}{12})|\mathcal{B}|^4\\
    \nonumber& &+(-\frac{10}{3}jk'+\frac{2}{3}j-\frac{2}{3}k'^3-\frac{1}{3}k'^2+\frac{4}{3}k'+\frac{1}{6})|\mathcal{B}|^3\\
    \nonumber& &+(2jk'^3-8jk'^2+4jk'-\frac{3}{2}k'^2-\frac{1}{2}k'+\frac{1}{12})|\mathcal{B}|^2\\
    \nonumber& &+(2jk'^3+4jk'^2-\frac{5}{3}jk'+\frac{2}{3}k'^3+\frac{1}{3}k'^2-\frac{1}{3}k'-\frac{1}{15})|\mathcal{B}|\\
    \nonumber& &+(-jk'^3+4j+2k'^3-6k'^2+6k'-2)
\end{eqnarray}

Evaluating this at the two regimes $|\mathcal{B}|=\min \{2j+1,k'+1\}$, we obtain:
\begin{multline}
    -\frac{64}{15}j^6+\frac{48}{5}j^5k'-\frac{32}{5}j^5+24j^4k'^2-\frac{56}{3}j^4k'-\frac{4}{3}j^4+\frac{8}{3}j^3k'^3+\frac{40}{3}j^3k'^2-\frac{64}{3}j^3k'+\frac{4}{3}j^3\\
    +4j^2k'^2+2j^2k'^2-\frac{16}{3}j^2k'+\frac{1}{3}j^2+\frac{1}{3}jk'^3+\frac{2}{3}jk'^2-\frac{2}{3}jk'+\frac{56}{15}j+2k'^3-6k'^2+\frac{29}{5}k'-\frac{21}{10}
\end{multline}
and
\begin{equation}
    -\frac{13}{10}k'^6+\frac{21}{5}jk'^5-\frac{1}{30}k'^5+\frac{8}{3}jk'^4+\frac{1}{12}k'^4-\frac{13}{3}jk'^3+2k'^3-\frac{5}{3}jk'^2-\frac{25}{4}k'^2+\frac{86}{15}k'+\frac{58}{15}j-\frac{21}{10}
\end{equation}
respectively. These are analytical expressions still. From here the fractional variance shift can be found using linearity of the trace and thermal density state.

\subsection{Driven Solution}\label{drivensoln}

We will consider the result of:
\begin{equation}
    \text{tr}(J_y e^{it(H+2\Omega\cos(\omega_0 t) J_x)}\rho_{th}e^{-it(H+2\Omega\cos(\omega_0 t) J_x)})
\end{equation}
Note that since $\rho_{\text{th}}\propto e^{-\beta H}$ we may commute the evolution of $H_{\text{int}}$ around it as $H_{\text{int}}$ (without the rotating-wave approximation) commutes with $J_x$. Further the term $a^\dag a$ in $H_0$ commutes with $J_y$, so we may cancel that term as well. Then this reduces to:
\begin{equation}
    \text{tr}(J_y e^{it(\omega_0J_z+2\Omega\cos(\omega_0 t)J_x)}\rho_{\text{th}}e^{-it(\omega_0J_z+2\Omega\cos(\omega_0 t)J_x)}).
\end{equation}
We will now move into the rotating frame defined by $\omega_0 J_z$ for the Hamiltonian $\omega_0 J_z+ 2\Omega\cos(\omega_0 t) J_x$. This becomes $\Omega J_x$ \cite{wood_cavity_2014}. Our expression is now:
\begin{equation}
    \text{tr}(e^{-it\Omega J_x}J_y e^{it\Omega J_x} \rho_{\text{th}}).
\end{equation}
If we only keep the terms with constant number of excitations, we obtain from \cite{wood_cavity_2014} that:
\begin{eqnarray}
    \text{tr}(e^{-it\Omega J_x}J_y e^{it\Omega J_x} \rho_{\text{th}})&=& \text{tr}(\frac{1}{2} (e^{-i\Omega t} (J_y+iJ_z)+e^{i\Omega t}(J_y-iJ_z))\rho_{\text{th}})\\
    &=& \text{tr}(\cos(\Omega t)J_y \rho_{\text{th}})+\text{tr}(\sin(\Omega t)J_z\rho_{\text{th}})\\
    &=& \sin(\Omega t) \text{tr}(J_z \rho_{\text{th}}), \label{behavior}
\end{eqnarray}
where the first trace is zero due to not preserving the number of excitations.

Computing $\text{tr}(\omega_0 J_z\rho_{\text{th}})$ first non-zero change:
\begin{eqnarray}
    &&\text{tr}(J_z L^2(j,k))\\
    &=& (1-j)l_1^2(j,k)+(|\mathcal{B}|-j)l_{|\mathcal{B}|}^2(j,k)+\sum_{\alpha=2}^{|\mathcal{B}|-1}(\alpha-j)(l_\alpha^2(j,k)+l_{\alpha+1}^2(j,k))\\
    &=& (1-j)(2j)(k')+(|\mathcal{B}|-j)(|\mathcal{B}|(2j-(|\mathcal{B}|-1))(k'+1-|\mathcal{B}|)\\
    \nonumber& & +\sum_{\alpha=2}^{|\mathcal{B}|-1}(\alpha-j)(\alpha (2j-\alpha+1)(k'-\alpha+1)+(\alpha+1)(2j-\alpha)(k'-\alpha))\\
    &=& |\mathcal{B}|^4+(-3j-k'-2)|\mathcal{B}|^3+(2j^2+3jk'+4j+k'+1)|\mathcal{B}|^2\\
    \nonumber& & +(-2j^2k'-2j^2-jk'-j)|\mathcal{B}|+(-2j^2k'+2jk')\\
    \nonumber& & +\sum_{\alpha=2}^{|\mathcal{B}|-1} 2\alpha^4+(-1-6j-2k')\alpha^3+(1+4j^2+j+6jk')\alpha^2\\
    \nonumber& &+(-4j^2k'-j+2jk')\alpha+2j^2k'\\
    &=& |\mathcal{B}|^4+(-3j-k'-2)|\mathcal{B}|^3+(2j^2+3jk'+4j+k'+1)|\mathcal{B}|^2\\
    \nonumber& & +(-2j^2k'-2j^2-jk'-j)|\mathcal{B}|+(-2j^2k'+2jk')\\
    \nonumber& & +A\frac{1}{30}(6|\mathcal{B}|^5-15|\mathcal{B}|^4+10|\mathcal{B}|^3-|\mathcal{B}|-30)+B\frac{1}{4}(|\mathcal{B}|^4-2|\mathcal{B}|^3+|\mathcal{B}|^2-4)\\
    \nonumber& & +C\frac{1}{6}(2|\mathcal{B}|^3-3|\mathcal{B}|^2+|\mathcal{B}|-6)+D\frac{1}{2}(|\mathcal{B}|^2-|\mathcal{B}|-2)+E(|\mathcal{B}|-2)\\
    &=& \frac{A}{5}|\mathcal{B}|^5+(1-\frac{A}{2}+\frac{B}{4})|\mathcal{B}|^4+(-3j-k'-2+\frac{A}{3}-\frac{B}{2}+\frac{C}{3})|\mathcal{B}|^3\\
    \nonumber& &+ (2j^2+3jk'+4j+k'+1+\frac{B}{4}-\frac{C}{2}+\frac{D}{2})|\mathcal{B}|^2\\
    \nonumber& &+ (-2j^2k'-2j^2-2jk'-j-\frac{A}{30}+\frac{C}{6}-\frac{D}{2}+E)|\mathcal{B}|\\
    \nonumber& &+(-2j^2k'+2jk'-A-B-C-D-2E)\\
    &=& \frac{2}{5}|\mathcal{B}|^5+\frac{1}{4}(-1-6j-2k')|\mathcal{B}|^4+\frac{1}{6}(8j^2+12jk'+2j-3)|\mathcal{B}|^3\\
    \nonumber& & +\frac{1}{4}(-8j^2k'+4jk'+6j+2k'+1)|\mathcal{B}|^2\\
    \nonumber& &+\frac{1}{30}(60j^2k'-40j^2-60jk'-10j+3)|\mathcal{B}|\\
    \nonumber& & +(-2j^2k'-4j^2-6jk'+6j+2k'-2)
\end{eqnarray}

\if{false}
\clearpage
\subsection{Constant Photon Counts}
We stratify by photon count $\kappa$ and take the conditional result:
\begin{eqnarray}
    |_{\kappa}&=&\sum_{w_H=0}^N e^{-\beta\omega_0 w_H}\sum_{j=N/2}^{N/2-w_H} d_j \sum_{m=-j}^{\min (j,w_H-j)} \langle j,m|(1+\frac{1}{2}H_{int}^2)(| \kappa \rangle\otimes | j,-j+w_H\rangle)(\langle \kappa |\otimes \langle j,-j+w_H|)|j,m\rangle\\
    &=& \sum_{w_H=0}^N e^{-\beta\omega_0 w_H}\sum_{j=N/2}^{\max (N/2-w_H, j_{min})} d_j \langle \kappa |\langle j,-j+w_H-(N/2-j) |(1+\frac{1}{2}H_{int}^2)|\kappa\rangle  | j,-j+w_H-(N/2-j)\rangle\\
    \nonumber &=& \sum_{w_H=0}^N e^{-\beta\omega_0 w_H}\sum_{j=N/2}^{\max (N/2-w_H, j_{min})} d_j \langle \kappa |\langle j,-N/2+w_H |(1+\frac{1}{2}H_{int}^2)|\kappa\rangle  | j,-N/2+w_H)\rangle \Theta(j-(-N/2+w_H))
\end{eqnarray}
First order term is cancelled since $\kappa$ is not constant.

 \begin{figure}[h]
    \centering
    \includegraphics[scale=0.5]{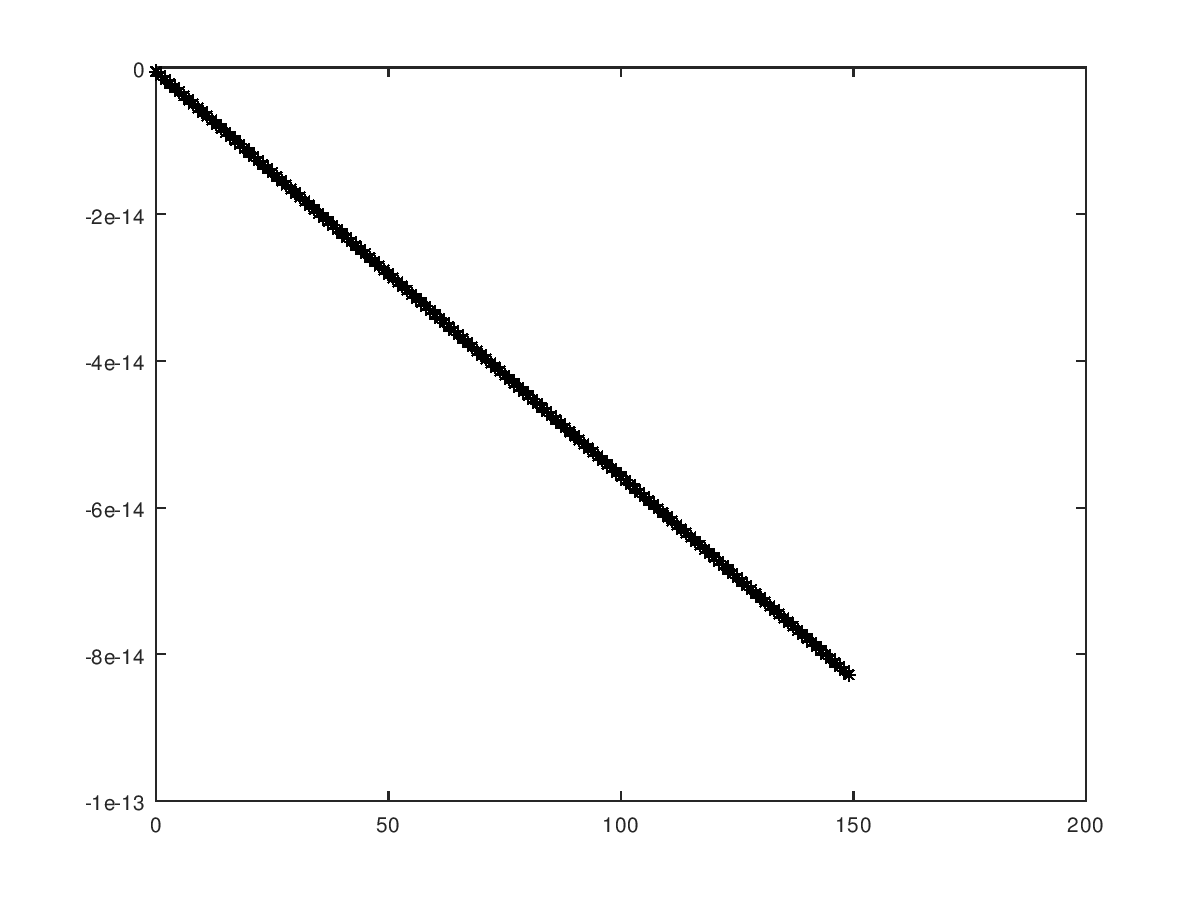}
    \caption{The fractional mean shift: The $x$-axis is the number of photons, while $y$-axis is the Lamb shift induced change in the mean, as a fraction of the unshifted mean. This curve is linear, and manifestly so from the $\langle \kappa |H_{int}^2|\kappa\rangle$ (photon count) term being able to be factored out. $N=100$, $T=0.6$ K. Typical $g_0$ and $\omega_0$ values.}
\end{figure}

\clearpage

\subsection{Observing Coherent States}
aaaaa
\fi
\end{widetext}

\end{document}

%% file: preamble.tex
\usepackage{xcolor}
\usepackage{graphicx}
\usepackage[left=23mm,right=13mm,top=35mm,columnsep=15pt]{geometry} 
\usepackage{adjustbox}
\usepackage{placeins}
\usepackage[T1]{fontenc}
\usepackage{csquotes}
\usepackage[vcentermath]{youngtab}
\usepackage{caption}

\usepackage[nointegrals]{wasysym}
\usepackage{amsmath}
\usepackage{amssymb}
\usepackage{amsthm}
\usepackage{braket}
\usepackage{bm}
\usepackage{mathtools}
\usepackage{tensor}
\usepackage{mathrsfs}
\usepackage[vcentermath]{youngtab}
\usepackage{subcaption}

\def\openone{\leavevmode\hbox{\small$1$\normalsize\kern-.33em$1$}}

\newcommand{\abs}[1]{\left| #1 \right|} 
\newcommand{\avg}[1]{\left< #1 \right>}

\newcommand{\matrixel}[3]{\left< #1 \vphantom{#2#3} \right|
 #2 \left| #3 \vphantom{#1#2} \right>} 
 
\let\baraccent=\= 
\renewcommand{\=}[1]{\stackrel{#1}{=}}

\newcommand{\ketbra}[2]{\left| #1 \vphantom{#2}\right>\!\!\left< #2\vphantom{#1}\right|}
\newcommand{\be}{\begin{equation}}
\newcommand{\bel}[1]{\begin{equation}\label{#1}}
\newcommand{\ee}{\end{equation}}

\newcommand{\dd}{\mathrm{d}}
\newcommand{\C}{\mathbb{C}}

\renewcommand{\H}{\hat{\mathcal{H}}}

\DeclarePairedDelimiter\floor{\lfloor}{\rfloor}
\graphicspath{ {./images/} }